\documentclass[11pt]{article}

\usepackage{amsfonts,latexsym,amsthm,amssymb,amsmath,amscd,euscript,tikz,mathtools}
\usepackage{framed}
\usepackage{fullpage}
\usepackage{color}
\usepackage[colorlinks=true,citecolor=blue,linkcolor=blue]{hyperref}
\usepackage{enumitem}
\usepackage{siunitx}
\usepackage{textcomp}
\usepackage{physics}
\usepackage{mathrsfs}
\usepackage{dsfont}
\usepackage[ruled,vlined]{algorithm2e}
\usepackage{titlesec}
\usepackage{tikz-cd}
\usepackage{thmtools}
\usepackage{thm-restate}
\usepackage{cleveref}
\usepackage{graphicx}
\usetikzlibrary{positioning,chains,fit,shapes,calc}

\allowdisplaybreaks[1]

\newtheorem{theorem}{Theorem}
\newtheorem{proposition}[theorem]{Proposition}
\newtheorem{lemma}[theorem]{Lemma}
\newtheorem{claim}[theorem]{Claim}
\newtheorem{corollary}[theorem]{Corollary}

\theoremstyle{definition}
\newtheorem{definition}[theorem]{Definition}

\newtheorem{remark}[theorem]{Remark}

\theoremstyle{remark}

\newcommand{\cA}{\mathcal{A}}\newcommand{\cB}{\mathcal{B}}
\newcommand{\cC}{\mathcal{C}}
\newcommand{\cE}{\mathcal{E}}

\newcommand{\cK}{\mathcal{K}}\newcommand{\cL}{\mathcal{L}}
\newcommand{\cM}{\mathcal{M}}
\newcommand{\cP}{\mathcal{P}}

\newcommand{\bC}{\mathbb{C}}
\newcommand{\bE}{\mathbb{E}}\newcommand{\bF}{\mathbb{F}}

\newcommand{\bN}{\mathbb{N}}

\newcommand{\bZ}{\mathbb{Z}}

\newcommand{\poly}{\operatorname{poly}}

\newcommand{\Enc}{\operatorname{Enc}}
\newcommand{\Dec}{\operatorname{Dec}}
\newcommand{\ListDec}{\operatorname{ListDec}}
\newcommand{\Rec}{\operatorname{Rec}}
\newcommand{\evl}{\operatorname{ev}}
\newcommand{\fevl}{\widetilde{\operatorname{ev}}}
\newcommand{\CSS}{\operatorname{CSS}}
\newcommand{\dis}{\operatorname{dis}}
\newcommand{\foldi}[1]{F_{#1}}
  
\newcommand{\nc}{\newcommand}

\nc{\on}{\operatorname}
\nc{\Spec}{\on{Spec}}
\nc{\Aut}{\textit{Aut}}
\nc{\id}{\textit{id}}
\nc{\chr}{\on{char}}
\nc{\im}{\on{im}}
\nc{\Hom}{\on{Hom}}
\nc{\lcm}{\on{lcm}}
\nc{\dual}[1]{\prescript{t}{}{#1}}
\nc{\transpose}[1]{{#1}^{\intercal}}
\nc{\Sym}{\on{Sym}}
\nc{\End}{\on{End}}
\nc{\stab}{\on{stab}}
\nc{\Li}{\on{Li}}
\nc{\spn}{\on{span}}
\nc{\sgn}{\on{sgn}}
\nc{\supp}{\on{supp}}
\nc{\Unif}{\on{Unif}}

\DeclareMathOperator*{\argmax}{arg\,max}
\DeclareMathOperator*{\argmin}{arg\,min}

\title{Quantum Locally Recoverable Codes\thanks{Research supported in part by a Simons Investigator award, and a UC Noyce initiative award. L.~Golowich is supported by a National Science Foundation Graduate Research Fellowship under Grant No.~DGE 2146752.}
}
\author{Louis Golowich \\
  UC Berkeley \\
  \href{mailto:lgolowich@berkeley.edu}{\texttt{lgolowich@berkeley.edu}}
  \and
  Venkatesan Guruswami \\
  UC Berkeley \\
  \href{mailto:venkatg@berkeley.edu}{\texttt{venkatg@berkeley.edu}}
}

\begin{document}

\pagenumbering{gobble}

\maketitle


\begin{abstract}
  Classical locally recoverable codes, which permit highly efficient recovery from localized errors as well as global recovery from larger errors, provide some of the most useful codes for distributed data storage in practice. In this paper, we initiate the study of quantum locally recoverable codes (qLRCs). In the long term, like their classical counterparts, such qLRCs may be used for large-scale quantum data storage. Furthermore, our results have concrete implications for quantum LDPC codes, which are widely applicable to near-term quantum error-correction, as local recoverability is a weakening of the LDPC property.

  After defining quantum local recoverability, we provide an explicit construction of qLRCs based on the classical LRCs of Tamo and Barg (2014), which we show have (1) a close-to-optimal rate-distance tradeoff (i.e.~near the Singleton bound), (2) an efficient decoder, and (3) permit good spatial locality in a physical implementation. The analysis for both the distance and the efficient decoding of these quantum Tamo-Barg (qTB) codes is significantly more involved than in the classical case. Nevertheless, we obtain close-to-optimal parameters by introducing a ``folded'' version of these qTB codes, which we then analyze using a combination of algebraic techniques. We furthermore present and analyze two additional constructions using more basic techniques, namely random qLRCs, and qLRCs from AEL distance amplification. Each of these constructions has some advantages, but neither achieves all 3 properties of our folded qTB codes described above.

  We complement these constructions with Singleton-like bounds that show our qLRC constructions achieve close-to-optimal parameters. We also apply these results to obtain Singleton-like bounds for qLDPC codes, which to the best of our knowledge are novel. We then show that even the weakest form of a stronger locality property called local correctability, which permits more robust local recovery and is achieved by certain classical codes, is impossible quantumly.
\end{abstract}

\newpage

\tableofcontents

\newpage

\pagenumbering{arabic}


\section{Introduction}
\label{sec:intro}
Classical locally recoverable codes (LRCs) provide one of the most important coding theoretic tools for distributed data storage. Such codes are defined to permit highly efficient recovery from common localized errors, as well as larger-scale recovery from rarer but more ``catastrophic'' global errors.

In this paper, we initiate the study of quantum locally recoverable codes (qLRCs). In particular, we define qLRCs, present and analyze constructions, and also prove fundamental limitations on the achievable parameters and properties. While our constructions can be viewed as quantum generalizations of classical constructions, the analysis becomes surprisingly intricate, and requires new ideas that were not needed classically. Our results may also shed light on the study of locality in quantum coding theory, for instance as it pertains to quantum LDPC codes.

Classically, the properties of a LRC are well suited for the needs of large datacenters, which can cost billions of dollars to build and maintain, and must often account for localized server failures while also handling occasional more global failures. Indeed, companies such as Microsoft \cite{huang_erasure_2012} and Facebook \cite{muralidhar_f4_2014} have implemented LRCs to obtain improved performance for data storage.

Currently, experimental quantum computers remain at a vastly smaller scale than that of the classical datacenters in which LRCs are often used in practice. However, it is not implausible that quantum computing technology eventually follows its classical counterpart by growing to the scale where codes such as qLRCs become an integral part of quantum data storage.

Furthermore, our study of qLRCs reveals the potential for more broad and near-term implications as well. Indeed, locality properties in quantum codes, such as the ability to decode using local measurements (i.e.~LDPC codes), are of particular importance for quantum error correction. Yet such locality is notoriously difficult to achieve in the quantum setting. Indeed, the first linear-distance quantum LDPC (qLDPC) codes were only recently constructed \cite{panteleev_asymptotically_2022,leverrier_quantum_2022-1,dinur_good_2023}, and good quantum codes with stronger properties such as local testability have yet to be constructed. This difficulty of achieving locality is in contrast to the classical setting, where good LDPC codes have been known for decades (e.g.~\cite{sipser_expander_1996}), good locally testable codes were recently constructed \cite{panteleev_asymptotically_2022,dinur_locally_2022}, and other strong locality properties such as local correctability exist in linear-distance, albeit low-rate codes.

From this perspective, our study of qLRCs provides a new angle to investigate locality properties in quantum codes. Indeed, classical local recoverability requires each code component to participate in one low-weight parity check, while quantum local recoverability requires each code component to participate in two low-weight stabilizers. Thus local recoverability can be viewed as a weaking of the LDPC property, in which each code component participates in many low-weight parity checks/stabilizers. Our study of qLRCs can therefore be viewed as progress towards understanding stronger locality properties possessed by qLDPC codes. One concrete example of this connection is provided in Section~\ref{sec:singleton}, where we show that qLRCs, and therefore also qLDPC codes, of constant locality $r=O(1)$ must have relative distance bounded away from $1/2$; to the best of our knowledge such a bound for qLDPC codes had not been previously shown.

\subsection{Our Contributions}
In this section we present the contributions of our paper. For details on notation or basic definitions, the reader is referred to Section~\ref{sec:prelim}.

\subsubsection{Definition of Quantum Local Recoverability}
\label{sec:qlrcdefinf}
To begin, we define qLRCs. Recall that a classical LRC (cLRC) is a classical code $C$ such that for every $c\in C$ and every component $i$, the value of $c_i$ can be recovered by looking at the restriction of $c$ to just $r-1$ other components.

\begin{definition}[Informal statement of Definition~\ref{def:qLRC}]
  A \textbf{quantum locally recoverable code (qLRC) of locality $r$} is a quantum code $\cC$ such that if any single qudit of a code state $\ket{\psi}\in\cC$ is erased (i.e.~it experiences a completely depolarizing channel), the original code state $\ket{\psi}$ can be recovered by applying a recovery channel that accesses only $r-1$ other code state qudits.
\end{definition}

For intuition, recall that a classical linear code is a cLRC of locality $r$ if each component takes part in a parity-check of weight $\leq r$. Similarly, we show that a quantum CSS code $\cC=\CSS(C_X,C_Z)$ is a qLRC if both $C_X,C_Z$ are cLRCs, so that every qudit takes part in a low weight $X$-parity-check and a low-weight $Z$-parity-check (see Corollary~\ref{cor:CSSLRC}).

We remark that classically, there is also a notion of message locally recoverable codes (mLRCs), which require that every message (instead of codeword) symbol can be recovered from $r-1$ codeword symbols. As any linear classical code has a systematic encoding, meaning that the first $k$ codeword symbols equal the message, classical mLRCs are strictly weaker than LRCs. However, the local indistinguishability property of quantum codes (see Lemma~\ref{lem:localind}) implies that local queries to quantum codes cannot reveal anything about the message. Thus mLRCs do not exist quantumly, at least in the regime where the locality is less than the distance.

\subsubsection{Explicit Construction of qLRCs}
\label{sec:qTBinf}
One of our principal technical contributions is the following explicit construction of qLRCs.

\begin{theorem}[Folded quantum Tamo-Barg codes; informal statement of Corollary~\ref{cor:fqTBdisnice} combined with Lemma~\ref{lem:qTBdim}]
  \label{thm:fqTBdisinf}
  For every prime number $r$ and every $0<R<1$, there exists an infinite explicit family of qLRCs of locality $r$, rate $\geq R$, relative distance
  \begin{equation*}
    \delta \geq \frac{1-R}{2}-O\left(\frac{1}{\sqrt{r}}\right),
  \end{equation*}
  and alphabet size $n^{O(r^2)}$, where $n$ denotes the block length.
\end{theorem}

In Theorem~\ref{thm:fqTBdisinf} (and in future informal result statements in Section~\ref{sec:intro}), for readability we state slightly looser bounds than the formal result statements. For instance, our actual distance bound in Corollary~\ref{cor:fqTBdisnice} is stronger than than stated in Theorem~\ref{thm:fqTBdisinf} for high rates, and in particular shows that $\delta\geq\Omega(1/r)$ for all $R\leq 1-10/r$.

We remark that while the alphabet size $n^{O(r^2)}$ may seem large, in the LRC literature one typically thinks of each code component as being a fairly large entity, so such a polynomial alphabet size is not unreasonable. Classically each code component could for instance be a hard drive, while quantumly each component would likely be itself a fault-tolerant quantum memeory.

We prove Theorem~\ref{thm:fqTBdisinf} by introducing a quantum CSS version of the classical LRCs of Tamo and Barg \cite{tamo_family_2014}, which achieve the optimal classical rate-distance-locality tradeoff \cite{gopalan_locality_2012}. The classical Tamo-Barg (TB) codes are constructed as subcodes of Reed-Solomon codes by carefully inserting low-weight parity checks.

Multiple complications arise when converting TB codes into CSS codes. Specifically, to ensure the CSS orthogonality relations are satisfied, we define a quantum Tamo-Barg (qTB) code to be a CSS code $\cC=\CSS(C,C)$ consisting of two copies of a classical code $C$ that contains a TB code as a subcode, but also contains some added low-weight codewords. These added low-weight codewords also lie in $C^\perp$, so they do not necessarily degrade the distance of $\cC$, which equals the minimum weight of an element of $C\setminus C^\perp$. However, these added low-weight codewords make the distance analysis significantly more challenging, and we are only able to show a ``Johnson-like'' bound on the distance of qTB codes:

\begin{theorem}[Quantum Tamo-Barg codes; informal statement of Theorem~\ref{thm:qTBdis}]
  \label{thm:qTBdisinf}
  For every prime number $r$ and every $0<R<1$, there exists an infinite explicit family of qLRCs of locality $r$, rate $\geq R$, relative distance
  \begin{equation*}
    \delta \geq 1-\sqrt{\frac{1+R}{2}}-O\left(\frac{1}{r}\right),
  \end{equation*}
  and alphabet size $n+1$, where $n$ denotes the block length.
\end{theorem}

To obtain the improved bound $\delta\geq 1-(1+R)/2-O(1/\sqrt{r})$ in Theorem~\ref{thm:fqTBdisinf}, we ``fold'' together code components, thereby reblocking the symbols into larger components. We then use a combination of algebraic techniques to bound the distance of these folded qTB (fqTB) codes; the two main tools are a root detection method involving a determinant polynomial, and an uncertainty principle over finite fields.

In Section~\ref{sec:explicitoverviewinf} below, we describe the formal construction of our (f)qTB codes, and outline the proofs of Theorem~\ref{thm:qTBdisinf} and Theorem~\ref{thm:fqTBdisinf}. A more detailed description of these codes is provided in Section~\ref{sec:construct}, and the distance bounds are formally proven in Section~\ref{sec:distance}.



\subsubsection{Singleton-Like Bound for qLRCs, with Implications for qLDCP Codes}
\label{sec:singletoninf}
Our fqTB codes in Theorem~\ref{thm:fqTBdisinf} achieve an optimal rate-distance tradeoff as the locality $r$ grows large. In particular, we prove the following fundamental limitation on any qLRC.

\begin{theorem}[Singleton-like bound; informal statement of Theorem~\ref{thm:singletongen}]
  \label{thm:singletongeninf}
  If $\cC$ is a qLRC of locality $r$ and rate $R$, then $\cC$ has relative distance
  \begin{equation*}
    \delta\leq\frac{1-R}{2}-\Omega\left(\frac{1}{r}\right).
  \end{equation*}
\end{theorem}

The reader is referred to Theorem~\ref{thm:singletongen} for the specific value of the constant hidden in the the $\Omega(1/r)$ term above.

For comparison, recall that the ordinary quantum Singleton bound states that every quantum code of rate $R$ has relative distance $\delta\leq(1-R)/2+O(1/n)$, where $n$ denotes the block length. Thus for a fixed rate, Theorem~\ref{thm:singletongeninf} shows that imposing local recoverability with locality $r$ decreases the optimal relative distance of a quantum code by at least $\Omega(1/r)$. Meanwhile, our explicit construction in Theorem~\ref{cor:fqTBdisnice} has relative distance $O(1/\sqrt{r})$ below that of the Singleton bound.

Theorem~\ref{thm:singletongeninf} implies a fundamental difference between the classical and quantum cases for LRCs of low rate. Classically, the ordinary Singleton bound says that a code of rate $R$ has relative distance $\delta\leq 1-R+O(1/n)$, while for cLRCs of locality $r$ this bound becomes $\delta\leq 1-R\cdot r/(r-1)+O(1/n)$. The TB codes achieve this latter bound, proving its tightness. Therefore in particular, for fixed locality $r$, by letting the rate $R\rightarrow 0$, we see that there exist classical LRCs of relative distance $\delta\rightarrow 1$. In contrast, whereas there exist quantum codes of relative distance approaching $1/2$, Theorem~\ref{thm:singletongeninf} shows every qLRC has relative distance at most $1/2-\Omega(1/r)$, which is bounded away from $1/2$ for fixed $r$.

As we mentioned previously, because every (q)LRC is a by definition a (q)LDPC code, it follows that qLDPC codes of locality $r$ have relative distance $\leq 1/2-\Omega(1/r)$ for arbitrarily large alphabets. To the best of our knowledge, such a bound has not been previously shown in the literature. This result is in again in contrast to the classical case, where for instance the $q$-ary Hadamard code is an LDPC code of locality $3$ with relative distance approaching $1$ as $q$ grows large.

The discussion above raises the interesting question as to how tight Theorem~\ref{thm:singletongeninf} is for qLDPC codes. That is, what is the additional cost to the optimal rate-distance tradeoff of requiring a quantum code be LDPC, compared to just being an LRC?

\subsubsection{Comparison to Basic Constructions}
\label{sec:basicinf}
While our explicit fqTB codes in Theorem~\ref{thm:fqTBdisinf} are $O(1/\sqrt{r})$ below the Singleton-like bound in Theorem~\ref{thm:singletongeninf}, we show that a randomized construction improves this gap to $O(1/r)$, at the cost of explicitness and efficiency.

\begin{proposition}[Random qLRCs; informal statement of Proposition~\ref{prop:random}]
  \label{prop:randominf}
  For every $r\geq 3$ and $\delta>0$, there exists a randomized construction that with high probability gives a qLRC of locality $r$, rate $R$, relative distance
  \begin{equation*}
    \delta \geq \frac{1-R}{2}-O\left(\frac{1}{r}\right),
  \end{equation*}
  and alphabet size $2^{O(r)}$.
\end{proposition}

The constant hidden in the $O(1/r)$ term in Proposition~\ref{prop:randominf} is larger than the constant in the $\Omega(1/r)$ term in Theorem~\ref{thm:singletongeninf}, so our bound on the randomized construction is still $O(1/r)$ below that of our Singleton-like bound.

Thus the randomized construction in Proposition~\ref{prop:random} achieves relative distance $O(1/r)$ below our Singleton-like bound with has alphabet size $2^{O(r)}$. These parameters improve upon our fqTB construction, which has relative distance $O(1/\sqrt{r})$ below the Singleton-like bound, and has alphabet size $n^{O(r^2)}$.

However, the main disadvantage of the randomized construction is its non-explicitness. As a result, we have no efficient algorithm to certify that a randomly sampled qLRC has good distance, and we have no efficient decoding algorithm for errors in unknown locations. In contrast, our fqTB codes are explicit, so their distance bound is guaranteed. Furthermore, we show that they have an efficient decoding algorithm (see Section~\ref{sec:decodeinf} below).

One way of derandomizing the random qLRCs in Proposition~\ref{prop:randominf} is to use them as in inner code in the concatenation and distance amplification scheme of Alon, Edmunds, and Luby (AEL) \cite{alon_linear_1995}. This technique has been used extensively in classical coding theory \cite{guruswami_expander-based_2001,guruswami_near-optimal_2002,guruswami_linear_2003,guruswami_explicit_2008,hemenway_linear-time_2018,kopparty_high-rate_2016,gopi_locally_2018,hemenway_local_2020}, but has only recently been considered in the quantum setting \cite{bergamaschi_approaching_2022,wills_tradeoff_2023}. We show that applying AEL using a random qLRC as an ``inner code,'' which is small enough to be found efficiently via brute force, yields the following result.

\begin{proposition}[qLRCs from AEL; informal statement of Proposition~\ref{prop:qLRCAEL}]
  \label{prop:qLRCAELinf}
  For every fixed $0<R<1$, it holds for all sufficiently large $r\in\bN$ that there exists an infinite family of efficiently constructable qLRCs of locality $r$, rate $R$, relative distance
  \begin{equation*}
    \delta \geq \frac{1-R}{2}-O_R\left(\frac{1}{r^{1/4}}\right),
  \end{equation*}
  and alphabet size $q=2^{O(r)}$, where the $O_R$ above hides a constant depending on $R$.
\end{proposition}

As described in Remark~\ref{rem:effconst}, the codes in Proposition~\ref{prop:qLRCAELinf} are technically only efficiently constructable by a randomized algorithm with high probability, but can be made truly explicit by using a slightly more complicated construction, with slightly worse parameters. Also, like our fqTB codes, we show that these codes from AEL have efficient decoders from errors in unknown locations.

The alphabet size $q=2^{O(r)}$ in Proposition~\ref{prop:qLRCAELinf} is smaller than the $q=n^{O(r^2)}$ of our fqTB codes in Theorem~\ref{thm:fqTBdisinf}. However, qLRCs from AEL have worse rate-distance-locality tradeoff, as their relative distance is $O_R(1/r^{1/4})$ below the Singleton-like bound, compared to only $O(1/\sqrt{r})$ for our fqTB codes.

Our (f)qTB codes may have additional practical advantages over the qLRCs from AEL. For instance, our (f)qTB codes can achieve locality as small as $r=3$, whereas the the minimum possible locality in the qLRCs from AEL is $r=9$. Furthermore, as one step in the AEL construction redistributes code symbols according to the edges of an expander graph, the resulting size-$r$ sets of code components used for local recovery form an $r$-uniform hypergraph with a complex expanding structure. In contrast, the recovery sets of our (f)qTB codes form a partition of the code components, which is ideally suited for a physical implementation with good spatial locality. That is, our (f)qTB codes can easily be implemented in 1, 2, or 3-dimensional space such that each local recovery operation only involves code components that are close together; such a spatially local implementation would be much less feasible for a qLRC from AEL.

\subsubsection{Efficient Decoding Algorithms}
\label{sec:decodeinf}
This section presents our results on the decodability of our qLRCs. As we discussed previously, in the classical setting, LRCs are typically used for data storage where errors correspond to events such as server failures that are detectable. Such errors occur in known locations, so they can be treated as erasures. Because every linear code can be efficiently decoded from a number of erasures up to the distance using Gaussian elimination, the efficiency of decoding is often not a primary concern for cLRCs. Note that all of our qLRC constructions are stabilizer (and in fact CSS) codes, which can similarly be decoded efficiently from erasures.

However, in the quantum setting there are additional potential applications of efficient decoding from errors in unknown locations. Due to the inherently noisy nature of quantum states, even if a qLRC is used to store a large quantum state where each code component is itself stored in a fault-tolerant memory, these individual fault-tolerant memories may eventually accumulate errors. As such, it may be beneficial to occasionally perform a global decoding procedure to reduce the overall error rate in the long term.

Furthermore, as we have previously discussed, qLRCs can be viewed as a stepping stone towards better understanding stronger locality properties, such as the LDPC property, which are important for near-term quantum error correction. As efficient decoding is critical for these error correction applications, it is desirable that the qLRCs we study are also efficiently decodable.

Below we present a polynomial-time decoding algorithm for the fqTB codes in Theorem~\ref{thm:fqTBdisinf}.

\begin{theorem}[Decoding fqTB codes; informal statement of Corollary~\ref{cor:fqTBdecradiusnice}]
  \label{thm:fqTBdecinf}
  The fqTB codes in Theorem~\ref{thm:fqTBdisinf} of block length $n$, prime locality parameter $r$, and rate $R$ can be decoded from errors acting on an unknown
  \begin{equation*}
    \frac12\left(\frac{1-R}{2}-O\left(\frac{1}{\sqrt{r}}\right)\right)
  \end{equation*}
  fraction of the code components in time $n^{O_r(1)}$, provided the alphabet size is increased to some sufficiently large $q=n^{O_r(1)}$ with respect to $r$. Here $O_r(1)$ denotes a sufficiently large constant depending only on $r$.
\end{theorem}

Theorem~\ref{thm:fqTBdecinf} provides a decoding algorithm for the fqTB codes with decoding radius up to half our distance bound in Theorem~\ref{thm:fqTBdisinf}, which is therefore optimal barring an improvement in the distance bound. While the algorithm runs in polynomial time $n^{O_r(1)}$ for fixed locality parameter $r$, this algorithm is inefficient for growing $r$. We address this issue by providing the following decoding algorithm for unfolded qTB codes, which therefore also applies to folded qTB codes, and whose running time is a polynomial independent of $r$.

\begin{theorem}[Decoding (f)qTB codes; informal statement of Theorem~\ref{thm:qTBdec}]
  \label{thm:qTBdecinf}
  An (unfolded or folded) qTB code of block length $n$, prime locality parameter $r$, and rate $R$ can be decoded from errors acting on an unknown
  \begin{equation*}
    \frac12\left(1-\sqrt{\frac{1+R}{2}}-O\left(\frac{1}{r}\right)\right)
  \end{equation*}
  fraction of the code components in time $n^{O(1)}$.
\end{theorem}

The decoder in Theorem~\ref{thm:qTBdecinf} simply performs $r-1$ calls to a classical Reed-Solomon (list) decoder, and then performs some postprocessing on the resulting outputs (see Algorithm~\ref{alg:qTBdec}). Therefore this algorithm should be efficient in practice, as Reed-Solomon decoders have been optimized for practical use.

The decoding radius in Theorem~\ref{thm:qTBdecinf} is half our distance bound in Theorem~\ref{thm:qTBdisinf}. Therefore this decoding radius is optimal among decoders for qTB codes, barring an improvement to our distance bound in Theorem~\ref{thm:qTBdisinf}. As
\begin{equation*}
  \frac{1-R}{4} \leq 1-\sqrt{\frac{1+R}{2}} \leq \frac{1-R}{2},
\end{equation*}
the unfolded qTB codes achieve distance and decoding radius within roughly a factor of $2$ of the optimal values as dictated by our Singleton-like bound in Theorem~\ref{thm:singletongeninf}.

We prove both Theorem~\ref{thm:qTBdecinf} and Theorem~\ref{thm:fqTBdecinf} using similar techniques as in the proof of Theorem~\ref{thm:qTBdisinf}. Specifically, let $\cC=\CSS(C,C)$ denote the qTB code. The decoding task for $\cC$ can be reduced to the following problem: given a corrupted codeword $a$ of $C$, we want an efficient algorithm the recovers some $c\in C$ that is close to $a$ in Hamming weight. For this purpose, we show how every $c\in C$ can be modified to obtain some $c'$ that is a codeword of a Reed-Solomon code. Applying the same modification to $a$, we obtain a corrupted Reed-Solomon codeword $a'$. We then apply a Reed-Solomon (list) decoder to recover $c'$, which we can then map back to the desired codeword $c$.

To decode fqTB codes in Theorem~\ref{thm:fqTBdecinf}, we use the same algorithm, except now we apply a folded Reed-Solomon list decoder \cite{guruswami_explicit_2008} to recover $c'$ from $a'$. As folded Reed-Solomon codes have a larger list-decoding radius than ordinary Reed-Solomon codes, and our fqTB distance bound in Theorem~\ref{thm:fqTBdisinf} is better than our qTB distance bound in Theorem~\ref{thm:qTBdisinf}, we obtain a larger decoding radius in Theorem~\ref{thm:fqTBdecinf} than in Theorem~\ref{thm:qTBdecinf}.

\subsubsection{Impossibility of Quantum Locally Correctable Codes}
Our results on qLRCs described above indicate that while optimal local recoverability is more nuanced and difficult to achieve quantumly than classically, there do exist constructions of qLRCs approaching the optimal parameters. It is therefore natural to consider quantum analogues of stronger forms of locality that exist classically.



Local correctability and local decodability provide particularly notable examples of such properties. Recall that an LRC of block length $n$ and locality $r$ has the property that any single code component $i\in[n]$ is erased, the value of a codeword at that component can be recovered from $\leq r-1$ unerased components. A \textit{locally correctable code (LCC)} has the stronger property that after a linear number $\Omega(n)$ of code components are erased, the value of a codeword at each erased component $i$ can be recovered from $\leq r-1$ unerased components.

Equivalently, an LRC requires the value of a codeword at each component $i\in[n]$ to be recoverable from some recovery set of $\leq r-1$ other components. In contrast, an LCC requires each component $i\in[n]$ to have $\Omega(n)$ disjoint recovery sets, each of which contains $\leq r-1$ components that can be used to recover the codeword's value at $i$.

\textit{Local decodable codes (LDCs)} are defined similarly, except they only need to support local recovery for message components, rather than for codeword components.

The Hadamard and Reed-Muller codes provide examples of classical LCCs and LDCs. It is therefore natural to ask whether there are quantum versions of these types of codes.

As local access to a quantum code of large distance cannot reveal any information about the message, local decodability seems to make little sense quantumly, at least in the regime of large distance and small locality. However, quantum local correctability may seem to be a reasonable strengthening of quantum local recoverability.

We show below that quantum local correctability is impossible in a strong sense: if a quantum code has even two disjoint recovery sets for a single qudit, then that qudit is unentangled with the remainder of the code state, and contains no information about the message state. Recall that LCCs are typically required to have linearly many such disjoint recovery sets for each code component. Thus qLRCs, which have a single recovery set for each qudit, are in some sense the limit of what is possible for quantum local correctability.

\begin{theorem}[Impossibility of quantum LCCs; informal statement of Theorem~\ref{thm:notworec}]
  \label{thm:notworecinf}
  Let $\cC$ be a quantum code of block length $n$ such that for some qudit $i\in[n]$, there exist subsets $I_i^1,I_i^2\subseteq[n]$ satisfying $I_i^1\cap I_i^2=\{i\}$ such that the following holds for each $b=1,2$: if qudit $i$ of a code state $\psi\in\cC$ is erased (i.e.~it experiences a completely depolarizing channel), the original code state $\psi$ can be recovered by applying a recovery channel that only accesses qudits in $I_i^b$. Then there exists a 1-qudit state $\alpha$ such that every $\psi\in\cC$ can be decomposed as $\psi=\alpha_i\otimes\psi_{[n]\setminus\{i\}}$.
\end{theorem}

While we prove Theorem~\ref{thm:notworecinf} for general quantum codes $\cC\subseteq(\bC^d)^{\otimes n}$, it is illustrative to consider CSS or stabilizer codes. Recall from Section~\ref{sec:qlrcdefinf} that for a CSS code $\cC=\CSS(C_X,C_Z)$, a set $I_i\ni i$ can be used to recover qudit $i$ if both $C_X$ and $C_Z$ have parity checks whose support contains $i$ and lies inside $I_i$. Thus if $I_i^1,I_i^2$ are both recovery sets for component $i$ with $I_i^1\cap I_i^2=\{i\}$, then there are parity checks $c_X'\in C_X^\perp$ and $c_Z'\in C_Z^\perp$ such that $i\in\supp(c_X')\subseteq I_i^1$ and $i\in\supp(c_Z')\subseteq I_i^2$. But then $\supp(c_X')\cap\supp(c_Z')=\{i\}$, so $c_X'\cdot c_Z'=(c_X')_i(c_Z')_i\neq 0$, which contradicts the orthogonality condition $C_X^\perp\subseteq C_Z$ required for $\cC$ to be a well defined CSS code. Thus Theorem~\ref{thm:notworecinf} holds for CSS codes. A similar proof holds for stabilizer codes; In Theorem~\ref{thm:notworec} we prove the more general result for arbitrary quantum codes.

\subsection{Overview of (Folded) Quantum Tamo-Barg codes}
\label{sec:explicitoverviewinf}
In this section, we provide more details on our (folded) quantum Tamo-Barg codes described in Section~\ref{sec:qTBinf}, and we overview the techniques we use to prove the distance bounds in Theorem~\ref{thm:fqTBdisinf} and Theorem~\ref{thm:qTBdisinf}. The construction and analysis of these codes comprise one of the main technical contributions of our paper.

\subsubsection{Background and Classical Tamo-Barg Codes}
Our qTB codes are CSS codes whose associated classical codes are polynomial evaluation codes that are closely related to the classical LRCs of Tamo and Barg \cite{tamo_family_2014}. To describe these codes, we use the following notation. For a subset $S\subseteq\bZ_{\geq 0}$, let $\bF_q[X]^S=\{\sum_{i\in S}a_iX^i:a\in\bF_q^S\}$ denote the space of polynomials for which only monomials $X^i$ for $i\in S$ can have nonzero coefficients. Then let $\evl:\bF_q[X]^S\rightarrow\bF_q^{\bF_q^*}\cong\bF_q^{q-1}$ denote the evaluation map on nonzero points in $\bF_q$, so that $\evl(f)=(f(x))_{x\in\bF_q^*}$

With this notation, a classical Tamo-Barg (TB) code \cite{tamo_family_2014} specified by a prime power $q$, a locality parameter $r|(q-1)$, and an integer $\ell\in[q]$ is given by $\evl(\bF_q[X]^S)$ for
\begin{equation*}
  S = \{i\in[\ell]:i\not\equiv r-1\pmod{r}\}.
\end{equation*}
\cite{tamo_family_2014} showed that this code has alphabet size $q$, block length $q-1$, dimension $\ell-\lfloor l/r\rfloor$, distance $\geq q-\ell$, and is locally recoverable with locality $r$.

\subsubsection{Quantum Tamo-Barg codes}
\label{sec:qTBoverviewinf}
In this section, we define our quantum analogue of TB codes, and give an overview of the analysis. For more details, the reader is referred to Section~\ref{sec:qTBcode} and Section~\ref{sec:qTBdis}.

As described in Section~\ref{sec:qTBinf}, to define a qTB code, we must first modify the classical TB code by adding some low-weight codewords in order to construct a well-defined associated quantum CSS code. These added low-weight codewords can be interpreted as \textit{piecewise linear functions}, as described below.

\begin{definition}[Quantum Tamo-Barg codes; restatement of Definition~\ref{def:qTBcode}]
  \label{def:qTBcodeinf}
  For a prime power $q$, a locality parameter $r|(q-1)$ with $r\geq 3$, and an integer $\ell\in[q]$, we define the \textbf{quantum Tamo-Barg (qTB) code} to be the CSS code $\cC=\CSS(C,C)$ with $C=\evl(\bF_q[X]^S)$ for
  \begin{equation}
    \label{eq:qTBsupportinf}
    S = \{i\in[\ell]:i\not\equiv r-1\pmod{r}\} \cup \{i\in[q-1]:i\equiv 1\pmod{r}\}.
  \end{equation}
\end{definition}

By some basic algebraic manipulations (see Lemma~\ref{lem:qTBdef}), we show that $C^\perp\subseteq C$, so the qTB code $\cC=\CSS(C,C)$ is indeed a well-defined quantum CSS code. Note that by construction $\cC$ has alphabet size $q$ and block length $q-1$. A straightforward calculation (see Lemma~\ref{lem:qTBdim}) shows that $\cC$ has dimension
\begin{align*}
  k
  &= 1+|\{q-\ell\leq i\leq\ell-1:i\not\equiv\pm 1\pmod{r}\}| \\
  &= (2\ell-q)\cdot\left(1-\frac{2}{r}\right)+\epsilon
\end{align*}
for some $\epsilon\in[-2,2]$.

The code $C$ in the definition above is by definition the sum of a classical TB code $\evl(\bF_q[X]^{[\ell]\setminus(-1+r\bZ)})$ with the code\footnote{In Section~\ref{sec:fqTBcode} we denote the code $\evl(\bF_q[X]^{[q-1]\cap(1+r\bZ)})$ by $B^\perp$, but for simplicity of notation in this section we denote it by $P$.} $P:=\evl(\bF_q[X]^{[q-1]\cap(1+r\bZ)})$. We call this latter code $P$ the space of \textbf{piecewise linear functions} due to the following lemma.

\begin{lemma}[Restatement of Lemma~\ref{lem:piecelin}]
  \label{lem:piecelininf}
  Let $\Omega_r=\{x\in\bF_q^*:x^r=1\}$ denote the $r$th roots of unity. Then $P=\evl(\bF_q[X]^{[q-1]\cap(1+r\bZ)})$ consists of all functions $f:\bF_q^*\rightarrow\bF_q$ that can be expressed in the form $f(x)=\beta_{x\Omega_r}\cdot x$ for some $\beta\in\bF_q^{\bF_q^*/\Omega_r}$.
\end{lemma}

Lemma~\ref{lem:piecelininf} follows directly from the fact that the evaluation of a polynomial $h(X)$ on inputs in $x\Omega_r$ equals the evaluation of $h(X)\pmod{X^r-x^r}$ on such inputs. The lemma implies that every function in $P$ equals a linear function on the restriction to inputs in each coset $x\Omega_r$.

We show in Lemma~\ref{lem:qTBdef} that $P\subseteq C^\perp$, from which we obtain the local recoverability of $\cC$:

\begin{corollary}[Restatement of Corollary~\ref{cor:qTBLRC}]
  \label{cor:qTBLRCinf}
  The qTB code $\cC=\CSS(C,C)$ given in Definition~\ref{def:qTBcodeinf} is locally recoverable with locality $r$.
\end{corollary}
\begin{proof}
  As described in Section~\ref{sec:qlrcdefinf} and formalized in Corollary~\ref{cor:CSSLRC}, it suffices to show that for each $\alpha\in\bF_q^*$, there exists some $f\in C^\perp$ such that $\alpha\in\supp(f)$ and $|\supp(f)|\leq r$. But because $P\subseteq C^\perp$ as described above, the piecewise linear function $f:\bF_q^*\rightarrow\bF_q$ given by $f(x)=x$ for $x\in\alpha\Omega_r$ and $f(x)=0$ for $x\notin\alpha\Omega_r$ satisfies these criteria, as desired.
\end{proof}

It only remains to bound the distance of the qTB code $\cC$. For this purpose, we show the following result, which directly implies Theorem~\ref{thm:qTBdisinf}.

\begin{theorem}[Restatement of Theorem~\ref{thm:qTBdis}]
  \label{thm:qTBdisinf2}
  The qTB code $\cC$ in Definition~\ref{def:qTBcodeinf} with a prime locality parameter $r$ has distance at least
  \begin{equation}
    \label{eq:qTBdisinf2}
    d = (q-1)\left(1-\frac{1}{2r}-\sqrt{\frac{1}{4r^2}+\frac{r-1}{r}\cdot\frac{\ell-1}{q-1}}\right).
  \end{equation}
\end{theorem}

As $r\rightarrow\infty$, the bound in Theorem~\ref{thm:qTBdisinf2} gives
\begin{equation*}
  \frac{d}{q-1} \rightarrow 1-\sqrt{\frac{\ell-1}{q-1}}.
\end{equation*}
As was mentioned in Section~\ref{sec:qTBinf}, this bound is reminiscent of the Johnson bound (see e.g.~Theorem~7.3.1 of \cite{guruswami_essential_2022}), which in particular implies that a classical code of dimension $\ell$ and block length $q-1$ whose distance approaches the Singleton bound is \textit{list-decodable} from at least a fraction of errors approaching $1-\sqrt{\ell/(q-1)}$ as the alphabet size and block length grow large.

It is unclear if there is a deeper reason for this similarity between the Johnson bound and Theorem~\ref{thm:qTBdis}. However, it is interesting that we are able to improve the distance beyond this Johnson-like bound by folding our qTB codes, just as \cite{guruswami_explicit_2008} introduced folded RS codes to improve beyond the Johnson bound for the list-decoding radius of RS codes.

\begin{proof}[Proof sketch of Theorem~\ref{thm:qTBdisinf2}]
  Recall that the distance of $\cC=\CSS(C,C)$ equals the minimum weight of any element of $C\setminus C^\perp$. As $P\subseteq C^\perp$, it therefore suffices show that every element of $C\setminus P$ has weight at least the value $d$ given in~(\ref{eq:qTBdisinf2}).

  For this purpose, consider any $\evl(f)\in C\setminus P$. At a high level, the proof will proceed as follows. We define a polynomial $G$ such that whenever $f$ has multiple roots in a coset $x\Omega_r\subseteq\bF_q^*$, $G$ also has many roots in that coset. We can then bound the number of roots of $f$ by bounding the number of roots of $G$, which we in turn bound by the degree of $G$. By constructing $G$ to have low degree relative to the number of roots of $f$, we obtain the desired result.

  We now formally define $G$. By definition we may decompose $f(X)=g(X)+h(X)$, where $g(X)\in\bF_q[X]^{[\ell]\setminus(\pm 1+r\bZ)}$ evaluates to a nonzero classical TB codeword, and $h(X)\in\bF_q[X]^{[q-1]\cap(1+r\bZ)}$ evaluates to a piecewise linear function. Define the polynomial $G(X)\in\bF_q[X]$ by
  \begin{equation}
    \label{eq:GJohnsoninf}
    G(X) = \prod_{i=1}^{r-1}(\omega_r^{-i}g(\omega_r^iX)-g(X)).
  \end{equation}
  By definition, the degree of $G$ satisfies
  \begin{equation}
    \label{eq:degGinf}
    \deg(G) \leq (r-1)(\ell-1).
  \end{equation}
  Meanwhile, we bound the number of roots of $G$ as follows. For every $x\in\bF_q^*$ and every $r$th root of unity $\omega_r^i\in\Omega_r\setminus\{1\}$ such that $f(x)=0$ and $f(\omega_r^ix)=0$, then
  \begin{align*}
    \omega_r^{-i}g(\omega_r^ix)-g(x)
    &= -\omega_r^{-i}h(\omega_r^ix)+h(x) = -h(x)+h(x) = 0,
  \end{align*}
  where the second equality above holds because $h$ is piecewise linear. Thus $G$ has a root for every ordered pair $(x,\omega_r^ix)$ of roots of $f$ whose ratio is an $r$th root of unity $\omega_r^i$. Summing over all such pairs of roots of $f$ (see Section~\ref{sec:qTBdis} for details), we find that the total number of roots of $G$ is at least
  \begin{equation*}
    \frac{r}{q-1}(q-1-|\evl(f)|)^2-(q-1-|\evl(f)|).
  \end{equation*}
  This expression must be bounded above by the RHS of~(\ref{eq:degGinf}); rearranging terms in the resulting inequality yields the desired bound $|\evl(f)|\geq d$ for $d$ given in~(\ref{eq:qTBdisinf2}).
\end{proof}

\subsubsection{Folded Quantum Tamo-Barg Codes}
\label{sec:fqTBoverviewinf}
As described in Section~\ref{sec:qTBinf} and Section~\ref{sec:qTBoverviewinf}, we were only able to prove a ``Johnson-like'' bound on the distance of qTB codes, which does not approach our Singleton-like bounds for qLRCs described in Section~\ref{sec:singletoninf}. We address this issue by introducing a ``folded'' version of qTB codes, for which we show the distance does approach the Singleton bound for large localities $r$. The proof of this distance bound for folded qTB codes (Theorem~\ref{thm:fqTBdisinf}) is quite involved, so in this section we simply provide a brief description of the main ideas involved. The reader is referred to Section~\ref{sec:fqTBdis} for the full proof.

We define folded qTB codes by grouping together symbols of qTB codes as described below.

\begin{definition}[Restatement of Definition~\ref{def:fqTBcode}]
  \label{def:fqTBcodeinf}
  As in Definition~\ref{def:qTBcodeinf}, let $\cC=\CSS(C,C)$ be the qTB code with parameters $q,r,\ell$. Given an additional folding parameter $s|(q-1)/r$, we define the \textbf{folded quantum Tamo-Barg (fqTB) code} $\tilde{\cC}$ to be the quantum code of alphabet size $q^s$ and block length $(q-1)/s$ obtained as follows. Fix a generator $\omega_{q-1}$ for $\bF_q^*$, and then for every $i\in[(q-1)/s]$, we block together the $s$ components (each of alphabet size $q$) at positions $\{\omega_{q-1}^{si},\omega_{q-1}^{si+1},\dots,\omega_{q-1}^{si+s-1}\}$ in $\cC$ into a single component (of alphabet size $q^s$) of the folded code $\tilde{\cC}$. 
\end{definition}

A folded qTB code by construction inherits the rate and local recoverability of the underlying unfolded qTB code. Thus to prove Theorem~\ref{thm:fqTBdisinf}, we simply need to bound the distance of the fqTB code. For this purpose, define $S$ as in~(\ref{eq:qTBsupportinf}), so that the underlying unfolded qTB code $\cC=\CSS(C,C)$ has $C=\evl(\bF_q[X]^S)$. The distance of $\tilde{\cC}$ by definition equals the minimum over all codewords $\evl(f)\in C\setminus C^\perp$ of the number of distinct blocks $\{\omega_{q-1}^{si},\omega_{q-1}^{si+1},\dots,\omega_{q-1}^{si+s-1}\}$ (see Definition~\ref{def:fqTBcodeinf}) in which $f$ takes at least one nonzero value.

At a high level, we bound this minimum distance by considering two types of message polynomials $f(X)=\sum_{i\in S}f_iX^i\in\bF_q[X]^S$ separately. For the first type of message polynomial, we use a similar argument as described in the proof sketch of Theorem~\ref{thm:qTBdisinf2} in Section~\ref{sec:qTBoverviewinf}. We are able to leverage the folding to replace the polynomial $G(X)$ in~(\ref{eq:GJohnsoninf}) with another polynomial, which consists of the composition of $f$ with a determinant polynomial. This alternative choice of $G(X)$ detects roots of $f$ more efficiently relative to its degree, and hence yields a better distance bound (see Claim~\ref{claim:determinant}). However, this method breaks down when the coefficients $f_i$ of $f$ are supported in a small number of distinct values $i\pmod{r}$, that is, when there are $\ll r$ distinct values $i\pmod{r}$ among all $i$ with $f_i\neq 0$. The reason for the breakdown on these ``bad'' polynomials $f$ stems from the fact that we need $G(X)$ to be a nonzero polynomial, which becomes more difficult when $f$ has fewer nonzero coefficients.

Fortunately, we can consider such ``bad'' message polynomials $f$ separately, and instead apply an uncertainty principle (Proposition~\ref{prop:uncertainty}) to bound the weight of the associated codeword (Claim~\ref{claim:uncertainty}). Intuitively, this uncertainty principle implies that if the coefficients $f_i$ are zero for most values of $i\pmod{r}$, then most of the evaluation points in $\evl(f)$ must be nonzero, so $\evl(f)$ has large weight.

Combining our bounds from the two types of message polynomials described above, we obtain a bound on the distance of the fqTB code $\cC$, which yields Theorem~\ref{thm:fqTBdisinf}.



\subsection{Open Questions}
Our work leads the the following open questions:
\begin{itemize}
\item Can explicit qLRCs be constructed with a better rate-distance-locality tradeoff than given by Theorem~\ref{thm:fqTBdisinf}? Two approaches here are to improve our distance bound for fqTB codes in Theorem~\ref{thm:fqTBdisinf}, or to introduce new constructions of qLRCs for which stronger bounds can be shown. In general, the goal is to close the gap to the Singleton-like bounds in Section~\ref{sec:singleton}.
\item How does the optimal rate-distance-locality tradeoff of qLRCs differ from that of qLDPC codes? As described in Section~\ref{sec:singletoninf}, our bounds on qLRCs imply bounds on qLDPC codes, but it is an interesting question whether the additional structure in qLDPC codes can be used to show stronger bounds. This question also points towards the more general line of inquiry into relationships between notions of locality in quantum codes, of which local recoverability and LDPC are two examples of interest.
\end{itemize}

\section{Preliminaries}
\label{sec:prelim}
In this section, we introduce some notation, and then present preliminary definitions and results pertaining to classical codes, quantum codes, and local recoverability.

\subsection{Notation}
For $n\in\bN$, we let $[n]=\{0,\dots,n-1\}$.

For a prime power $q$, let $\bF_q$ denote the finite field of order $q$. For $n\in\bN$ and for $x,y\in\bF_q^n$, let $x\cdot y=\sum_{i\in[n]}x_iy_i$ denote the standard dot product.

Let $\bF_q^*=\bF_q\setminus\{0\}$ denote the multiplicative group of $\bF_q$. For $r|(q-1)$, we let $\omega_r\in\bF_q^*$ denote a primitive $r$th root of unity, and we let $\Omega_r=\{1,\omega_r,\dots,\omega_r^{r-1}\}$ denote the group of all $r$th roots of unity. While the choice of $\omega_r$ is not unique, we typically will think of fixing some $(q-1)$st root of unity $\omega_{q-1}\in\bF_q^*$, and then let $\omega_r=\omega_{q-1}^{(q-1)/r}$.

For a finite alphabet $A$, we denote the Hamming distance between two strings $x,y\in A^n$ by $\dis(x,y)=|\{i\in[n]:x_i\neq y_i\}|$. For two sets $X,Y\subseteq A^n$, the minimum Hamming distance is denoted $\dis(X,Y)=\min_{x\in X,y\in Y}\dis(x,y)$. If $A$ is an abelian group (e.g.~$A=\bF_q$, or $A=\bF_q^s$) we denote the Hamming weight of a string $x\in A^n$ by $|x|=\dis(0,x)=|\{i\in[n]:x_i\neq 0\}|$, so that $\dis(x,y)=|x-y|$.

We denote by $\cM(A)$ the set of density matrices on a set $A$ of qudits, where the local dimension $a$ of the qudits will be clear from context. Thus $\rho\in\cM(A)$ if $\rho$ is a positive semi-definite Hermitian operator $\rho:(\bC^a)^{\otimes A}\rightarrow(\bC^a)^{\otimes A}$ of trace $1$. For a pure state $\ket{\psi}\in(\bC^a)^{\otimes A}$, we let $\psi=\ket{\psi}\bra{\psi}\in\cM(A)$ denote the associated density matrix. For a density matrix $\rho\in\cM(A)$, if $B\subseteq A$ is a subset of the qudits, we denote the reduced density matrix $\rho_B=\Tr_{A\setminus B}\rho\in\cM(B)$.

A quantum channel $\cK$ from qudits $A$ to qudits $B$ is a completely positive, trace preserving map $\cK:\cM(A)\rightarrow\cM(B)$. Equivalently, a quantum channel is a map $\cK:\cM(A)\rightarrow\cM(B)$ that can be expressed in the form $\cK(\rho)=\sum_{\nu}K_\nu\rho K_\nu^\dagger$ for a set of operators $\{K_\nu:(\bC^a)^{\otimes A}\rightarrow(\bC^a)^{\otimes B}\}_\nu$, called Kraus operators, which satisfy $\sum_\nu K_\nu^\dagger K_\nu=I$.

We let the weight of a quantum channel $\cK:\cM(A)\rightarrow\cM(A)$ refer to the number of distinct qudits on which the channel acts nontrivially. Therefore if $\cK$ has weight $w$, then there exists a subset $W\subseteq A$ of size $|W|=w$ such that each Kraus operator $K_\nu$ of $K$ can be decomposed as $K_\nu=K_\nu^W\otimes I_{A\setminus W}$ for $K_\nu^W:(\bC^a)^{\otimes W}\rightarrow(\bC^a)^{\otimes W}$.

\subsection{Classical Codes}
Below we define the notion of a classical code.

\begin{definition}
  A \textbf{classical code} of \textbf{block length $n$}, \textbf{dimension $k$}, and \textbf{alphabet size $a$} is a subset $C\subseteq[a]^n$ of size $|C|=a^k$.

  The \textbf{distance $d$} of $C$ is defined as the minimum Hamming distance $d=\min_{c\neq c'\in C}\dis(c,c')$ between any two distinct codewords $c,c'\in C$.

  We say that $C$ is \textbf{linear} if $a=q$ is a prime power so that $[a]\cong\bF_q$ and $C\subseteq\bF_q^n$ is a linear subspace. Similarly, $C$ is \textbf{$\bF_q$-linear} if $a=q^s$ for some prime power $q$ so that $[a]\cong\bF_q^s$ and $C\subseteq(\bF_q^s)^n\cong\bF_q^{sn}$ is linear subspace over the field $\bF_q$. The \textbf{dual} of a linear code is $C^\perp=\{c'\in\bF_q^n:c'\cdot c=0\forall c\in C\}$. Nonzero elements of the dual code $C^\perp$ are sometimes referred to as \textbf{parity checks} of $C$.
\end{definition}

All classical codes satisfy that following fundamental limitation on the tradeoff between distance and dimension, whose proof is a simple application of the pigeonhole principle.

\begin{proposition}[Singleton bound]
  \label{prop:cSB}
  If $C$ is a classical code of block length $n$, dimension $k>0$, and distance $d$, then
  \begin{equation*}
    k \leq n-(d-1).
  \end{equation*}
\end{proposition}

A classical (linear) code can be expressed as the image of a (linear) encoding map $\Enc:\bF_q^k\rightarrow\bF_q^n$ that maps messages to codewords encoding the messages. However, we will mostly consider the space of codewords without reference to the encoding map.

Throughout this paper we restrict attention to linear or $\bF_q$-linear codes, as will be clear from context. In particular, $\bF_q$-linear codes will arise from linear codes over $\bF_q$ by ``folding'', that is, by grouping together symbols of the linear code into larger symbols.

\begin{definition}
  A \textbf{classical locally recoverable code (cLRC)} of \textbf{locality $r$}, block length $n$, and alphabet size $a$ is a classical code $C\subseteq[a]^n$ together with a set of \textbf{local recovery maps} $\Rec_i$ for every $i\in[n]$ satisfying the following properties:
  \begin{enumerate}
  \item There exists a set $I_i\subseteq[n]$ of size $|I_i|\leq r$ with $i\in I_i$ such that $\Rec_i$ is a function $\Rec_i:[a]^{I_i\setminus\{i\}}\rightarrow[a]$.
  \item It holds for every codeword $c\in C$ that
    \begin{equation*}
      \Rec_i(c|_{I_i\setminus\{i\}}) = c_i.
    \end{equation*}
  \end{enumerate}
\end{definition}

We specify that $i\in I_i$ for notational convenience due to the following well known form of linear cLRCs. However, we emphasize that an LRC with locality $r$ can recover any given codeword symbol from $r-1$ other symbols.

\begin{lemma}[Well known]
  \label{lem:cLRClin}
  If $C\subseteq\bF_q^n$ is a linear code such that every $i\in[n]$ lies in the support of some parity check of Hamming weight $\leq r$, then $C$ is locally recoverable with locality $r$.
\end{lemma}
\begin{proof}
  Fix $i\in[n]$, and let $c'\in C^\perp$ be a parity check of weight $\leq r$ with $i\in I_i:=\supp(c')$. Then by definition every $c\in C$ satisfies $c_i=-\sum_{j\in I_i\setminus\{i\}}(c_j'/c_i')c_j$, which gives the desired recovery function for index $i$.
\end{proof}

\subsection{Polynomial Evaluation Codes}
\label{sec:polyeval}
Below we describe a particularly useful type of classical linear code given by evaluations of polynomials. Here we let $\bF_q[X_1,\dots,X_m]$ denote the ring of $m$-variate polynomials over $\bF_q$.

\begin{definition}
  \label{def:polyeval}
  For a subset $S\subseteq\bZ_{\geq 0}$, let
  \begin{equation*}
    \bF_q[X]^S=\left\{\sum_{i\in S}a_iX^i:a\in\bF_q^S\right\}
  \end{equation*}
  denote the $|S|$-dimensional subspace of $\bF_q[X]$ consisting of polynomials with zero coefficients for all monomials $X^i$ with $i\notin S$. Define the \textbf{polynomial evaluation map}
  \begin{equation*}
    \evl:\bF_q[X]^S\rightarrow\bF_q^{\bF_q^*}\cong\bF_q^{q-1}
  \end{equation*}
  by $\evl(f)=(f(x))_{x\in\bF_q^*}$, so that $\evl(f)$ outputs the list of evaluations of $f$ at all points in $\bF_q^*$. The image $\evl(\bF_q[X]^S)$ of $\evl$ is a linear code, called a \textbf{polynomial evaluation code}, with alphabet size $q$, block length $q-1$, and dimension $|S|$.
\end{definition}

Below, we present the well-known Reed-Solomon codes, along with their folded counterpart.

\begin{definition}
  \label{def:RS}
  For a prime power $q$ and an integer $\ell\in[q]$, the \textbf{Reed-Solomon (RS) code} is the polynomial evaluation code $C=\evl(\bF_q[X]^{[\ell]})$.

  Given an additional folding parameter $s|(q-1)$, the \textbf{folded Reed-Solomon (fRS) code} $\tilde{C}$ is the $\bF_q$-linear code of alphabet size $q^s$ and block length $(q-1)/s$ obtained as follows. Fix a generator $\omega_{q-1}$ for $\bF_q^*$, and then for every $i\in[(q-1)/s]$, block together the $s$ components (each an element of $\bF_q$) at positions $\{\omega_{q-1}^{si},\omega_{q-1}^{si+1},\dots,\omega_{q-1}^{si+s-1}\}$ in $C$ into a single component (which is an element of $\bF_q^s$) of the folded code $\tilde{C}$.
\end{definition}

By definition Reed-Solomon codes have block length $q-1$, dimension $\ell$, and distance $d=q-\ell$. Specifically, it holds that $d\geq(q-1)-(\ell-1)=q-\ell$ because ever polynomial of degree $<\ell$ has $<\ell$ roots, and $d\leq q-\ell$ by the Singleton bound. Folded Reed-Solomon codes similarly have block length $(q-1)/s$, dimension $\ell/s$, and distance $(q-\ell)/s$.

Our decoding algorithm for the quantum LRCs we construct will rely on the efficient list-decodability of the (folded) Reed-Solomon codes, as stated in the known results below. However, we first must define list decoding.

\begin{definition}
  Let $C$ be a code of block length $n$ over alphabet $A$, and let $\delta\geq 0$. An \textbf{$e$-list-decoding} algorithm for $C$ is an algorithm that takes as input a corrupted codeword $b\in A^n$, and outputs the list of every codeword $c\in C$ such that $\dis(b,c)\leq e$.
\end{definition}

In the statements below, recall that the RS code has block length $q-1$ and rate $R=\ell/(q-1)$, while the fRS code has block length $(q-1)/s$ and rate $R=\ell/(q-1)$.

\begin{theorem}[\cite{guruswami_improved_1998}]
  \label{thm:RSdec}
  The RS code with parameters $q,\ell$ has an $e$-list-decoding algorithm that runs in time $q^{O(1)}$ for
  \begin{equation*}
    e = (q-1)(1-\sqrt{R})
  \end{equation*}
\end{theorem}

\begin{theorem}[\cite{guruswami_explicit_2008}]
  \label{thm:fRSdec}
  The fRS code with parameters $q,\ell,s$ has an $e$-list decoding algorithm that runs in time $q^{O(\sqrt{s})}$, where
  \begin{equation*}
    e = \frac{q-1}{s}\left(1-\left(1+\frac{2}{\sqrt{s}}\right)R^{1-1/\sqrt{s}}\right)-2,
  \end{equation*}
  assuming that $s\geq s_0$ and $q\geq q_0=q_0(s)$ for sufficiently large constants $s_0,q_0(s)$.
\end{theorem}

Note that the running times of the above algorithms implicitly give bounds on the lengths of the lists they output.

\subsection{Quantum Codes}
In this section, we describe relevant background on quantum codes.

\begin{definition}
  A \textbf{quantum code} of \textbf{block length $n$}, \textbf{dimension $k$}, and \textbf{local dimension (that is, alphabet size) $a$} is a $a^k$-dimensional linear subspace $\cC\subseteq(\bC^a)^{\otimes n}$.
\end{definition}

The definition above of a quantum code as a linear subspace of Hilbert space assumes a unitary encoding map, as we may express $\cC=\{\ket{\Enc\ket{\phi}\ket{0^{n-k}}}:\ket{\phi}\in(\bC^a)^{\otimes k}\}$ for a unitary encoding map $\Enc:(\bC^a)^{\otimes n}\rightarrow(\bC^a)^{\otimes n}$. In this paper we restrict attention to codes with such unitary encodings unless explicitly stated otherwise. However, there are more general coding schemes with arbitrary channel encodings; we leave it as an open question whether such non-unitary encodings could be beneficial for local recovery.

Below we define the distance of a quantum code as one plus the maximum number of erasures that the code can correct.

\begin{definition}
  Let $\cC$ be a quantum code of block length $n$. Given a subset $S\subseteq[n]$, we say that $\cC$ can \textbf{decode from erasures in $S$} if there exists a decoding channel $\Dec^S:\cM([n]\setminus S)\rightarrow\cM([n])$ that satisfies
  \begin{equation*}
    \Dec^S(\psi_{[n]\setminus S}) = \psi.
  \end{equation*}
  for every $\psi\in\cC$.

  The \textbf{distance} of $\cC$ is then defined as the maximum value $d\in\bN$ such that $\cC$ can decode from erasures in every $S\subseteq[n]$ of size $|S|<d$.
\end{definition}

The well-known fact below (see e.g.~Fact~2 of~\cite{anshu_circuit_2022} for a proof) shows that for a quantum code of distance $d$, any $d-1$ codeword symbols contain no information about the encoded message. Note that there is no classical analogue to this fact, as classically every linear code has a systematic encoding, meaning that the first $k$ codword symbols equal the message.

\begin{lemma}[Local indistinguishability]
  \label{lem:localind}
  If $\cC$ is a quantum code of block length $n$ that can decode from erasures in some set $S\subseteq[n]$, then the reduced density matrix $\psi_S$ is the same for all $\psi\in\cC$.
\end{lemma}

The following well known quantum analogue of Proposition~\ref{prop:cSB} presents a tradeoff between distance and dimension for quantum codes. See for instance Section~12.4.3 of \cite{nielsen_quantum_2010} for a proof.

\begin{proposition}[Singleton bound]
  \label{prop:qSB}
  If $\cC$ is a quanum code of block length $n$, dimension $k>0$, and distance $d$, then
  \begin{equation*}
    k \leq n-2(d-1).
  \end{equation*}
\end{proposition}

We next define stabilizer codes, a well-known class of quantum codes that are defined as the simultaneous $+1$ eigenspace of a set of Pauli operators, called \textit{stabilizers}. We will first need to define the Pauli operators.

\begin{definition}
  Let $\bF_q$ be a finite field of characteristic $p$. For $\alpha\in\bF_q$, define the \textbf{$q$-ary Pauli operators} $X^\alpha,Z^\alpha\in\bC^{\bF_q\times\bF_q}\cong\bC^{q\times q}$ so that for $x\in\bF_q$,
  \begin{align*}
    X^\alpha\ket{x} &= \ket{x+\alpha} \\
    Z^\alpha\ket{x} &= e^{(2\pi i/p)\tr_{\bF_q/\bF_p}(\alpha x)}\ket{x}.
  \end{align*}
  Observe that Paulis always commute up to a phase, and specifically $e^{(2\pi i/p)\tr_{\bF_q/\bF_p}(\alpha\beta)}X^\alpha Z^\beta=Z^\beta X^\alpha$.
  
  For $n\in\bN$, an \textbf{$n$-qudit Pauli} is an operator of the form $X^\alpha Z^\beta=\bigotimes_{i=1}^nX^{\alpha_i}Z^{\alpha_i}$ for $\alpha,\beta\in\bF_q^n$. We take $\cP_q^n$ denote the group of $n$-qudit Paulis modulo the global phase, so that $\cP_q^n\cong\bF_q^{2n}$. The \textbf{support $\supp(P)$} of an $n$-qudit Pauli $P\in\cP_q^n$ is defined as the set $i\in[n]$ of qudits on which $P$ acts nontrivially. That is, $\supp(X^\alpha Z^\beta)=\{i\in[n]:(\alpha_i,\beta_i)\neq(0,0)\}$. The \textbf{Hamming weight} of $P$ is $|P|=|\supp(P)|$.
\end{definition}

As $\cP_q^n$ is isomorphic to $\bF_q^{2n}$, it is a vector space where the action of $\gamma\in\bF_q$ on $X^\alpha Z^\beta\in\cP_q^n$ gives $X^{\gamma\alpha}Z^{\gamma\beta}$. We therefore let a \textit{subspace}, or (by slight abuse of terminology) \textit{subgroup}, of $\cP_q^n$ be any subset that is a $\bF_q$-subspace under the isomorphism $\cP_q^n\cong\bF_q^{2n}$. Therefore in particular we require a ``subgroup'' of $\cP_q^n$ to be closed under the action of $\bF_q$. Furthermore, although we defined $\cP_q^n$ to be the group of Paulis modulo a global phase, we abuse notation and call a subgroup of $\cP_q^n$ \textit{abelian} if all the operators in the subgroup commute with each other, taking into account the global phase. We also denote $(\cP_q^s)^n=\cP_q^{sn}$, which we interpret as the group of length-$n$ strings of elements in $\cP_q^s$; this group will naturally arise when we consider folded stabilizer codes in Section~\ref{sec:fqTBcode}.

\begin{definition}
  For a vector space $\bF_q^s$ over a finite field, a \textbf{stabilizer code of block length $n$ over the alphabet $\bF_q^s$} is a subgroup of the Hilbert space $(\bC^{\bF_q^s})^{\otimes n}$ that is specified as the simultaneous $+1$ eigenspace of some abelian subgroup of $(\cP_q^s)^n$. This abelian subgroup is called the \textbf{stabilizer group} of the code, and its elements are the code's \textbf{stabilizers}.
\end{definition}

We will mostly be concerned with a specific type of stabilizer code called a CSS code, which is specified by two classical codes satisfying an orthogonality condition.

\begin{definition}
  Given classical codes $C_X,C_Z\subseteq(\bF_q^s)^n$ such that $C_X^\perp\subseteq C_Z$, the associated \textbf{CSS code} $\cC=\CSS(C_X,C_Z)$ is given by $\cC=\spn\{\sum_{y\in C_X^\perp}\ket{x+y}:x\in C_Z\}\subseteq(\bC^{\bF_q^s})^{\otimes n}$. Equivalently, $\cC$ is the stabilizer code with stabilizer group $\{X^\alpha Z^\beta:\alpha\in C_X^\perp,\beta\in C_Z^\perp\}$.
\end{definition}

It is well known that to decode a CSS code $\cC=\CSS(C_X,C_Z)$, it is sufficient to have classical decoders for $C_X$ and $C_Z$, as stated below:

\begin{proposition}[Well known]
  \label{prop:CSSdec}
  Let $\cC=\CSS(C_X,C_Z)$ be a CSS code of block length $n$ over the alphabet $A=\bF_q^s$ of size $a=|A|$. Let $e\geq 0$ be an integer such that for each permutation $(\alpha,\beta)$ of $(X,Z)$, there exists a (classical) decoding algorithm $\Dec_\alpha$ that takes as input a (classical) corrupted codeword $c+b$ for some $c\in C_\alpha$ and some corruption $b\in A^n$ of Hamming weight $|b|\leq e$, and outputs some $c'\in C_\alpha$ such that $c'-c\in C_\beta^\perp$.

  Then $\cC$ has a decoding algorithm $\Dec$ that recovers from errors of weight $e$, so $\cC$ has distance $d\geq 2e+1$. Furthermore, if each $\Dec_\alpha$ has running time $T_\alpha(n,a)$, then $\Dec$ has running time $T_X(n,a)+T_Z(n,a)+O(n^3\poly\log a)$.
\end{proposition}
\begin{proof}[Proof sketch]
  The result is standard, so we just provide a brief outline. Let $C_X=\ker H_X$ and $C_Z=\ker H_Z$. Given a corrupted code state $\rho=\cE(\psi)$ for a code state $\psi\in\cC$ and an error channel $\cE$ acting on $\leq e$ of the codeword qudits, the decoder $\Dec(\rho)$ first performs syndrome measurements for $H_X,H_Z$ to collapse $\rho$ to $E\psi E^\dagger$, so that the error on $\psi$ is collapsed to some Pauli $E=X^{b_X}Z^{b_Z}$ of weight $\leq e$, where $s_X=H_Xb_X$ and $s_Z=H_Zb_Z$ are the outputs of the syndrome measurements. Then for each permutation $(\alpha,\beta)$ of $(X,Z)$, the decoder performs Gaussian elimination to find some $a_\alpha$ such that $H_\alpha a_\alpha=s_\alpha$, so that $a_\alpha=c_\alpha+b_\alpha$ for some (currently unknown) $c_\alpha\in C_\alpha$. The decoder then runs $\Dec_\alpha(a_\alpha)$, which outputs $c_\alpha+p_\alpha$ for some $p_\alpha\in C_\beta^\perp$. Letting $b_\alpha'=a_\alpha-(c_\alpha+p_\alpha)=b_\alpha-p_\alpha$ and $E'=X^{b_X'}Z^{b_Z'}$, the decoder then applies ${E'}^\dagger$ to output ${E'}^\dagger E\psi E^\dagger E'=\psi$. This final equality holds because by assumption $X^{p_X}$ and $Z^{p_Z}$ are stabilizers of $\cC$, so they preserve $\psi$.

  The running time of $\Dec$ defined above is $T_X(n,a)+T_Z(n,a)+O(n^3\poly\log a)$ because outside of the calls to $\Dec_X$ and $\Dec_Z$, all operations run in time $O(n^2\poly\log a)$, except the Gaussian elimination which runs in time $O(n^3\poly\log a)$.
\end{proof}

\begin{definition}
  \label{def:qLRC}
  A \textbf{quantum locally recoverable code (qLRC)} of \textbf{locality $r$}, block length $n$, and local dimension (i.e.~alphabet size) $a$ is a quantum code $\cC\subseteq(\bC^a)^{\otimes n}$ together with a set of \textbf{local recovery channels} $\Rec_i$ for every $i\in[n]$ satisfying the following properties:
  \begin{enumerate}
  \item There exists a set $I_i\subseteq[n]$ of size $|I_i|\leq r$ with $i\in I_i$ such that $\Rec_i$ is a quantum channel $\Rec_i:\cM(I_i\setminus\{i\})\rightarrow\cM(I_i)$ with input qudits $I_i\setminus\{i\}$ and output qudits $I_i$.
  \item It holds for every code state $\psi\in\cC$ that
    \begin{equation*}
      \Rec_i\otimes I_{[n]\setminus I_i}(\psi_{[n]\setminus\{i\}}) = \psi.
    \end{equation*}
  \end{enumerate}
\end{definition}

That is, a qLRC permits local recovery against erasure of a single qudit, or equivalently, against a corruption in a single known location.

Below, we present a quantum analogue of Lemma~\ref{lem:cLRClin} for stabilizer codes.

\begin{proposition}
  \label{prop:stabLRC}
  Let $\cC$ be a stabilizer code with stabilizer group $S\subseteq\cP_q^n$. Assume that for every $i\in[n]$, there exist stabilizers $P=X^\alpha Z^\beta,Q=X^\gamma Z^\delta\in S$ with $|\supp(P)\cup\supp(Q)|\leq r$ such that $(\alpha_i,\beta_i)=(1,0)$ and $(\gamma_i,\delta_i)=(0,1)$. Then $\cC$ is locally recoverable with locality $r$.
\end{proposition}
\begin{proof}[Proof sketch]
  The result is a standard application of the well-known fact that a stabilizer code detects any Pauli error that anticommutes with one of the code's stabilizers. We will therefore just provide a brief proof sketch. For any given $i\in[n]$, let $P=X^\alpha Z^\beta,Q=X^\gamma Z^\delta\in S$ be the Paulis given by the proposition statement, and let $I_i=\supp(P)\cup\supp(Q)$. We will construct a recovery channel to revert the action of an arbitrary error channel (such as a totally depolarizing channel) acting on qudit $i$. The recovery channel first performs syndrome measurements for $P$ and $Q$, specifically by performing measurements on the corrupted code state for the operators $P^\eta$ and $Q^\eta$ for all $\eta\in\bF_q$. These measurements will collapse the error on qudit $i$ to a simultaneous eigenspace of these operators $P^\eta$ and $Q^\eta$ for $\eta\in\bF_q$, and the measurement outcomes (i.e.~the syndrome) give the eigenvalues of the projected state for each of these operators. That is, the error on qudit $i$ is projected to some single-qudit Pauli $E=X^{e_X}Z^{e_Z}$ for $e_X,e_Z\in\bF_q^n$ both supported in component~$i$. The syndrome measurement provides the phases $(EP^\eta)(P^\eta E)^\dagger$ and $(EQ^\eta)(Q^\eta E)^\dagger$ for all $\eta\in\bF_q$. These phases can be used to compute the symplectic inner products $s_P=e_X\cdot\beta-e_Z\cdot\alpha$ and $s_Q=e_X\cdot\delta-e_Z\cdot\gamma$. By the assumption that $(\alpha_i,\beta_i)=(1,0)$ and $(\gamma_i,\delta_i)=(0,1)$, it follows that $s_P=-(e_Z)_i$ and $s_Q=(e_X)_i$. That is, $s_P$ and $s_Q$ give the $Z$ and $X$ errors that occured on qudit $i$, respectively. The recovery channel then simply applies $X^{-e_X}Z^{-e_Z}$ to revert the error and recover the original code state.
\end{proof}

Proposition~\ref{prop:stabLRC} immediately implies that a CSS code is locally recoverable if its two classical codes are locally recoverable linear codes in the sense of Lemma~\ref{lem:cLRClin}:

\begin{corollary}
  \label{cor:CSSLRC}
  A CSS code $\cC=\CSS(C_X,C_Z)$ is locally recoverable with locality $r$ if for every $i\in[n]$, there exist parity checks $c_X'\in C_X^\perp,c_Z'\in C_Z^\perp$ such that $i\in\supp(c_X')\cap\supp(c_Z')$ and $|\supp(c_X')\cup\supp(c_Z')|\leq r$.
\end{corollary}


\section{Singleton-Like Bounds for qLRCs}
\label{sec:singleton}
This section presents Singleton-like bounds for qLRCs, by combining techniques for proving quantum Singleton bounds (see for example Section~12.4.3 of \cite{nielsen_quantum_2010}) and classical LRC Singleton-like bounds \cite{gopalan_locality_2012}.

Below, for a qLRC $\cC$ of block length $n$ and locality $i$, we let $I_i\subseteq[n]$ denote the set of size $r$ with $i\in I_i$ such that the local recovery map $\Rec_i$ for $\cC$ recovers the $i$th code component from the components in $I_i\setminus\{i\}$.

\begin{theorem}
  \label{thm:singletongen}
  Let $\cC$ be a qLRC of block length $n$, dimension $k>0$, distance $d$, and locality $r$. Then
  \begin{equation*}
    k \leq n-2(d-1)-\left\lfloor\frac{n-(d-1)}{r}\right\rfloor-\left\lfloor\frac{n-2(d-1)-\left\lfloor\frac{n-(d-1)}{r}\right\rfloor}{r}\right\rfloor,
  \end{equation*}
\end{theorem}

Omitting the last term in the bound in Theorem~\ref{thm:singletongen}, it follows that $k\leq n-2(d-1)-(n-(d-1))/r+1$, or equivalently,
\begin{equation*}
  d-1 \leq \frac{\left(1-\frac{1}{r}\right)n-k+1}{2+\frac{1}{r}}.
\end{equation*}
Thus for fixed locality $r$, even as the rate $k/n\rightarrow 0$, the relative distance of a qLRC satisfies $d/n\leq 1/2-\Omega(1/r)$. In contrast, general (non-LRC) quantum codes can have $d/n\rightarrow 1/2$ as $k/n\rightarrow 0$. Thus for codes of low rate, there is a fundamental cost to imposing a constant locality $r$ in quantum codes.

This fact differs from the classical case, as the classical Singleton-like bound \cite{gopalan_locality_2012} states that every cLRC of block length $n$, dimension $k>0$, distance $d$, and locality $r$ satisifies
\begin{equation*}
  d \leq n-k-\left\lceil\frac{k}{r}\right\rceil+2.
\end{equation*}
Tamo and Barg \cite{tamo_family_2014} provided an explicit construction of cLRCs that exactly meet this bound. Therefore there exist cLRCs with $k/n\rightarrow 0$ and $d/n\rightarrow 1$, which approaches the classical Singleton bound for general (non-locally-recoverable) codes.

Because all quantum LDPC stabilizer codes are by definition quantum LRCs, the above observation in particular implies that there is a cost (in terms of rate-distance tradeoff) to requiring a quantum code be LDPC with some fixed locality. As for LRCs, letting the rate $k/n\rightarrow 0$ gives a disconnect between the quantum case and the classical case. Specifically, Theorem~\ref{thm:singletongen} implies that all qLDPC codes with constant locality (i.e.~check weight) $r$ have relative distance $d/n\leq 1/2-\Omega(1/r)$. In contrast, classically there are LDPC codes with constant locality $r$ and distance approaching $1$, such as the repetition code ($r=2$) and the Hadamard code ($r=3$) over a growing field size.

It is an interesting question to determine the optimal tradeoff between rate, distance, and locality for qLDPC codes in addition to for qLRCs, and to classify the additional cost of requiring a quantum code to be LDPC compared to just requiring local recoverability.

\begin{proof}[Proof of Theorem~\ref{thm:singletongen}]
  We partition the code symbols $Q=[n]$ into five disjoint subsets $S_1,V_1,S_2,V_2,W$. Here as a shorthand we denote $A\sqcup B$ by $AB$, so we will have $Q=S_1V_1S_2V_2W$. To construct this partition, we first construct $S_1,V_1$ as follows:
  \begin{enumerate}
  \item Initialize sets $S_1=\emptyset$ and $T_1=\emptyset$.
  \item\label{it:setS1} Repeat the following step
    \begin{equation*}
      s_1 = \left\lfloor\frac{n-(d-1)}{r}\right\rfloor
    \end{equation*}
    times: choose some $i\in Q\setminus S_1T_1$, add $i$ to $S_1$, and add all elements of $I_i\setminus(\{i\}\cup S_1)$ to $T_1$. Note that by construction $S_1\cap T_1=\emptyset$.
  \item\label{it:setV1} Set $V_1$ to be any $(d-1)$-element subset of $Q\setminus S_1T_1$.
  \end{enumerate}
  The above procedure is guaranteed to successfully output $S_1,V_1$ with $|S_1|=s_1$ and $|V_1|=d-1$, as each iteration in step~\ref{it:setS1} adds $1$ element to $S_1$ and adds at most $r$ elements to $S_1T_1$, so after all $s_1$ iterations it still holds that $|Q\setminus S_1T_1| \geq n-rs_1 \geq d-1$.
  
  Let $Q_1=Q\setminus S_1V_1$. For a code state $\rho\in\cC$, given access to the restriction $\rho_{Q_1}=\Tr_{Q\setminus Q_1}\rho$ to components in $Q_1$, we may by construction pass through all $i\in S_1$ one-by-one in the same order $S_1$ was constructed, and apply $\Rec_i$ at each step, to recover $\rho_{Q_1S_1}$. Then as $Q\setminus Q_1S_1=V_1$, which has size $d-1$, the global decoding algorithm of $\cC$ recovers $\rho$ from $\rho_{Q_1S_1}$. Thus the code components in $Q_1$ can be used to completely recover the entire codeword. It follows by Lemma~\ref{lem:localind} that the reduced density matrix $\rho_{Q\setminus Q_1}=\rho_{S_1V_1}$ is the same for all code states $\rho\in\cC$.

  Now if $|Q_1|<d-1$, then the global decoding algorithm for $\cC$ can recover a code state from components in $Q\setminus Q_1=S_1V_1$. But as we showed above that a code state can also be recovered from components in $Q_1$, the code must have dimension $k=0$, as otherwise we would be able to clone a code state by breaking it into the parts $S_1V_1$ and $Q_1$, and recovering the entire state from each part. But we assumed $k>0$, a contradiction, so it must be that $|Q_1|\geq d-1$.

  We then construct $S_1,V_2$ using a similar procedure as above, but on the restriction to components in $Q_1$:
  \begin{enumerate}
  \item Initialize sets $S_2=\emptyset$ and $T_2=\emptyset$.
  \item\label{it:setS2} Repeat the following step
    \begin{equation*}
      s_2 = \left\lfloor\frac{|Q_1|-(d-1)}{r}\right\rfloor = \left\lfloor\frac{n-2(d-1)-s_1}{r}\right\rfloor = \left\lfloor\frac{n-2(d-1)-\left\lfloor\frac{n-(d-1)}{r}\right\rfloor}{r}\right\rfloor
    \end{equation*}
    times: choose some $i\in Q_1\setminus S_2T_2$, add $i$ to $S_2$, and add all elements of $I_i\setminus(\{i\}\cup S_2)$ to $T_2$.
  \item\label{it:setV2} Set $V_2$ to be any $(d-1)$-element subset of $Q_1\setminus S_2T_2$.
  \end{enumerate}
  Note that the procedure above is guaranteed to successfully output $S_2,V_2$ with $|S_2|=s_2$ and $|V_2|=d-1$ because $|Q_1|\geq d-1$ as shown above, and each iteration of step~(\ref{it:setS1}) adds 1 element to $S_2$ and adds at most $r$ elements to $S_2T_2$, so after all $s_2$ iterations $|S_2|=s_2$ and $|Q_1\setminus S_2T_2|\geq|Q_1|-|S_2T_2|\geq |Q_1|-rs_2\geq d-1$. Letting $Q_2=Q\setminus S_2V_2$, then by similar reasoning as above for $Q_1$, the code components in $S_2V_2$ can be recovered from components in $Q_2$, so all code states $\rho\in\cC$ have the same reduced density matrix $\rho_{Q\setminus Q_2}=\rho_{S_2V_2}$.

  Now let $W=Q\setminus S_1V_1S_2V_2$, so that we have our desired partition $Q=S_1V_1S_2V_2W$. Let $q$ be the local dimension of $\cC$, and let $A$ be a set of $k$ additional qudits of local dimension $q$. Define
  \begin{equation}
    \label{eq:entcodestate}
    \ket{\psi}_{AQ} = \frac{1}{q^k}\sum_{x\in[q]^k}\ket{x}_A\otimes(\Enc\ket{x})_Q
  \end{equation}
  to be the state obtained by applying the encoding map $\Enc$ of $\cC$ to one register from a maximally entanged pair of message states $\sum_{x\in[q]^k}\ket{x}\otimes\ket{x}$. Let $\psi=\ket{\psi}\bra{\psi}$ be the associated density matrix.

  For $B\subseteq AQ$, let $S(B)=S(\psi_B)=\Tr(-\psi_B\log_q\psi_B)$ denote the von Neumann entropy of qudits $B$ in the state $\ket{\psi}$. Then
  \begin{align}
    \label{eq:entropies}
    \begin{split}
      S(A)+S(S_1V_1) &= S(AS_1V_1) = S(S_2V_2W) \leq S(S_2V_2)+S(W) \\
      S(A)+S(S_2V_2) &= S(AS_2V_2) = S(S_1V_1W) \leq S(S_1V_1)+S(W).
    \end{split}
  \end{align}
  In both lines above, the inequality holds by the subadditivity of entropy. Meanwhile, the first equality in both lines holds because for $i=1,2$, we showed above that all code states $\rho\in\cC$ have the same reduced density matrix $\rho_{S_iV_i}$, which must equal $\psi_{S_iV_i}$. Therefore $\psi_{AS_iV_i}=q^{-k}\sum_{k\in[q]^k}\ket{x}\bra{x}\otimes\rho_{S_iV_i}=\psi_A\otimes\psi_{S_iV_i}$, and thus $S(\psi_{AS_iV_i})=S(\psi_A)+S(\psi_{S_iV_i})$ as entropy is additive over tensor products.

  Adding the two inequalities in~(\ref{eq:entropies}) gives that $S(A)\leq S(W)$. Now $\psi_A=q^{-k}\sum_{x\in[q]^k}\ket{x}\bra{x}$, so $S(A)=k$. Meanwhile, $S(W)\leq|W|=n-|S_1V_2S_wV_2|=n-2(d-1)-s_1-s_2$. Thus
  \begin{equation*}
    k = S(A) \leq S(W) = n-2(d-1)-\left\lfloor\frac{n-(d-1)}{r}\right\rfloor-\left\lfloor\frac{n-2(d-1)-\left\lfloor\frac{n-(d-1)}{r}\right\rfloor}{r}\right\rfloor,
  \end{equation*}
  as desired.
\end{proof}



While Theorem~\ref{thm:singletongen} applies for arbitrary qLRCs, our explicit construction of qLRCs in Section~\ref{sec:construct} has a specific structure: the recovery sets $I_i$ form a partition of the code components $[n]$, that is, for all $i,j\in[n]$ either $I_i=I_j$ or $I_i\cap I_j=\emptyset$.

Below, we show that if we assume the recovery sets $I_i$ have this partition structure, then a stronger Singleton-like bound than in Theorem~\ref{thm:singletongen} holds. We will assume for simplicity that all recovery sets $I_i$ have size exactly $|I_i|=r$, as is the case for our explicit constructions in Section~\ref{sec:construct}; a similar result also holds in the more general case where $|I_i|\leq r$.

\begin{theorem}
  \label{thm:singletonpart}
  Let $\cC$ be a qLRC of block length $n$, dimension $k>0$, distance $d$, and locality $r$. Assume that for all $i\in[n]$, $|I_i|=r$, and for all $i,j\in[n]$, either $I_i=I_j$ or $I_i\cap I_j=\emptyset$. Then
  \begin{equation*}
    k \leq \left(1-\frac{2}{r}\right)n - 2\left(d-1-\left\lceil\frac{d-1}{r-1}\right\rceil\right).
  \end{equation*}
\end{theorem}
\begin{proof}
  The proof is similar to that of Theorem~\ref{thm:singletongen}, except we choose the partition $S_1V_1S_2V_2W$ of the qudits more carefully. Specifically, again let $Q=[n]$ denote the set of code qudits, and as a shorthand denote $AB=A\sqcup B$ for disjoint $A,B\subseteq Q$.


  Because $\{I_i:i\in Q\}$ forms a partition of $Q=[n]$ with all $|I_i|=r$, we must have $r|n$. Denote by $J_1,\dots,J_{n/r}$ the distinct partition elements, so $\{I_i:i\in[n]\}=\{J_j:j\in[n/r]\}$ and therefore $[n]=J_1\sqcup\cdots\sqcup J_{n/r}$.
  
  We define a partition $Q=S_1V_1S_2V_2W$ of the qudits as follows. Let $a=\lceil(d-1)/(r-1)\rceil$, and fix an arbitrary partition $[n/r]=A_1\sqcup A_2\sqcup B$ such that $|A_1|=|A_2|=a$. Choose $S_1,S_2\subseteq[n]$ to be a pair of disjoint subsets such that $S_1$ consists of a single element of $J_j$ for each $j\in A_1\sqcup B$, and $B$ consists of a single element of $J_j$ for each $j\in A_2\sqcup B$. Note that because $r\geq 2$ by definition, we can indeed choose such $S_1,S_2$ that are disjoint.

  Then define $V_1$ to be any set of $d-1$ elements in $\sqcup_{j\in A_2}J_j\setminus S_2$, and define $V_2$ to be any set of $d-1$ elements in $\sqcup_{j\in A_1}J_j\setminus S_2$. Note that there exist such sets $V_1,V_2$ because by construction $|\sqcup_{j\in A_2}J_j\setminus S_2|=|\sqcup_{j\in A_1}J_j\setminus S_2|=(r-1)a\geq d-1$.

  Finally define $W=Q\setminus S_1V_1S_2V_2$, so that by construction $Q=S_1V_1S_2V_2W$ forms a partition of $Q=[n]$.

  For $b=1,2$, if we are given access to components in $Q\setminus S_bV_b$ of a code state $\rho\in\cC$, then as $V_b$ is by construction disjoint from $\sqcup_{i\in S_b}I_i$, we can apply the local recovery maps $\Rec_i$ for $i\in S_b$ to recover the components of $\rho$ in $S_b$. Then we have all components of $\rho$ in $Q\setminus V_b$, so as $|V_b|=d-1$, we can apply the global decoding map for $\cC$ to recover $\rho$. It follows that the components in $Q\setminus S_bV_b$ completely determine the code state $\rho$, so the reduced density matrix $\rho_{S_bV_b}$ must be the same for all $\rho\in\cC$.

  We now proceed as in the proof of Theorem~\ref{thm:singletongen}, though with different values for $|S_1|,|S_2|$. Define $\ket{\psi}$ as in~(\ref{eq:entcodestate}), so that~(\ref{eq:entropies}) again holds here. Summing these inequalities in~(\ref{eq:entropies}), we conclude as in the proof of Theorem~\ref{thm:singletongen} that
  \begin{align*}
    k
    &\leq |W| \\
    &= n-2(d-1)-|S_1|-|S_2| \\
    &= n-2(d-1)-2(a+(n/r-2a)) \\
    &= \left(1-\frac{2}{r}\right)n - 2\left(d-1-\left\lceil\frac{d-1}{r-1}\right\rceil\right),
  \end{align*}
  as desired.
\end{proof}

The separation between the bounds in Theorem~\ref{thm:singletongen} and Theorem~\ref{thm:singletonpart} raises the possibility that qLRCs could benefit from having recovery sets $I_i$ that do not simply form a partition of $[n]$. Such a situation would be in contrast to the classical case. Specifically, the classical LRCs of \cite{tamo_family_2014} have recovery sets $I_i$ forming a partition of $[n]$, yet they exactly meet the Singleton-like bound of \cite{gopalan_locality_2012}, which applies to LRCs with arbitrary recovery structure (see also Chapter~19 of \cite{guruswami_essential_2022}). Note that the cLRCs of \cite{tamo_family_2014} have linear-sized alphabets; the case of smaller alphabets is considered in \cite{guruswami_how_2019}.

\section{Basic Constructions Using Known Techniques}
\label{sec:basic}
In this section, we describe how known techniques yield some basic qLRCs, and highlight the issues with such techniques that our construction addresses.

In Section~\ref{sec:random}, we show how to randomly sample a qLRC whose rate is within $O(1/r)$ of the Singleton-like bounds in Section~\ref{sec:singleton} over an alphabet of size $q=2^{O(r)}$. However, this construction is non-explicit, and as a result we have no efficient algorithm to certify the distance bound, and no efficient algorithm to decode from errors in unknown locations. In contrast, our fqTB codes presented in Section~\ref{sec:construct} below only have rate within $O(1/\sqrt{r})$ of the Singleton-like bound, and have alphabet size growing polynomially in the block length. However, our fqTB codes are explicit, so their distance bound is guaranteed, and we provide an efficient decoding algorithm against errors in uknown locations in Section~\ref{sec:decode}.

In Section~\ref{sec:ael}, we then show how to explicitly construct an efficiently-decodable qLRC using Alon-Edmunds-Luby (AEL) \cite{alon_linear_1995} distance amplification and alphabet reduction. This construction concatenates a small random qLRC, which can be found efficiently via brute force, with a large quantum code of high rate and constant relative distance. The key step of the AEL construction is to then permute the symbols according to the edges of an expander graph, which amplifies the relative distance to almost that of the inner code. The resulting construction (see Proposition~\ref{prop:qLRCAEL} and the subsequent discussion) yields qLRCs whose rate is within $O(1/r^{1/4})$ of the Singleton-lke bounds in Section~\ref{sec:singleton} over an alphabet of size $q=2^{O(r)}$. In contrast, as mentioned above, our fqTB codes in Section~\ref{sec:construct} have a better rate-distance-locality tradeoff, as their rate is within $O(1/\sqrt{r})$ of the Singleton-like bound.



\subsection{Random qLRCs}
\label{sec:random}
In this section, we analyze the following natural random construction of a qLRC.

\begin{definition}[Random qLRC]
  \label{def:random}
  Given a block length $n$, a locality parameter $r|n$, an integer $\ell\in[n/2-n/r]$, and an alphabet $\bF_q$, we define a random qLRC to be a CSS code $\cC=\CSS(C_X,C_Z)$ that is sampled as follows. Initialize parity check matrices $H_X,H_Z$ for $C_X,C_Z$ respectively to be $n/r\times n$ matrices such that for $i\in[n]$ and $j\in[n/r]$,
  \begin{align*}
    (H_X)_{i,j} &= \begin{cases}
      1,&j\in\{ri,\dots,ri+r-1\} \\
      0,&\text{otherwise}
    \end{cases} \\
    (H_Z)_{i,j} &= \begin{cases}
      -(r-1),&j=ri \\
      1,&j\in\{ri+1,\dots,ri+r-1\} \\
      0,&\text{otherwise}.
    \end{cases}
  \end{align*}
  That is, $H_X,H_Z$ are initialized so that $\CSS(\ker H_X,\ker H_Z)$ is a well-defined CSS code that is a qLRC of maximal possible rate, whose recovery sets form a partition of the code components $[n]$. We then sequentially add $\ell$ random rows to each of $H_X$ and $H_Z$ subject to the orthogonality relations, as follows:
  \begin{enumerate}
  \item\label{it:makeHX} Repeat $\ell$ times: sample a uniformly random vector in $\text{row-span}(H_Z)^\perp\setminus\text{row-span}(H_X)$ and add it as a new row in $H_X$.
  \item\label{it:makeHZ} Repeat $\ell$ times: sample a uniformly random vector in $\text{row-span}(H_X)^\perp\setminus\text{row-span}(H_Z)$ and add it as a new row in $H_Z$.
  \end{enumerate}
  Thus we have sampled matrices $H_X,H_Z\in\bF_q^{(n/r+\ell)\times n}$ with orthgonal row-spans. Thus we may let $C_X=\ker H_X,C_Z=\ker H_Z$, obtain our desired well-defined CSS code $\cC=\CSS(C_X,C_Z)$.
\end{definition}

The following lemma is immediate from Definition~\ref{def:random}.

\begin{lemma}
  \label{lem:randomrate}
  The random qLRC $\cC$ in Definition~\ref{def:random} is a qLRC of locality $r$ and dimension $k=n-2(n/r+\ell)$. Furthermore, each $i\in[n]$ has recovery set $I_i=\{ar,\dots,ar+r-1\}$ for $a=\lfloor i/r\rfloor$.
\end{lemma}

We now analyze this distance of these random qLRCs. We will need the following definition.

\begin{definition}
  The \textbf{$q$-ary entropy function $H_q:[0,1]\rightarrow[0,1]$} is defined by
  \begin{equation*}
    H_q(x) = x\log_q(q-1) - x\log_q(x) - (1-x)\log_q(1-x).
  \end{equation*}
\end{definition}

The following proposition bounds the distance of random qLRCs. The proof is similar to that of the well-known Gilbert-Varshamov (GV) bound for truly random codes.

\begin{proposition}
  \label{prop:random}
  Given any $\delta,\epsilon>0$, the distance $d(\cC)$ of a random qLRC $\cC$ from Definition~\ref{def:random} with parameters $n$, $r$, and $\ell\geq(H_q(\delta)+\epsilon)n$ over the alphabet $\bF_q$ satisfies
  \begin{equation}
    \label{eq:randomentbound}
    \Pr[d(\cC)\geq\delta n] > 1-2q^{-\epsilon n}.
  \end{equation}
  In particular, if $q\geq 2^{2/\epsilon}$, then it holds for all $\ell\geq(\delta+\epsilon)n$ that
  \begin{equation}
    \label{eq:randomnoentbound}
    \Pr[d(\cC)\geq\delta n] > 1-2q^{-\epsilon n/2}.
  \end{equation}
\end{proposition}
\begin{proof}
  We first prove~(\ref{eq:randomentbound}). Consider any fixed nonzero $y\in\bF_q^n$. We will first compute $\Pr[y\in C_Z\setminus C_X^\perp]$, so we can then union bound over all low-weight $y$. After step~\ref{it:makeHX} in the sampling procedure in Definition~\ref{def:random}, then $H_X$ is fully constructed and $C_X^\perp=\text{row-span}(H_X)$. If $y\in C_X^\perp$ then $y\notin C_Z\setminus C_X^\perp$, so we will condition on $y\notin C_X^\perp$. For $0\leq i\leq\ell$, let $H_Z^{(i)}$ denote the matrix $H_Z$ after the $i$th iteration of step~\ref{it:makeHZ} of the sampling procedure.

  Now if $y\in C_Z\setminus C_X^\perp$, then for every $0\leq i\leq\ell-1$, it must hold that $y\in\ker H_Z^{(i)}$, and that $y\in\ker H_Z^{(i+1)}$, where $H_Z^{(i+1)}$ is obtained from $H_Z^{(i)}$ by adding a row given by a uniformly random $h\in C_X\setminus\text{row-span}(H_Z^{(i)})$. Because $y\notin C_X^\perp$ by assumption, exactly $1/q$-fraction of the vectors in $C_X$ are orthogonal to $y$. Conditioning on the event that $y\in\ker H_Z^{(i)}$, then all vectors in $\text{row-span}(H_Z^{(i)})$ are orthogonal to $y$. Thus less than $1/q$-fraction of all $h\in C_X\setminus\text{row-span}(H_Z^{(i)})$ are orthogonal to $y$, that is,
  \begin{equation*}
    \Pr[y\in\ker H_Z^{(i+1)}\setminus C_X^\perp \mid y\in\ker H_Z^{(i)}\setminus C_X^\perp] = \Pr_{h\sim\Unif(C_X\setminus\text{row-span}(H_Z^{(i)}))}[y\cdot h=0] < \frac{1}{q}.
  \end{equation*}
  Therefore
  \begin{align*}
    \Pr[y\in C_Z\setminus C_X^\perp]
    &= \Pr[y\notin C_X^\perp]\cdot\Pr[y\in\ker H_Z^{(0)}|y\notin C_X^\perp] \\
    &\hspace{2em}\cdot\prod_{i=0}^{\ell-1}\Pr[y\in\ker H_Z^{(i+1)}\setminus C_X^\perp|y\in\ker H_Z^{(i)}\setminus C_X^\perp] \\
    &< \frac{1}{q^\ell}.
  \end{align*}

  Union bounding over all $y\in\bF_q^n$ of Hamming weight $|y|\leq\delta n$ gives
  \begin{align*}
    \Pr[\exists y\in C_Z\setminus C_X^\perp\text{ s.t. }|y|\leq\delta n]
    &< |\{y\in\bF_q^n:|y|\leq\delta n\}|\cdot\frac{q}{q^\ell} \\
    &\leq q^{H_q(\delta)n-\ell} \\
    &= q^{-\epsilon n},
  \end{align*}
  where the second inequality above follows from the well-known fact that a radius-$\delta n$ Hamming ball in $\bF_q^n$ has volume $\leq q^{H_q(\delta)n}$ (see for instance Proposition~3.3.3 in \cite{guruswami_essential_2022}), and the final equality follows from the assumption that $\ell\geq(H_q(\delta)+\epsilon)n$.

  Now observe that the final distribution over codes $\cC=\CSS(C_X,C_Z)$ remains the same if we swap the order of steps~\ref{it:makeHX},\ref{it:makeHZ} in the sampling procedure in Definition~\ref{def:random}, as in both cases by symmetry we obtain a uniform distribution over pairs of matrices $H_X,H_Z\in\bF_q^{(n/r+\ell)\times n}$ satisfying $H_XH_Z^\top=0$ such that the first $n/r$ rows of each matrix are as given in Definition~\ref{def:random}. Thus swapping $C_X$ and $C_Z$ in the above argument gives that
  \begin{equation*}
    \Pr[\exists y\in C_X\setminus C_Z^\perp\text{ s.t. }|y|\leq\delta n] < q^{-\epsilon n}.
  \end{equation*}
  Then because $d(\cC)=\min_{y\in(C_Z\setminus C_X^\perp)\cup(C_X\setminus C_Z^\perp)}|y|$, we conclude that
  \begin{equation*}
    \Pr[d(\cC)\leq\delta n] \leq 2q^{-\epsilon n}.
  \end{equation*}

  It remains the prove the second claim in the proposition statement. That is, assuming that $q\geq 2^{2/\epsilon}$ and $\ell\geq(\delta+\epsilon)n$, we want to prove that~(\ref{eq:randomnoentbound}) holds. But by definition $H_q(\delta)\leq\delta+1/\log_2q$, so for $q\geq 2^{2/\epsilon}$ then $H_q(\delta)\leq\delta+\epsilon/2$, which means that $\ell\geq(\delta+\epsilon)n\geq(H_q(\delta)+\epsilon/2)n$. Then applying~(\ref{eq:randomentbound}) with $\epsilon/2$ replacing $\epsilon$ gives
  \begin{equation*}
    \Pr[d(\cC)\geq\delta n] > 1-2q^{-\epsilon n/2},
  \end{equation*}
  as desired.
\end{proof}

Letting the alphabet size $q\rightarrow\infty$, we immediately obtain the following corollary.

\begin{corollary}
  For every fixed $\delta>0$ and for every alphabet size $q$, there exists $\epsilon=\epsilon(q)>0$ with $\epsilon(q)\rightarrow 0$ as $q\rightarrow\infty$ such that a random qLRC $\cC$ from Definition~\ref{def:random} with parameters $n$, $r$, $\ell=\lceil(H_q(\delta)+\epsilon/4)n\rceil$ has dimension
  \begin{equation*}
    k \geq \left(1-\frac{2}{r}\right)n-2d-\epsilon n,
  \end{equation*}
  and with high probability has distance $d=\delta n$.
\end{corollary}

Thus for large $q$, the dimension of these random qLRCs is within $2d/r$ of the Singleton-like upper bound in Theorem~\ref{thm:singletonpart}.

\begin{remark}
  For simplicity and consistency with the rest of the paper, we have chosen to consider random CSS qLRCs. By not restricting to CSS codes and instead taking random stabilizer qLRCs, one can obtain a slightly better dependence on alphabet size. That is, a smaller alphabet size $q$ would be needed to achieve any given rate-distance tradeoff dictated by Lemma~\ref{lem:randomrate} and Proposition~\ref{prop:random}. However, in the limit of large $q$, which is our primary concern, the rate-distance tradeoff is the same whether or not we restrict to CSS codes. A similar phenomenon occurs with ordinary random quantum codes (with no locality constraints), and is the reason that random stabilizer codes (i.e.~the quantum GV bound) have better distance than random CSS codes over small alphabets.
\end{remark}

\subsection{Explicit qLRCs from AEL Distance Amplification}
\label{sec:ael}
In this section, we show how to construct explicit qLRCs using Alon-Edmunds-Luby (AEL) distance amplification and alphabet reduction \cite{alon_linear_1995}. This distance amplification technique has seen extensive use in classical coding theory \cite{guruswami_expander-based_2001,guruswami_near-optimal_2002,guruswami_linear_2003,guruswami_explicit_2008,hemenway_linear-time_2018,kopparty_high-rate_2016,gopi_locally_2018,hemenway_local_2020}, but has only recently been studied in the quantum setting \cite{bergamaschi_approaching_2022,wills_tradeoff_2023}. Here we apply the AEL technique with an inner code given by a random qLRC from Section~\ref{sec:random}.

\subsubsection{Review of AEL Technique}
\label{sec:AELreview}
Here review the AEL distance amplification and alphabet reduction technique of \cite{alon_linear_1995}. We will subsequently specify how we parameterize it to obtain qLRCs. The construction will use the following notion of an expander.

\begin{definition}
  A $\Delta$-regular bipartite graph $G=(L,R,E)$ with $|L|=|R|=n$ is \textbf{$\lambda$-pseudorandom} if it holds for every $S\subseteq L$ and $T\subseteq R$ that
  \begin{equation*}
    \left||E(S,T)|-\frac{\Delta|S||T|}{n}\right| \leq \lambda\Delta\sqrt{|S||T|}.
  \end{equation*}
\end{definition}

Recall that a \textit{$\lambda$-spectral expander} is a graph for which the second largest absolute value of an eigenvalue of the random walk matrix is $\leq\lambda$. By the well-known expander mixing lemma (see for instance Leamm~4.15 in \cite{vadhan_pseudorandomness_2012}), the double cover of a $\lambda$-spectral expander is $\lambda$-pseudorandom, so the explicit near-Ramanujan graphs of \cite{mohanty_explicit_2021} prove the following.

\begin{proposition}[\cite{mohanty_explicit_2021}]
  \label{prop:expanders}
  For every integer $\Delta\geq 3$, there exists an infinite explicit family of $\Delta$-regular $\lambda$-pseudorandom graphs for $\lambda=2/\sqrt{\Delta}$.
\end{proposition}

We now describe the AEL construction. Let $\cC_{\text{out}}$ and $\cC_{\text{in}}$ be $[[n_{\text{out}},k_{\text{out}}=R_{\text{out}}n_{\text{out}},d_{\text{out}}=\delta_{\text{out}}n_{\text{out}}]]_{q_\text{out}}$ and $[[n_{\text{in}},k_{\text{in}}=R_{\text{in}}n_{\text{in}},d_{\text{in}}=\delta_{\text{in}}n_{\text{in}}]]_{q_\text{in}}$ quantum codes respectively such that $q_{\text{out}}=q_{\text{in}}^{k_{\text{in}}}$. Let $\cC_\diamond$ denote the $[[n_{\text{out}}n_{\text{in}},k_{\text{out}},k_{\text{in}}]]_{q_{\text{in}}}$ concatenated code, so that its encoder
\begin{equation*}
  \Enc_\diamond = \Enc_{\text{in}}^{\otimes n_{\text{out}}}\circ\Enc_{\text{out}}
\end{equation*}
first applies the encoder for $\cC_{\text{out}}$, and then applies the encoder for $\cC_{\text{in}}$ to each symbol of the resulting codeword of $\cC_{\text{out}}$, where we use the assumption that $q_{\text{out}}=q_{\text{in}}^{k_{\text{in}}}$ to view each such symbol as a message for $\cC_{\text{in}}$.

Now for some $\Delta|n_{\text{in}}$, we partition each inner code block of size $n_{\text{in}}$ into $n_{\text{in}}/\Delta$ subsets each of size $\Delta$, and we fold together the $\Delta$ components in each subset to obtain a folded concatenated code $\tilde{\cC}_\diamond$, which is a $[[n=n_{\text{out}}n_{\text{in}}/\Delta,k=k_{\text{out}}k_{\text{in}}/\Delta]]_q$ code of alphabet size $q=q_{\text{in}}^\Delta$.

To construct our final $[[n,k]]_q$ code $\cC$, we permute the symbols of $\tilde{\cC}_\diamond$ according to a $\Delta$-regular $\lambda$-pseudorandom bipartite graph $G=(L,R,E)$ with $|L|=|R|=\tilde{n}_\diamond$. Formally, let $\pi_G:[n]\times[\Delta]\rightarrow[n]\times[\Delta]$ be the permutation that maps $(i,j)\in[n]\times[\Delta]$ to the unique $(i',j')\in[n]\times[\Delta]$ such that the $j$th edge incident to vertex $i\in L$ equals the $j'$th edge incident to vertex $i'\in R$. Note that here we have fixed some ordering of the edges incident to each vertex. We then let $\cC$ have encoder
\begin{equation*}
  \Enc = \pi_G\circ\Enc_\diamond
\end{equation*}
that simply encodes the message into $\cC_\diamond$, and then permutes the symbols according to $\pi_G$.  That is, for $j\in[\Delta]$, the $j$th symbol of component $i\in[n]$ in $\tilde{\cC}_\diamond$ is mapped to the $j'$th symbol of component $i'$ in $\cC$ for $(i',j')=\pi_G(i,j)$. Thus for the composition here we view $\pi_G$ as an invertible map $\mathbb{F}_q^{n\times\Delta} \to \mathbb{F}_q^{n\times\Delta}$ in the natural way.

The following lemma is immediate from the definition of the code $\cC$

\begin{lemma}
  \label{lem:AELrate}
  $\cC$ has rate $R=R_{\text{out}}R_{\text{in}}$.
\end{lemma}.

The distance of $\cC$ is given by the following proposition. The statement and proof of this AEL distance bound follows from \cite{alon_linear_1995}, and was previously shown in the quantum setting by \cite{bergamaschi_approaching_2022}; we repeat the proof in Appendix~\ref{app:omitproofs} for completeness in our specific setting. Below, we use the definitions and notation from the AEL construction above. We also say a code has \textit{decoding radius $a$} if it has a (possibly inefficient) decoding algorithm that corrects against errors in up to $a$ unknown locations.

\begin{proposition}
  \label{prop:ael}
  If $\cC_{\text{out}}$ and $\cC_{\text{in}}$ have respective decoding radii $\alpha_{\text{out}}n_{\text{out}}$ and $\alpha_{\text{in}}n_{\text{in}}$, then $\cC$ has decoding radius at least $\alpha n$ for $\alpha=\alpha_{\text{in}}-\lambda\sqrt{\alpha_{\text{in}}/\alpha_{\text{out}}}$. Furthermore, if the decoding algorithms for $\cC_{\text{out}}$ and $\cC_{\text{in}}$ run in time $\poly(n\log q)$, so does the decoding algorithm for $\cC$.
\end{proposition}

Because a code of distance $d$ may be (possibly inefficiently) decoded from $e$ errors for every $e<d/2$, and a code with decoding radius $e$ has distance $>2e$, the following corollary is immediate.

\begin{corollary}
  \label{cor:aeldis}
  If $\cC_{\text{out}}$ and $\cC_{\text{in}}$ have respective distances $\delta_{\text{out}}n_{\text{out}}$ and $\delta_{\text{in}}n_{\text{in}}$, then $\cC$ has distance at least $\delta n$ for $\delta=\delta_{\text{in}}-2\lambda\sqrt{\delta_{\text{in}}/\delta_{\text{out}}}$.
\end{corollary}

\subsubsection{Application to qLRCs}
We now apply the AEL described above to qLRCs. In this section we again continue to use the definitions and notation for the AEL construction from Section~\ref{sec:AELreview}.

\begin{lemma}
  \label{lem:AELloc}
  Let the inner code $\cC_{\text{in}}$ be a qLRC of locality $r_{\text{in}}|\Delta$, whose recovery sets $\{I_i:i\in[n_{\text{in}}]\}$ all have size $r_{\text{in}}$ and form a partition of $[n_{\text{in}}]$. The AEL construction then yields a qLRC $\cC$ of locality $\Delta r_{\text{in}}$.
\end{lemma}
\begin{proof}
  We may perform the folding step of the AEL construction described in Section~\ref{sec:AELreview} such that all elements of each $I_i$ are folded into the same length-$\Delta$ component in $\tilde{\cC}_\diamond$. The permutation $\pi_G$ then sends these $r$ elements of $I_i$ to distinct components in $\cC$, assuming we choose a $\lambda$-pseudorandom graph $G$ arising from a 2-cover of a $\lambda$-spectral expander with no self-loops; the expanders of \cite{mohanty_explicit_2021} described in Proposition~\ref{prop:expanders} indeed have no self-loops. Thus each of the $\Delta$ symbols of a component of $\cC$ can be recovered from $r_{\text{in}}$ symbols in other components, so $\cC$ is a qLRC of locality $r_{\text{in}}\Delta$.
\end{proof}

In summary, if we apply AEL with a qLRC $\cC_{\text{in}}$ of locality $r_{\text{in}}$, we obtain a qLRC of locality $r=\Delta r_{\text{in}}$ that has rate $R=R_{\text{out}}R_{\text{in}}$ and relative distance at least $\delta_{\text{in}}-2\lambda\sqrt{\delta_{\text{in}}/\delta_{\text{out}}}$, where we can take $\lambda=2/\sqrt{\Delta}$.

To optimize the parameters for the resulting qLRC $\cC$, we can take $\cC_{\text{out}}$ to be a quantum Reed-Solomon code, which is the CSS code formed by two (generalized) Reed-Solomon codes\footnote{If $C$ is a length-$(q-1)$ classical RS code as defined in Definition~\ref{def:RS}, then indeed $\CSS(C,C)$ is a well-defined CSS code of of alphabet size $q$ that lies on the quantum Singleton bound. This construction can be extended to general block lengths $n$ (such that $n+1$ is not a prime power) by using \textit{generalized} RS codes. In a generalized RS code, the messages are still low-degree polynomials, but the encoding map only evaluates the message polynomial at some points in $\bF_q^*$, and multiplies these evaluations by some fixed coefficients in $\bF_q^*$. By choosing appropriate such coefficients for generalized RS codes $C_X$ and $C_Z$, we can ensure that the CSS orthogonality relations are satisfied, so that $\CSS(C_X,C_Z)$ is a well-defined CSS code on the quantum Singleton bound. See for instance Section~3.1 of \cite{bergamaschi_approaching_2022} for more details.}. This code has alphabet size $q_{\text{out}}=O(n_{\text{out}})$ and lies on the quantum Singleton bound so that $R_{\text{out}}\geq 1-2\delta_{\text{out}}$. Then taking a random qLRC from Section~\ref{sec:random} as the inner code $\cC_{\text{in}}$, we obtain the following.

\begin{proposition}
  \label{prop:qLRCAEL}
  For every $r_{\text{in}},\Delta\geq 3$ such that $r_{\text{in}}|\Delta$, and for every $\delta_{\text{in}},\delta_{\text{out}},\epsilon>0$, there exists an infinite family of qLRCs with locality $r=r_{\text{in}}\Delta$, alphabet size $q=2^{O(\Delta/\epsilon)}$, rate
  \begin{equation*}
    R = (1-2\delta_{\text{out}})\left(1-2\delta_{\text{in}}-\frac{2}{r_{\text{in}}}-2\epsilon\right),
  \end{equation*}
  and relative distance at least
  \begin{equation*}
    \delta = \delta_{\text{in}}-4\sqrt{\frac{\delta_{\text{in}}}{\delta_{\text{out}}\Delta}}.
  \end{equation*}
  Furthermore, a qLRC of block length $n$ in this family can be constructed by a $\poly(n\log q)$-time randomized algorithm with high probability, and can be decoded by a $\poly(n\log q)$-time algorithm from up to $\delta n/2$ errors.
\end{proposition}
\begin{proof}
  As described above, we instantiate the AEL construction with $\cC_{\text{out}}$ to be a (generalized) qRS code of relative distance at least $\delta_{\text{out}}$, rate $R_{\text{out}}=1-2\delta_{\text{out}}$, and alphabet size $q_{\text{out}}=O(n_{\text{out}})$, and we instantiate $\cC_{\text{in}}$ to be a random qLRC from Proposition~\ref{prop:random} of relative distance at least $\delta_{\text{in}}$, rate $R_{\text{in}}=1-2\delta_{\text{in}}-2/r_{\text{in}}-2\epsilon$, locality $r_{\text{in}}$, and alphabet size $q_{\text{in}}=2^{O(1/\epsilon)}$. We take the expander $G$ to be given by Proposition~\ref{prop:expanders}. Taking $\cC$ to be the resulting code from applying the AEL construction, then the bounds on the parameters of $\cC$ in the proposition statement follow from Lemma~\ref{lem:AELloc}, Proposition~\ref{prop:random}, Lemma~\ref{lem:randomrate}, Lemma~\ref{lem:AELrate}, and Corollary~\ref{cor:aeldis}. 

  To efficiently construct $\cC$, we must simply find an inner code $\cC_{\text{in}}$ satisfying the distance bound in Proposition~\ref{prop:random}, as the code $\cC_{\text{out}}$ and the expander $G$ are by definition explicit. But because $q_{\text{out}}=O(n_{\text{out}})$, in $\poly(n\log q)$ time we may repeatedly generate a random qLRC $\cC_{\text{in}}=\CSS(C_X,C_Z)$ as described in Section~\ref{sec:random}, and via brute force check every codeword in $C_X\setminus C_Z$ and $C_Z\setminus C_X$ to compute the distance of $\cC_{\text{in}}$. By Proposition~\ref{prop:random}, with high probability a random such $\cC_{\text{in}}$ will have the desired distance $\delta_{\text{in}}$, so we successfully construct $\cC_{\text{in}}$ and therefore $\cC$ with high probability. The efficiency of the decoding algorithm follows from Proposition~\ref{prop:ael} along with the fact that RS (and therefore also qRS) codes are known to have efficient decoders from errors of weight up to half the distance (e.g.~Theorem~\ref{thm:RSdec}).
\end{proof}

\begin{remark}
  \label{rem:effconst}
  The codes in Proposition~\ref{prop:qLRCAEL} are efficiently constructable, but not technically explicit. This distinction arises because $q_{\text{out}}=\Theta(n_{\text{out}})$, so while the number of possible inner codewords $q_{\text{out}}=|\bF_{q_{\text{in}}}^{k_{\text{in}}}|$ is polynomial in $n\log q$, the number of possible inner codes $\cC_{\text{in}}=\CSS(C_X,C_Z)$ is superpolynomial in $n\log q$. We can resolve this issue to obtain explicit qLRCs $\cC$ with AEL by using an outer code $\cC_{\text{out}}$ with a smaller alphabet size than a quantum RS code. For instance, the alphabet size of $\cC_{\text{out}}$ can be reduced with another application of AEL, at the cost of a small loss in parameters.
\end{remark}

The qLRCs of locality $r$ and relative distance $\delta$ in Proposition~\ref{prop:qLRCAEL} have rate
\begin{align*}
  R
  &\leq (1-2\delta_{\text{out}})\left(1-2\delta-8\sqrt{\frac{\delta}{\delta_{\text{out}}\Delta}}-\frac{2\Delta}{r}-2\epsilon\right)
\end{align*}
for any $\epsilon>0$ that determines the alphabet size $q=2^{O(\Delta/\epsilon)}$. The RHS above is maximized when $\Delta$ is chosen to minimize the expression $8\sqrt{\delta/\delta_{\text{out}}\Delta}+2\Delta/r$; solving this minimization gives $\Delta=(2r)^{2/3}(\delta/\delta_{\text{out}})^{1/3}$, which in turn gives
\begin{align*}
  R
  &\leq (1-2\delta_{\text{out}})\left(1-2\delta-6\left(\frac{4\delta}{\delta_{\text{out}}r}\right)^{1/3}-2\epsilon\right) \\
  &\leq 1-2\delta-\Omega\left(\frac{\delta}{r}\right)^{1/3}.
\end{align*}
Thus the qLRCs constructed with AEL in Proposition~\ref{prop:qLRCAEL} have rate at least $\Omega(\delta/r)^{1/3}$ below the Singleton-like bound in Theorem~\ref{thm:singletongen}.

Furthermore, if we let $0<\delta<1/2$ be a fixed constant and let $r$ grow large, then the above inequality implies that
\begin{align*}
  R &\leq 1-2\delta-\Omega\left(\delta_{\text{out}}+\frac{1}{(\delta_{\text{out}}r)^{1/3}}\right).
\end{align*}
The RHS above is maximized at $\delta_{\text{out}}=\Theta(1/r^{1/4})$, giving
\begin{align}
  \label{eq:r14bound}
  R &\leq 1 - 2\delta - \Omega\left(\frac{1}{r}\right)^{1/4}.
\end{align}
Thus for fixed $\delta$ as $r$ grows large, the qLRCs in Proposition~\ref{prop:qLRCAEL} have rate at least $\Omega(r)^{1/4}$ below the Singelton-like bound in Theorem~\ref{thm:singletongen}.

The bound in~(\ref{eq:r14bound}) is tight, in the sense that if we fix $0<\delta<1/2$ in Proposition~\ref{prop:qLRCAEL} and then take $\delta_{\text{out}}=1/r^{1/4}$, $\Delta=r^{3/4}$, and $\epsilon=1/r^{1/4}$, Proposition~\ref{prop:qLRCAEL} gives qLRCs of relative distance $\delta$ and rate
\begin{equation*}
  R \geq 1 - 2\delta - O\left(\frac{1}{r}\right)^{1/4}
\end{equation*}
over an alphabet of size $q=2^{O(r)}$.

\section{Explicit Construction of qLRCs: Quantum Tamo-Barg Codes}
\label{sec:construct}
In this section, we present our construction of qLRCs, and show that they are indeed locally recoverable. The construction is essentially a quantum CSS version of the classical construction of Tamo and Barg \cite{tamo_family_2014}, though the quantum case is more involved due to issues surrounding orthogonality and degeneracy.

\subsection{Classical Tamo-Barg Codes}
We begin by describing the classical LRC construction of Tamo and Barg \cite{tamo_family_2014}. Recall below that for $n\in\bN$, we denote $[n]=\{0,\dots,n-1\}$. We will also use the notation from Definition~\ref{def:polyeval}. In particular, recall that the map $\evl:\bF_q[X]\rightarrow\bF_q^{\bF_q^*}$ maps a polynomial $f$ to the list of its evaluations on all inputs in $\bF_q^*$.

\begin{definition}[\cite{tamo_family_2014}]
  \label{def:TBcode}
  For a prime power $q$, a locality parameter $r|(q-1)$ with $r\geq 3$, and an integer $\ell\in[q]$, the \textbf{Tamo-Barg (TB) code} is the polynomial evaluation code $\evl(\bF_q[X]^S)$ for
  \begin{equation*}
    S = \{i\in[\ell]:i\not\equiv r-1\pmod{r}\}.
  \end{equation*}
\end{definition}

The following proposition of \cite{tamo_family_2014} states that the TB code is a cLRC achieving the Singleton-like bound.

\begin{proposition}[\cite{tamo_family_2014}]
  \label{prop:tb}
  The TB code in Definition~\ref{def:TBcode} has alphabet size $q$, block length $q-1$, dimension $\ell-\lfloor l/r\rfloor$, distance $\geq q-\ell$, and is locally recoverable with locality $r$.
\end{proposition}

Below we present a quantum analogue of TB codes that we show are also locally recoverable (Corollary~\ref{cor:qTBLRC}). The proof of local recoverability in both the classical and quantum cases simply rely on some polynomial manipulations (Lemma~\ref{lem:qTBdef} and Lemma~\ref{lem:piecelin} below).

\subsection{Quantum Tamo-Barg Codes}
\label{sec:qTBcode}
We now introduce our quantum analogue of Tamo-Barg codes.

\begin{definition}
  \label{def:qTBcode}
  For a prime power $q$, a locality parameter $r|(q-1)$ with $r\geq 3$, and an integer $\ell\in[q]$, we define the \textbf{quantum Tamo-Barg (qTB) code} to be the CSS code $\cC=\CSS(C_X,C_Z)$ with $C_X=C_Z=\evl(\bF_q[X]^S)$ for
  \begin{equation}
    \label{eq:qTBsupport}
    S = \{i\in[\ell]:i\not\equiv r-1\pmod{r}\} \cup \{i\in[q-1]:i\equiv 1\pmod{r}\}.
  \end{equation}
  By construction $\cC$ has local dimension $q$ and block length $q-1$.
\end{definition}

To understand the definition of a qTB code, and to see that it gives a well defined CSS code, consider the following general fact, which presents a certain kind of ``intersection'' of CSS codes.
\begin{lemma}
  \label{lem:intersect}
  For CSS codes $\cA=\CSS(A_X,A_Z)$ and $\cB=\CSS(B_X,B_Z)$ of block length $n$ over $\bF_q$, there exists a CSS code $\cC=\CSS(C_X,C_Z)$ given by
  \begin{align*}
    C_X &= (A_X\cap B_X)+B_Z^\perp \\
    C_Z &= (A_Z\cap B_Z)+B_X^\perp,
  \end{align*}
  so that
  \begin{align*}
    C_X^\perp &= (A_X^\perp\cap B_Z)+B_X^\perp \\
    C_Z^\perp &= (A_Z^\perp\cap B_X)+B_Z^\perp.
  \end{align*}
\end{lemma}
\begin{proof}
  By definition
  \begin{align*}
    C_X^\perp
    &= ((A_X\cap B_X)+B_Z^\perp)^\perp \\
    &= (A_X\cap B_X)^\perp\cap B_Z \\
    &= (A_X^\perp+B_X^\perp)\cap B_Z \\
    &= (A_X^\perp\cap B_Z)+B_X^\perp,
  \end{align*}
  and by the same reasoning $C_Z^\perp=(A_Z^\perp\cap B_X)+B_Z^\perp$. Then because $A_X^\perp\subseteq A_Z$, it follows that $C_X^\perp\subseteq C_Z$, so $\cC$ is a well defined CSS code.
\end{proof}

We then obtain the following.

\begin{lemma}
  \label{lem:qTBdef}
  For $r|(q-1)$ with $r\geq 3$ and $\ell\in[q]$, let
  \begin{align*}
    A &= \evl(\bF_q[X]^{[\ell]}) \\
    B &= \evl(\bF_q[X]^{[q-1]\setminus(-1+r\bZ)}).
  \end{align*}
  Then $A\cap B$ is a TB code. Furthermore,
  \begin{align}
    \label{eq:polyduals}
    \begin{split}
      A^\perp &= \evl(\bF_q[X]^{[q-\ell]\setminus\{0\}}) \\
      B^\perp &= \evl(\bF_q[X]^{[q-1]\cap(1+r\bZ)}) \subseteq B.
    \end{split}
  \end{align}
  If $\ell\geq q/2$, then $A^\perp\subseteq A$, so letting $C = (A\cap B)+B^\perp$,
  then $\CSS(C,C)$ is a qTB code with
  \begin{equation*}
    C^\perp = (A^\perp\cap B)+B^\perp = \evl(\bF_q[X]^T)\subseteq C.
  \end{equation*}
  for
  \begin{equation}
    \label{eq:qTBdualsupport}
    T = \{i\in[q-\ell]\setminus\{0\}:i\not\equiv r-1\pmod{r}\} \cup \{i\in[q-1]:i\equiv 1\pmod{r}\}.
  \end{equation}
\end{lemma}
\begin{proof}
  It suffices to show the equalities in~(\ref{eq:polyduals}), as then the remaining claims in the lemma statement follow from the definitions of TB and qTB codes along with Lemma~\ref{lem:intersect}.

  The first equality in~(\ref{eq:polyduals}) simply states the well known fact that the dual of a Reed-Solomon code is a (generalized) Reed-Solomon code. Formally, by dimension counting it suffices to show that $\evl(\bF_q[X]^{[q-\ell]\setminus\{0\}})\subseteq A^\perp$. By using monomials as a basis for $\bF_q[X]^{[\ell]}$ and for $\bF_q[X]^{[q-\ell]\setminus\{0\}}$, it then suffices to show that for every $i\in[\ell]$ and $j\in[q-\ell]\setminus\{0\}$, we have $\evl(X^i)\cdot\evl(X^j)=0$. But
  \begin{align}
    \label{eq:dotmonomials}
    \evl(X^i)\cdot\evl(X^j)
    &=\sum_{\alpha\in\bF_q^*}\alpha^{i+j},
  \end{align}
  which indeed equals $0$ for $i+j\not\equiv 0\pmod{q-1}$.

  We now prove the second equality in~(\ref{eq:polyduals}). By dimension counting, it suffices to show that $\evl(\bF_q[X]^{[q-1]\cap(1+r\bZ)})\subseteq B^\perp$. By using monomials as a basis for $\bF_q[X]^{[q-1]\setminus(-1+r\bZ)}$ and for $\bF_q[X]^{[q-1]\cap(1+r\bZ)}$, it then suffices to show that for every $i\in[q-1]\setminus(-1+r\bZ)$ and $j\in[q-1]\cap(1+r\bZ)$, we have $\evl(X^i)\cdot\evl(X^j)=0$. But then~(\ref{eq:dotmonomials}) again holds, and $i+j\not\equiv 0\pmod{q-1}$, so $\evl(X^i)\cdot\evl(X^j)=0$ as desired.
\end{proof}

Lemma~\ref{lem:qTBdef} implies that the qTB code $\cC=\CSS(C,C)$ satisfies $C^\perp\subseteq C$, so it is indeed a well defined CSS code. Below we compute the dimension of this code.

\begin{lemma}
  \label{lem:qTBdim}
  The qTB code $\cC=\CSS(C,C)$ with parameters $q,r,\ell$ has dimension
  \begin{align*}
    k
    &= 1+|\{q-\ell\leq i\leq\ell-1:i\not\equiv\pm 1\pmod{r}\}| \\
    &= (2\ell-q)\cdot\left(1-\frac{2}{r}\right)+\epsilon
  \end{align*}
  for some $\epsilon\in[-2,2]$.
\end{lemma}
\begin{proof}
  Define $S\subseteq[q-1]$ as in~(\ref{eq:qTBsupport}) and $T\subseteq S$ as in~(\ref{eq:qTBdualsupport}), so that $C=\bF_q[X]^S$ and $C^\perp=\bF_q[X]^T$. Then by definition $\cC$ has dimension
  \begin{align*}
    k
    &= \dim(C)-\dim(C^\perp) \\
    &= |S\setminus T| \\
    &= |\{0\}\cup\{q-\ell\leq i\leq\ell-1:i\not\equiv\pm 1\pmod{r}\}|,
  \end{align*}
  which proves the first equality in the lemma statement. The second equality in the lemma statement then follows because $|\{q-\ell\leq i\leq\ell-1:i\not\equiv\pm 1\pmod{r}\}|$ can differ from $(2\ell-q)(1-2/r)$ by at most $2$.
\end{proof}

We now show that the qTB code $\cC=\CSS(C,C)$ is locally recoverable. By Lemma~\ref{cor:CSSLRC}, it suffices to show that $C^\perp$ contains low-weight parity checks whose supports cover all $n$ code components. As $B^\perp\subseteq C^\perp$ (see Lemma~\ref{lem:qTBdef}), it in fact suffices to show that $B^\perp$ contains such low-weight parity checks, as is shown below.

\begin{lemma}
  \label{lem:piecelin}
  The code $B^\perp=\evl(\bF_q[X]^{[q-1]\cap(1+r\bZ)})$ as defined in Lemma~\ref{lem:qTBdef} may be equivalently characterized as follows. Let $\Omega_r=\{x\in\bF_q^*:x^r=1\}\subseteq\bF_q^*$ denote the $r$th roots of unity. Then $B^\perp$ consists of all functions $f:\bF_q^*\rightarrow\bF_q$ that can be expressed in the form $f(x)=\beta_{x\Omega_r}\cdot x$ for some $\beta\in\bF_q^{\bF_q^*/\Omega_r}$.
\end{lemma}

We refer to elements of $B^\perp$ as \textbf{piecewise linear functions}, as Lemma~\ref{lem:piecelin} states that $f\in B^\perp$ if and only if for each coset $x\Omega_r$ of $\Omega_r$, the function $f(x)$ agrees with the linear function $\beta_{x\Omega_r}\cdot x'$ on the restriction to inputs $x'\in x\Omega_r$. Note that these linear functions $x'\mapsto\beta_{x\Omega_r}\cdot x'$ are truly linear maps (not just affine linear).

\begin{proof}[Proof of Lemma~\ref{lem:piecelin}]
  By dimension counting, it suffices to show that every $\evl(f(X))\in B^\perp=\evl(\bF_q[X]^{[q-1]\cap(1+r\bZ)})$ is piecewise linear, that is, that $f(x)=\beta_{x\Omega_r}\cdot x$ for some $\beta\in\bF_q^{\bF_q^*/\Omega_r}$.

  Fix some $f(X)=\sum_{i\in[q-1]}f_iX^i$ with $\evl(f)\in B^\perp$, so that
  \begin{equation}
    \label{eq:finot1modr}
    f_i = 0 \hspace{.5cm} \forall i\not\equiv 1\pmod{r}.
  \end{equation}
  For every coset $\alpha\Omega_r\in\bF_q^*/\Omega_r$, then $X^r-\alpha^r$ vanishes on all $X\in\alpha\Omega_r$, so $g(X):=f(X)\pmod{X^r-\alpha^r}$ agrees with $f(X)$ on all $X\in\alpha\Omega_r$. But by definition $g(X)=\sum_{i\in[r]}g_iX^i$ with $g_i=\sum_{j\in[(q-1)/r]}f_{i+rj}\alpha^{rj}$, so~(\ref{eq:finot1modr}) implies that $g_i=0$ for all $i\neq 1$. Thus for all $X\in\alpha\Omega_r$ we have $f(X)=g(X)=g_iX$, so $f$ is indeed piecewise linear, as desired.
\end{proof}

Lemma~\ref{lem:piecelin} immediately yields the result of \cite{tamo_family_2014} that TB codes are cLRCs (Proposition~\ref{prop:tb}), and similarly shows that qTB codes are qLRCs:

\begin{corollary}
  \label{cor:qTBLRC}
  The qTB code given in Definition~\ref{def:qTBcode} is locally recoverable with locality $r$.
\end{corollary}
\begin{proof}
  By Lemma~\ref{cor:CSSLRC}, it suffices to show that the qTB code $\cC=\CSS(C,C)$ has $C$ (classically) locally recoverable with locality $r$. Now by Lemma~\ref{lem:qTBdef}, $C^\perp$ contains $B^\perp=\evl(\bF_q[X]^{[q-1]\cap(1+r\bZ)})$. For every $\alpha\in\bF_q^*$, Lemma~\ref{lem:piecelin} implies that $B^\perp$ contains the function
  \begin{equation*}
    f_\alpha(x) = \begin{cases}
      x,&x\in\alpha\Omega_r\\
      0,&x\notin\alpha\Omega_r,
    \end{cases}
  \end{equation*}
  which has $\supp(f)=\alpha\Omega_r$, that is, $f_\alpha$ is supported on $r$ components including $\alpha$. Thus Lemma~\ref{lem:cLRClin} implies that $C$ is a cLRC with locality $r$, so $\cC$ is a qLRC with locality $r$.
\end{proof}

\subsection{Folded Quantum Tamo-Barg Codes}
\label{sec:fqTBcode}
We are not able to show that the qTB codes in Definition~\ref{def:qTBcode} have distance approaching the Singleton-like bound. Rather, we will only obtain such near-optimal distance for the \textit{folded} version of these codes, defined below.

\begin{definition}
  \label{def:fqTBcode}
  As in Definition~\ref{def:qTBcode}, let $\cC=\CSS(C_X,C_Z)$ be the qTB code with parameters $q,r,\ell$, so that $C_X=C_Z=\evl(\bF_q[X]^S)$ for $S\subseteq[q-1]$ given by~(\ref{eq:qTBsupport}). Given an additional folding parameter $s|(q-1)/r$, we define the \textbf{folded quantum Tamo-Barg (fqTB) code} $\tilde{\cC}$ to be the quantum code of local dimension $q^s$ and block length $(q-1)/s$ obtained as follows. Fix a generator $\omega_{q-1}$ for $\bF_q^*$, and then for every $i\in[(q-1)/s]$, we block together the $s$ components (each of local dimension $q$) at positions $\{\omega_{q-1}^{si},\omega_{q-1}^{si+1},\dots,\omega_{q-1}^{si+s-1}\}$ in $\cC$ into a single component (of local dimension $q^s$) of the folded code $\tilde{\cC}$. Let $\tilde{C}_X,\tilde{C}_Z$ denote the $\bF_q$-linear codes obtained by similarly folding $C_X,C_Z$, so that $\tilde{\cC}=\CSS(\tilde{C}_X,\tilde{C}_Z)$.
\end{definition}

As a point of notation, we will almost universally use $\tilde{\alpha}$ to denote the $s$-folded version of an object $\alpha$, where $s$ will be clear from context. Specifically, we denote $\tilde{\bF}_q^*=\{\{\omega_{q-1}^{si},\dots,\omega_{q-1}^{si+s-1}\}:i\in[(q-1)/s]\}$ to be the partition of $\bF_q^*$ into the $(q-1)/s$ blocks of size $s$ into which we folded the qTB code in Defintion~\ref{def:fqTBcode}. For a function $a:\bF_q^*\rightarrow\bF_q$, we denote $\tilde{a}:\tilde{\bF}_q^*\rightarrow\bF_q^s$ the associated folded function given by $\tilde{a}(S)=(a(x))_{x\in S}$ for $S\in\tilde{\bF}_q^*$. For instance, if $a$ is a codeword of the classical code $C$ associated to a qTB code $\cC=\CSS(C,C)$, then $\tilde{a}$ is a codeword of the folded code $\tilde{C}$ associated to the fqTB code $\tilde{\cC}=\CSS(\tilde{C},\tilde{C})$. We also extend this notation to subspaces $V\subseteq\bF_q^{\bF_q^*}$, so that $\tilde{V}=\{\tilde{a}:a\in V\}$; thus indeed $\tilde{C}$ refers to the folded version of a code $C$, meaning that $\tilde{C}$ is obtained by replacing all codewords $a\in C$ with their folded version $\tilde{a}$. We furthermore define $\fevl:\bF_q[X]\rightarrow(\bF_q^s)^{\tilde{\bF}_q^*}$ by $\fevl(f)=\widetilde{\evl(f)}$, meaning that $\fevl$ for applies the evaluation map $\evl$, and then folds the output.

For a given $x=\omega_{q-1}^b\in\bF_q^*$, we let $\foldi{x}\in\tilde{\bF}_q^*$ denote the unique element of $\tilde{\bF}_q^*$ that contains $x$, so that
\begin{equation*}
  \foldi{x} = \{\omega_{q-1}^{\lfloor b/s\rfloor\cdot s+j}:j\in[s]\}.
\end{equation*}
For a polynomial $f(X)\in\bF_q[X]$, we will also let $\tilde{f}(\foldi{x})=\fevl(f)(\foldi{x})$. Therefore in particular if $\fevl(f)$ is a folded codeword, then $\tilde{f}(F_x)\in\bF_q^s$ denotes the component of this folded codeword associated to the partition element in $\tilde{\bF}_q^*$ that contains $x\in\bF_q^*$.

Folding a code by definition preserves the rate, so a folded qTB code has the same rate as the associated unfolded qTB code, which we computed in Lemma~\ref{lem:qTBdim}. Below, we use our proof in Lemma~\ref{lem:piecelin} and Corollary~\ref{cor:qTBLRC} that qTB codes are qLRCs to show that folded qTB codes are also qLRCs.

\begin{corollary}
  \label{cor:fqTBLRC}
  The fqTB code give in Definition~\ref{def:fqTBcode} is locally recoverable with locality $r$.
\end{corollary}
\begin{proof}
  Let $\cC=\CSS(C,C)$ denote the qTB code, and let $\tilde{\cC}=\CSS(\tilde{C},\tilde{C})$ denote the associated folded qTB code. Our proof of Lemma~\ref{lem:piecelin} and Corollary~\ref{cor:qTBLRC} implies that for every $x\in\bF_q^*$, component $x$ of $\cC$ can be recovered from the components at positions in $x(\Omega_r\setminus\{1\})$.

  Now for every $F_x\in\tilde{\bF}_q^*$, component $F_x$ of $\tilde{\cC}$ consists of the components $x'$ of $\cC$ for all $x'\in F_x$. As stated above, each such component $x'$ of $\cC$ can be recovered from the components in $x'(\Omega_r\setminus\{1\})$ of $\cC$. But by definition the union of all these recovery sets $x'(\Omega_r\setminus\{1\})$ for $x'\in F_x$ equals the union of the $r-1$ partition elements $F_{x\omega_r^i}$ for $\omega_r^i\in\Omega_r\setminus\{1\}$, that is,
  \begin{equation*}
    \bigsqcup_{x'\in F_x}x'(\Omega_r\setminus\{1\}) = F_x(\Omega_r\setminus\{1\}).
  \end{equation*}
  Therefore component $F_x$ of $\tilde{C}$ can be recovered from these $r-1$ components $F_{x\omega_r^i}$ for $\omega_r^i\in\Omega_r\setminus\{1\}$ by simply separately recovering the symbol at each position $x'\in F_x\subseteq\bF_q^*$ from the symbols at positions in $x'(\Omega_r\setminus\{1\})$ using the recovery maps for the unfolded qTB code $\cC$. Therefore $\tilde{\cC}$ is a qLRC of locality $r$, as desired.
\end{proof}

\section{Bounding the Distance}
\label{sec:distance}
In this section, we construct explicit qLRCs approaching the Singleton-like bound. We first bound the distance of the qTB codes introduced in Definition~\ref{def:qTBcode} in Section~\ref{sec:construct}. Interestingly, we are only able to get a distance bound for these codes that is reminiscent of the Johnson bound, so that this bound in particular only approaches the Singleton bound for codes of small rate. However, we then show that the folded version of these qTB codes introduced in Defintion~\ref{def:fqTBcode} have distance approaching the Singleton bound for all rates, for sufficiently large locality parameters $r$. It remains an open question whether the unfolded qTB codes have distance approaching the Singleton bound.

Note that classical Tamo-Barg codes are subcodes of Reed-Solomon codes, so their distance bound immediately follows from the distance of RS codes. That is, the distance of a classical TB code with parameters $q,r,\ell$ is at least the distance $d=q-\ell$ of a classical RS code with parameters $q,\ell$. However, our quantum TB codes are not subcodes of RS codes due to the orthogonality condition for CSS codes which necessitates extra high degree terms in the polynomials. As a result our distance analysis in the quantum case is much more involved and delicate. Indeed, the following corollary of Theorem~\ref{thm:singletonpart} shows that the distance of a qTB code with parameters $q,r,\ell$ must be strictly worse asymptotically than the distance $d=q-\ell$ of a qRS code with parameters $q,\ell$.

\begin{corollary}
  \label{cor:qTBupperdis}
  The distance $d$ of the qTB code $\cC=\CSS(C,C)$ with parameters $q,r,\ell$ satisfies
  \begin{equation*}
    d \leq \left(1-\frac{1}{r}\right)(q-\ell)+5.
  \end{equation*}
\end{corollary}
\begin{proof}
  Plugging the expression for the dimension $k$ of $\cC$ from Lemma~\ref{lem:qTBdim} into the bound in Theorem~\ref{thm:singletonpart} gives that
  \begin{align*}
    (2\ell-q)\left(1-\frac{2}{r}\right)-2
    &\leq \left(1-\frac{2}{r}\right)(q-1)-2\left(1-\frac{1}{r-1}\right)(d-1)+2.
  \end{align*}
  Rearranging terms in the above expression then yields the desired result.
\end{proof}

Similar reasoning as in Corollary~\ref{cor:qTBupperdis} also implies that the folded qTB code with parameters $q,r,\ell,s$ has asymptotically worse distance than the folded qRS code with parameters $q,\ell,s$, as the folded qRS code lies on the quantum Singleton bound like its unfolded counterpart.

\subsection{Distance of Unfolded Quantum Tamo-Barg Codes}
\label{sec:qTBdis}
We now show the following bound on the distance of the qTB codes from Definition~\ref{def:qTBcode}. Recall that the following result was stated informally as Theorem~\ref{thm:qTBdisinf} in Section~\ref{sec:qTBinf}.

\begin{theorem}
  \label{thm:qTBdis}
  The qTB code $\cC$ in Definition~\ref{def:qTBcode} with a prime locality parameter $r$ has distance at least
  \begin{equation}
    \label{eq:qTBdis}
    d = (q-1)\left(1-\frac{1}{2r}-\sqrt{\frac{1}{4r^2}+\frac{r-1}{r}\cdot\frac{\ell-1}{q-1}}\right).
  \end{equation}
\end{theorem}

The reader is referred to Section~\ref{sec:qTBoverviewinf} for more context surrounding the above bound and its apparent similarity to the Johnson bound.

\begin{proof}[Proof of Theorem~\ref{thm:qTBdis}]
  Let $\cC=\CSS(C,C)$ be a qTB code with parameters $q,r,\ell$. Fix an arbitrary $\evl(f)\in C\setminus C^\perp$. Our goal is to show that $|\evl(f)|\geq d$.

  For a high-level proof overview, the reader is referred to the proof sketch in Section~\ref{sec:qTBoverviewinf}. In brief, we will define a polynomial $G$ that has many roots associated to the roots of $f$, but also has low degree. We then obtain the desired result by comparing bounds on the number of roots and the degree of $G$.

  We now give the formal proof. As $C=\evl(\bF_q[X]^S)$ for $S=([\ell]\setminus(-1+r\bZ))\cup([q-1]\cap(1+r\bZ))$, we may write $f(X)=g(X)+h(X)$, where $g(X)\in\bF_q[X]^{[\ell]\setminus(\pm 1+r\bZ)}$ and $h(X)\in\bF_q[X]^{[q-1]\cap(1+r\bZ)}$. Then $h(X)$ is piecewise linear by Lemma~\ref{lem:piecelin}. As $\evl(h)\in C^\perp$ and by assumption $\evl(f)\notin C^\perp$, we must have $g\neq 0$.
  
  Define a polynomial $G(X)\in\bF_q[X]$ by
  \begin{equation}
    \label{eq:GJohnson}
    G(X) = \prod_{i=1}^{r-1}(\omega_r^{-i}g(\omega_r^iX)-g(X)).
  \end{equation}
  Letting $g(X)=\sum_jg_jX^j$, then by definition the $j$th coefficient of $\omega_r^{-i}g(\omega_r^iX)-g(X)$ equals $(\omega_r^{i(j-1)}-1)g_j$. Because $g_j=0$ for all $j\equiv 1\pmod{r}$, and $r$ is prime so that $\omega_r^{i(j-1)}\neq 1$ for all $j\not\equiv 1\pmod{r}$, it follows that $(\omega_r^{i(j-1)}-1)g_j=0$ iff $g_j=0$. Therefore in particular $\deg(\omega_r^{-i}g(\omega_r^iX)-g(X))=\deg(g)\leq\ell-1$, so $\deg(G)=(r-1)\deg(g)\leq(r-1)(\ell-1)$. It also follows that $G\neq 0$.

  We now bound the number of roots of $G$ in terms of the number of roots of $f$. If $f(x)=0$ and $f(\omega_r^ix)=0$ for some $x\in\bF_q^*$ and some $1\leq i\leq r-1$, then
  \begin{align*}
    \omega_r^{-i}g(\omega_r^ix)-g(x)
    &= -\omega_r^{-i}h(\omega_r^ix)+h(x) = -h(x)+h(x) = 0,
  \end{align*}
  where the second equality above holds because $h$ is piecewise linear.

  Therefore letting $\Omega_r=\{1,\omega_r,\dots,\omega_r^{r-1}\}$ denote the $r$th roots of unity in $\bF_q^*$, then for every ordered pair of distinct elements $(x,y=\omega_r^ix)$ within some coset $x\Omega_r$ such that $f(x)=f(y)=0$, the polynomial $G(X)$ has an associated root at $X=x$. Therefore $G(X)$ has $(r-|\evl(f)|_{x\Omega_r}|)(r-|\evl(f)|_{x\Omega_r}|-1)$ roots within a given coset $x\Omega_r$, where $|\evl(f)|_{x\Omega_r}|$ denotes the Hamming weight of the restriction of $\evl(f)$ to $x\Omega_r$. Thus the total number of roots of $G$ is at least
  \begin{align*}
    \hspace{1em}&\hspace{-1em} \sum_{x\Omega_r\in\bF_q^*/\Omega_r}(r-|\evl(f)|_{x\Omega_r}|)(r-|\evl(f)|_{x\Omega_r}|-1) \\
    &= \frac{q-1}{r}\bE_{x\Omega_r\sim\Unif(\bF_q^*/\Omega_r)}[(r-|\evl(f)|_{x\Omega_r}|)^2]-(q-1-|\evl(f)|) \\
    &\geq \frac{q-1}{r}\bE_{x\Omega_r\sim\Unif(\bF_q^*/\Omega_r)}[(r-|\evl(f)|_{x\Omega_r}|)]^2-(q-1-|\evl(f)|) \\
    &= \frac{r}{q-1}(q-1-|\evl(f)|)^2-(q-1-|\evl(f)|).
  \end{align*}
  But we showed above that $G$ is a nonzero polynomial of degree $\leq(r-1)(\ell-1)$, so we must have
  \begin{equation*}
    \frac{r}{q-1}(q-1-|\evl(f)|)^2-(q-1-|\evl(f)|) \leq (r-1)(\ell-1).
  \end{equation*}
  Solving this quadratic equation and rearranging terms gives that $|\evl(f)|\geq d$ for $d$ given in~(\ref{eq:qTBdis}), as desired.
\end{proof}

\subsection{Distance of Folded Quantum Tamo-Barg Codes}
\label{sec:fqTBdis}
In this section, we show the following bound on the distance of the fqTB codes from Definition~\ref{def:fqTBcode}.

\begin{theorem}
  \label{thm:fqTBdis}
  Let $\tilde{\cC}$ be the fqTB code in Definition~\ref{def:fqTBcode} with parameters $q,r,\ell,s$ such that $r$ is prime and such that the uncertainty principle in Proposition~\ref{prop:uncertainty} holds for $r$ over $\bF_q$. Then $\tilde{\cC}$ has distance at least
  \begin{equation}
    \label{eq:fqTBdis}
    d = \frac{q-1}{s}\left(1-\frac{\ell-1}{q-1}-\epsilon\right).
  \end{equation}
  for
  \begin{equation}
    \label{eq:fqTBloss}
    \epsilon = \max_{1\leq m\leq r}\min\left\{\left(1-\frac{\ell-1}{q-1}\right)\frac{m-1}{r},\frac{1}{m}+\frac{m-1}{s}\right\}.
  \end{equation}
\end{theorem}

Applying the technical Lemma~\ref{lem:boundloss} proved in Appendix~\ref{app:teclem} to bound the expression $\epsilon$ in~(\ref{eq:fqTBloss}) yields us the following corollary.

\begin{corollary}
  \label{cor:fqTBdisnice}
  Let $\tilde{\cC}$ be the fqTB code in Definition~\ref{def:fqTBcode} with parameters $q,r,\ell,s$ such that $r$ is prime and such that the uncertainty principle in Proposition~\ref{prop:uncertainty} holds for $r$ over $\bF_q$. If $s\geq 2r^2$, then $\tilde{\cC}$ has distance at least
  \begin{equation*}
    d \geq \frac{q-1}{s}\left(1-\frac{\ell-1}{q-1}-\left(1+\frac{r^2}{s}\right)\sqrt{\frac{1}{r}\left(1-\frac{\ell-1}{q-1}\right)}\right).
  \end{equation*}
\end{corollary}

Recall that an informal statement of Corollary~\ref{cor:fqTBdisnice} was given in Theorem~\ref{thm:fqTBdisinf} in Section~\ref{sec:qTBinf}, which intuitively states the following. Because folding a code preserves its rate, Lemma~\ref{lem:qTBdim} implies that as $q,r,\ell,s\rightarrow\infty$, Corollary~\ref{cor:fqTBdisnice} provides a bound on the relative distance of a fqTB code of rate $R$ that approches
\begin{equation*}
  \frac{d}{(q-1)/s} \rightarrow \frac{1-R}{2}.
\end{equation*}
This bound indeed approaches the ordinary Singleton bound (Proposition~\ref{prop:qSB}) for general quantum codes.

There is one subtlety with our bound from Theorem~\ref{thm:fqTBdis} and Corollary~\ref{cor:fqTBdisnice}, namely, that we require the uncertainty principle in Proposition~\ref{prop:uncertainty} to hold for $r$ over $\bF_q$. Proposition~\ref{prop:uncertainty} shows that for every fixed prime $r$, this uncertainty principle can only fail to hold for finite fields whose characteristic lies in a finite set depending on $r$. Thus we must first fix the locality parameter $r$, and then let the block length grow arbitrarily large. Specifically, by letting $q$ be an arbitrary power of an arbitrary prime outside of this bad finite set depending on $r$, we obtain our desired infinite family of qLRCs satisfying the distance bounds described above.


To prove Theorem~\ref{thm:fqTBdis}, we will use the following well known fact about determinants. For completeness, a proof is provided in Appendix~\ref{app:teclem}.

\begin{lemma}[Well known]
  \label{lem:detrank}
  For every $m\in\bN$, the determinant polynomial $\det\in\bF_q[(X_{ij})_{i,j\in[m]}]$ has a root of multiplicity $m-t$ at every matrix $(x_{ij})_{i,j\in[m]}$ of rank $t$. 
\end{lemma}

We will also use the following uncertainty principle over finite fields. Below we use the following notation for a finite field $\bF_q$ and a prime $r|(q-1)$. For a polynomial $f(X)=\sum_{i\in[q]}f_iX^i\in\bF_q[X]$, we let $|f|$ denote the number of $i\in[q]$ for which the coefficient $f_i$ is nonzero, while we let $|\evl(f)|_{\Omega_r}|$ denote the number of distinct $r$th roots of unity $\omega_r^i\in\Omega_r\subseteq\bF_q$ for which $f(\omega_r^i)$ is nonzero.

\begin{proposition}[\cite{goldstein_inequalities_2005}]
  \label{prop:uncertainty}
  For every fixed prime $r$, the following holds for all but finitely many primes $p$: if $\bF_q$ is a field of characteristic $p$ with $r|(q-1)$, then every nonzero $f\in\bF_q[X]$ of degree $<r$ satisfies
  \begin{equation*}
    |f|+|\evl(f)|_{\Omega_r}| \geq r+1.
  \end{equation*}
\end{proposition}

Proposition~\ref{prop:uncertainty} can be viewed as an uncertainty principle, as it says that if $f$ has few nonzero coefficients, then it has many nonzero evaluation points on $r$th roots of unity.

This uncertainty principle is implied by the statement that all minors of the Vandermonde matrix $(\omega_r^{ij})_{i,j\in[r]}$ are nonzero (a proof of this implication is given by \cite{tao_uncertainty_2004}). Over the complex field, this Vandermonde matrix also has no vanishing minors, as was first proved by Chebotar\"{e}v; additional proofs have since been given, e.g.~\cite{evans_generalized_1976,tao_uncertainty_2004,frenkel_simple_2004,goldstein_inequalities_2005}. The discussion in Section~6 of \cite{goldstein_inequalities_2005} shows that for every fixed $r$, Chebotar\"{e}v's theorem over $\bC$ implies the same result over finite fields of all but finitely many characteristics. Thus Proposition~\ref{prop:uncertainty} holds. We remark that a similar result is also shown in \cite{zhang_chebotarev_2019}.


\begin{proof}[Proof of Theorem~\ref{thm:fqTBdis}]
  Let $\cC=\CSS(C,C)$ be a qTB code with parameters $q,r,\ell$ as in the theorem statement, and let $\tilde{\cC}=\CSS(\tilde{C},\tilde{C})$ be the folding of $\cC$ for folding parameter $s|(q-1)/r$. Fix an arbitrary $f(X)=\sum_{i\in[q-1]}f_iX^i\in\bF_q[X]$ such that $\evl(f)\in C\setminus C^\perp$, with associated folded codeword $\fevl(f)\in\tilde{C}\setminus\tilde{C}^\perp$. Our goal is to show that $|\fevl(f)|\geq d$.

  Let $M=\{i\in[r]:\exists j\equiv i\pmod{r}\text{ with }f_j\neq 0\}$ and $m=|M|$. In the two claims below, we show two upper bounds on $|\fevl(f)|\geq d$; the first bound is tighter when $m$ is small, while the second is tighter when $m$ is large.
  \begin{claim}
    \label{claim:uncertainty}
    \begin{equation*}
      |\fevl(f)| \geq \frac{q-1}{s}\left(1-\frac{\ell-1}{q-1}\right)\left(1-\frac{m-1}{r}\right).
    \end{equation*}
  \end{claim}
  \begin{proof}
    We will show this bound on $|\fevl(f)|$ by first bounding $|\evl(f)|$, and then applying the fact that $|\fevl(f)|\geq|\evl(f)|/s$. To bound $|\evl(f)|$, we will use the uncertainty principle in Proposition~\ref{prop:uncertainty}. Recall that for every $x\in\bF_q^*$, then on the restriction to inputs in the coset $x\Omega_r$, $f$ agrees with $f\pmod{X^r-x^r}$, as $X^r-x^r$ vanishes on all $X\in x\Omega_r$. Now by definition $f\pmod{X^r-x^r}$ is a polynomial of degree $<r$ whose coefficients are supported within $M$. Therefore $|f\pmod{X^r-x^r}|\leq m$, so by Proposition~\ref{prop:uncertainty}, we have that either $\evl(f)|_{x\Omega_r}=0$ or $|\evl(f)|_{x\Omega_r}|\geq r+1-m$.

    We now bound the number of cosets $x\Omega_r\in\bF_q^*/\Omega_r$ for which $\evl(f)|_{x\Omega_r}=0$. If $M\subseteq\{1\}$ then $\evl(f)\in C^\perp$ by Lemma~\ref{lem:qTBdef}, contradicting the assumption that $\evl(f)\in C\setminus C^\perp$, so there must be some $i\in M\setminus\{1\}$. Then $\evl(f)|_{x\Omega_r}=0$ iff $f\pmod{X^r-x^r}=0$, which in turn can only occur if the $i$th coefficient of $f\pmod{X^r-x^r}$ equals $0$. But this $i$th coefficient is precisely $\sum_{j\in[(q-1)/r]}f_{i+rj}(x^r)^j$. Now the polynomial $\sum_{j\in[(q-1)/r]}f_{i+rj}Y^j$ has degree $\leq(\ell-1)/r$ because $f_{i+rj}=0$ if $i+rj\geq\ell$ by the definition of $C$. Thus this polynomial has $\leq(\ell-1)/r$ roots, so there are $\leq(\ell-1)/r$ values of $Y=x^r$ for which $\sum_{j\in[(q-1)/r]}f_{i+rj}(x^r)^j=0$, and thus there are $\leq(\ell-1)/r$ cosets $x\Omega_r\in\bF_q^*/\Omega_r$ for which $\evl(f)|_{x\Omega_r}=0$.

    Because we showed above that $|\evl(f)|_{x\Omega_r}|\geq r+1-m$ whenever $\evl(f)|_{x\Omega_r}\neq 0$, we may sum over all cosets $x\Omega_r$ to conclude that
    \begin{align*}
      |\evl(f)|
      &= \sum_{x\Omega_r\in\bF_q^*/\Omega_r}|\evl(f)|_{x\Omega_r}| \\
      &\geq \left(\frac{q-1}{r}-\frac{\ell-1}{r}\right)(r+1-m) \\
      &= (q-1)\left(1-\frac{\ell-1}{q-1}\right)\left(1-\frac{m-1}{r}\right).
    \end{align*}
    Then by definition
    \begin{equation*}
      |\fevl(f)| \geq \frac{|\evl(f)|}{s} \geq \frac{q-1}{s}\left(1-\frac{\ell-1}{q-1}\right)\left(1-\frac{m-1}{r}\right) \ . \qedhere
    \end{equation*}
  \end{proof}
  \begin{claim}
    \label{claim:determinant}
    \begin{equation*}
      |\fevl(f)| \geq \frac{q-1}{s}\left(1-\frac{\ell-1}{q-1}-\frac{1}{m}-\frac{m-1}{s}\right) \ . 
    \end{equation*}
  \end{claim}
  
  \begin{proof}
    We show this bound on $|\fevl(f)|$ using a similar method as in the proof of Theorem~\ref{thm:qTBdis}, except we replace the polynomial $G(X)$ in~(\ref{eq:GJohnson}) with~(\ref{eq:Gdet}) below. At a high level, we obtain our improved distance bound here by leveraging the code folding to more efficiently detect roots of $f$ using a certain determinant polynomial.

    Our use of the determinant polynomial to detect roots of $f$ is similar in spirit to analysis in \cite{guruswami_explicit_2016} of subspace designs based on folded RS codes. However, the details are different, and we require a more involved argument to show that the polynomial $G$ we obtain using the determinant is nonzero.
    

    As $C=\evl(\bF_q[X]^S)$ for $S=([\ell]\setminus(-1+r\bZ))\cup([q-1]\cap(1+r\bZ))$, we have a decomposition
    \begin{equation*}
      f(X)=g(X)+h(X)
    \end{equation*}
    for $g(X)=\sum_ig_iX^i\in\bF_q[X]^{[\ell]\setminus(\pm 1+r\bZ)}$ and $h(X)=\sum_ih_iX^i\in\bF_q[X]^{[q-1]\cap(1+r\bZ)}$. Because $\evl(f)\notin C^\perp$ and $\evl(h)\in C^\perp$ by assumption, we must have $g\neq 0$.
    
    Let $\det\in\bF_q[(X_{ij})_{i,j\in[m]}]$ denote the determinant polynomial, which takes as input an $m\times m$ matrix of variables $(X_{ij})_{i,j\in[m]}$ over $\bF_q$, and outputs the determinant of the input matrix. Define $G(X)\in\bF_q[X]$ by
    \begin{equation}
      \label{eq:Gdet}
      G(X) = \det(\omega_r^{-i}\cdot g(\omega_r^i\omega_{q-1}^j\cdot X))_{i,j\in[m]}.
    \end{equation}
    As $\deg(\det)=m$, we have $\deg(G)\leq m\cdot\deg(g)\leq m(\ell-1)$.

    We next show that $G(X)$ is a nonzero polynomial. As $g(X)=\sum_{a\in[\ell]\setminus(\pm 1+r\bZ)}g_aX^a$, we may decompose the matrix
    \begin{align}
      \begin{split}
        \label{eq:rank1decomp}
        A(X)
        &:= \left(\omega_r^{-i}\cdot g(\omega_r^i\omega_{q-1}^j\cdot X)\right)_{i,j\in[m]} \\
        &= \sum_ag_aX^a\left(\omega_r^{(a-1)i}\omega_{q-1}^{aj}\right)_{i,j\in[m]} \\
        &= \sum_ag_aX^a\begin{pmatrix}1\\\omega_r^{a-1}\\\omega_r^{2(a-1)}\\\vdots\\\omega_r^{(m-1)(a-1)}\end{pmatrix}\cdot\begin{pmatrix}1&\omega_{q-1}^a&\omega_{q-1}^{2a}&\cdots&\omega_{q-1}^{(m-1)a}\end{pmatrix} \\
        &= \sum_{i\in M}\begin{pmatrix}1\\\omega_r^{i-1}\\\omega_r^{2(i-1)}\\\vdots\\\omega_r^{(m-1)(i-1)}\end{pmatrix}\cdot\sum_{a\equiv i\pmod{r}}g_aX^a\begin{pmatrix}1&\omega_{q-1}^a&\omega_{q-1}^{2a}&\cdots&\omega_{q-1}^{(m-1)a}\end{pmatrix}.
      \end{split}
    \end{align}
    By definition $A(X)$ is an $m\times m$ matrix with entries in the ring $\bF_q[X]$, and $G(X)=\det(A(X))$. To show that $G\neq 0$, it therefore suffices to show that $A$ has full rank $m$, for which it in turn suffices to show that the sets of vectors
    \begin{equation}
      \label{eq:decompvecs}
      \left\{\begin{pmatrix}1\\\omega_r^{i-1}\\\omega_r^{2(i-1)}\\\vdots\\\omega_r^{(m-1)(i-1)}\end{pmatrix}:i\in M\right\} \hspace{1em}\text{ and }\hspace{1em} \left\{\sum_{a\equiv i\pmod{r}}g_aX^a\begin{pmatrix}1\\\omega_{q-1}^a\\\omega_{q-1}^{2a}\\\vdots\\\omega_{q-1}^{(m-1)a}\end{pmatrix}:i\in M\right\}
    \end{equation}
    from the rank-1 decomposition of $A$ in~(\ref{eq:rank1decomp}) are linearly independent over $\bF_q[X]$. Here it is equivalent to show linear independence over the field of rational functions $\bF_q(X)$ and (as we may clear denominators in any linear dependency) over the ring of polynomials $\bF_q[X]$, so we will show the latter.

    Now the first set of vectors in~(\ref{eq:decompvecs}) above form the columns of an $m\times m$ Vandermonde matrix, which is known to have full rank. Meanwhile, if there is some nontrivial $\bF_q[X]$-linear dependency among the second set of vectors in~(\ref{eq:decompvecs}), then taking the highest-degree term of the associated polynomials over $X$ gives a nontrivial $\bF_q$-linear dependency among the vectors $(\omega_{q-1}^{aj})_{j\in[m]}$ for $m$ distinct values of $a$. But these vectors again form the columns of an $m\times m$ Vandermonde matrix and thus cannot have any nontrivial linear dependencies. Therefore both sets of vectors in~(\ref{eq:decompvecs}) are indeed linearly independent, so the decomposition in~(\ref{eq:rank1decomp}) expressing $A(X)$ as the sum of $m$ rank-1 matrices implies that $A(X)$ has full rank $m$, and thus $G(X)=\det(A(X))$ is a nonzero polynomial.

    We now bound the number of roots of $G(X)$ in terms of the Hamming weight of $\fevl(f)$. For a given $x=\omega_{q-1}^b\in\bF_q^*$, recall that $\foldi{x}=\{\omega_{q-1}^{\lfloor b/s\rfloor\cdot s+j}:j\in[s]\}\in\tilde{\bF}_q^*$ denotes the index of the folded component of $\tilde{\cC}$ that contains the component $x$ of $\cC$.

    If $b\pmod{s}\in\{0,1,\dots,s-m\}$, then because $s|(q-1)/r$ by assumption, it holds for each $i\in[m]$ that
    \begin{equation*}
      \{\omega_r^i\omega_{q-1}^jx:j\in[m]\} = \{\omega_{q-1}^{b+i(q-1)/r+j}x:j\in[m]\} \subseteq \foldi{\omega_r^ix}.
    \end{equation*}
    Thus the $i$th row of the matrix $(\omega_r^i\omega_{q-1}^jx)_{i,j\in[m]}$ consists of $m$ elements all lying inside $\foldi{\omega_r^ix}$. Applying $f$ to each component, it follows that that $i$th row of $(f(\omega_r^i\omega_{q-1}^j))_{i,j\in[m]}$ consists of $m$ out of the $s$ components in the folded symbol $\fevl(f)_{\foldi{\omega_r^ix}}$.

    Let $Z_x = \{i\in[m]:\fevl(f)_{\foldi{\omega_r^ix}}=0\}$. If $b\pmod{s}\in\{0,1,\dots,s-m\}$, then for every $i\in Z_x$, the $i$th row of $(f(\omega_r^i\omega_{q-1}^j))_{i,j\in[m]}$ by definition consists of all zeros. In this case because $f(X)=g(X)+h(X)$, for every $i\in Z_x$ we have
    \begin{align*}
      (\omega_r^{-i}g(\omega_r^i\omega_{q-1}^jx))_{j\in[m]}
      &= (-\omega_r^{-i}h(\omega_r^i\omega_{q-1}^jx))_{j\in[m]} = (-h(\omega_{q-1}^jx))_{j\in[m]},
    \end{align*}
    where the second equality holds because $h(X)$ is by definition piecewise linear. Now the left hand side above equals the $i$th row of $A(x)\in\bF_q^{m\times m}$, while the right hand side is a vector that does not depend on $i$. Thus we have shown that if $b\pmod{s}\in\{0,1,\dots,s-m\}$, then the $i$th row of $A(x)$ is the same for all $i\in Z_x$. Therefore in this case $A(x)$ has rank $\leq m+1-|Z_x|$, so Lemma~\ref{lem:detrank} implies that the $m\times m$ determinant polynomial has a root of multiplicity $\geq|Z_x|-1$ at $A(x)$. Thus $G(X)=\det(A(X))$ has a root of multiplicity $\geq|Z_x|-1$ at $X=x$.

    Summing over all $x=\omega_{q-1}^b\in\bF_q^*$, it follows that the number of roots (including multiplicities) of $G(X)$ is at least
    \begin{align*}
      \sum_{b\in[q-1]:b\pmod{s}\in\{0,\dots,s-m\}}(|Z_{\omega_{q-1}^b}|-1)
      &\geq \sum_{x\in\bF_q^*}(|Z_x|-1)-\frac{(q-1)m}{s}\cdot(m-1) \\
      &= \sum_{x\in\bF_q^*}|Z_x|-(q-1)\left(\frac{m(m-1)}{s}+1\right) \\
      &= (q-1-|\fevl(f)|\cdot s)\cdot m-(q-1)\left(\frac{m(m-1)}{s}+1\right) \\
      &= (q-1)m\left(1-\frac{|\fevl(f)|\cdot s}{q-1}-\frac{1}{m}-\frac{m-1}{s}\right)
    \end{align*}
    where the second equality above holds becuase $|Z_x|$ equals the number of elements $x'\in x\Omega_r$ for which $\fevl(f)_{\foldi{x'}}=0$, so each of the $q-1-|\fevl(f)|\cdot s$ values $x'\in\bF_q^*$ with $\fevl(f)_{\foldi{x'}}=0$ contributes $1$ to the sum $\sum_{x\in\bF_q^*}|Z_x|$ a total of $m$ times, once for each $x\in x'\Omega_r$.

    But we showed above that $G$ is a nonzero polynomial of degree $\leq(\ell-1)m$, so $G$ has $\leq(\ell-1)m$ roots, and thus
    \begin{equation*}
      (q-1)m\left(1-\frac{|\fevl(f)|\cdot s}{q-1}-\frac{1}{m}-\frac{m-1}{s}\right) \leq (\ell-1)m.
    \end{equation*}
    Rearranging terms above then gives that
    \begin{equation*}
      |\fevl(f)| \geq \frac{q-1}{s}\left(1-\frac{\ell-1}{q-1}-\frac{1}{m}-\frac{m-1}{s}\right) \ . \qedhere
    \end{equation*}
  \end{proof}
  Combining the bounds in the two claims above immediately implies that $|\fevl(f)|\geq d$ for $d$ defined in~(\ref{eq:fqTBdis}), as desired.

\end{proof}

\section{Efficient Decoding Algorithm}
\label{sec:decode}
This section presents an efficient decoding algorithm for the qLRCs constructed in Section~\ref{sec:construct}. Note that the local recovery algorithms $\Rec_i$ are by construction efficient, as they are a special case of erasure-decoding for a CSS code. Here we present an efficient algorithm for the more difficult task of globally decoding from errors at a number of unknown locations approaching half the code distance. As in Section~\ref{sec:construct}, we consider both unfolded and folded quantum Tamo-Barg codes. For the unfolded codes our decoding algorithm is relatively simple but suboptimal. For the folded codes, our decoding algorithm approaches the optimal possible decoding radius as the locality $r$ grows large.

Classical Tamo-Barg codes are subcodes of Reed-Solomon codes, so a classical Tamo-Barg code can be efficiently decoded by simply running an efficient Reed-Solomon decoder, which was known to exist. However, our (folded) quantum Tamo-Barg codes are not subcodes of quantum Reed-Solomon codes due to duality conditions for CSS codes. Just as this distinction between the quantum and classical cases made our distance analysis in Section~\ref{sec:distance} much more complicated than that of classical TB codes, our decoding algorithm for the quantum case is also more involved, though it does eventually reduce to Reed-Solomon decoding.

Note that in classical distributed storage applications of LRCs, errors often correspond to phenomena like a server being unresponsive, and therefore occur at known locations in the codeword. In such a case, it suffices to be able to decode from erasures, which can be done efficiently for any linear code using Gaussian elimination. Similarly, quantum stabilizer codes, of which our construction in Section~\ref{sec:construct} is an instance, can be efficiently decoded from erasures.

However, it is not implausible that decoding from errors in unknown locations would be useful in applications of quantum LRCs. For instance, even if a qLRC is using in a distributed setting where each code component is itself stored in a fault-tolerant memory, because quantum states are naturally more error-prone than classical data, it is not unreasonable that occasional global error correction could be beneficial to maintaining high fidelity for the encoded state, even if most errors are caused by local failures in known locations.

\subsection{Unfolded Quantum Tamo-Barg Codes}
\label{sec:qTBdec}
In this section we present an efficient decoding algorithm for the qTB codes in Definition~\ref{def:qTBcode}. Our algorithm is motivated by the proof of Theorem~\ref{thm:qTBdis}, and ultimately reduces decoding the qTB codes to list decoding a classical Reed-Solomon code.

In particular, we apply Reed-Solomon list decoding up to the Johnson bound (Theorem~\ref{thm:RSdec} to obtain an efficient list decoder for our qTB codes, and then apply the qTB distance bound in Theorem~\ref{thm:qTBdis} to show that when the error is not too large, the list decoder in fact serves as a unique decoder, that is, there are no extraneous list elements.

\begin{algorithm}[h]
  \caption{Classical decoding algorithm $\Dec_C$ for a qTB code $\cC=\CSS(C,C)$ with parameters $q,r,\ell$. By Proposition~\ref{prop:CSSdec}, $\Dec_C$ can be used to obtain an efficient quantum decoding algorithm for $\cC$ from $\leq e$ errors in unknown locations. As a point of notation, $\ListDec_{\mathrm{RS}(q,\ell)}$ refers to the Reed-Solomon list decoder in Theorem~\ref{thm:RSdec}, whose output is a list of message polynomials. Also, we let $B^\perp=\evl(\bF_q[X]^{[q-1]\cap(1+r\bZ)})$ as in~(\ref{eq:polyduals}) denote the space of piecewise linear functions (see Lemma~\ref{lem:piecelin}).}
  \label{alg:qTBdec}
  \SetKwInOut{Input}{Input}
  \SetKwInOut{Output}{Output}
  
  \SetKwFunction{FnDec}{$\Dec_C$}
  \SetKwFunction{FnListDecRS}{$\ListDec_{\mathrm{RS}(q,\ell)}$}
  \SetKwProg{Fn}{Function}{:}{}

  \Input{Corrupted codeword $a:\bF_q^*\rightarrow\bF_q$ with $\dis(a,C)\leq e$}
  \Output{$c'\in C$ such that $\dis(c'-a,C^\perp)\leq e$}
  
  \Fn{\FnDec{$a$}}{
    $\cL\gets\emptyset$ \\
    \For{$i\in\{1,\dots,r-1\}$}{
      Define $a^{(i)}:\bF_q^*\rightarrow\bF_q$ by $a^{(i)}(x)=\omega_r^{-i}a(\omega_r^ix)-a(x)$ \\
      $\cL_i\gets\FnListDecRS{$a^{(i)}$}$ \\
      \For{$g^{(i)}(X)=\sum_{j\in[\ell]}g_j^{(i)}X^j\in\cL_i$}{
        \If{$g_j^{(i)}=0$ for every $j\equiv \pm 1\pmod{r}$}{
          Add $\evl\left(g(X):=\sum_{j\in[\ell]}(\omega_r^{(j-1)i}-1)^{-1}g_j^{(i)}X^j\right)$ to $\cL$
        }
      }
    }
    \KwRet{$\argmin_{\evl(g)\in\cL}\dis(\evl(g)-a,B^\perp)$}
  }
\end{algorithm}

\begin{theorem}
  \label{thm:qTBdec}
  Let $\cC$ be the qTB code in Definition~\ref{def:qTBcode} with a prime locality parameter $r$. Then $\cC$ can be decoded from errors in up to
  \begin{equation*}
    e = (q-1) \cdot \frac12\left(1-\frac{1}{2r}-\sqrt{\frac{1}{4r^2}+\frac{r-1}{r}\cdot\frac{\ell}{q-1}}\right)
  \end{equation*}
  unknown locations in time $q^{O(1)}$.
\end{theorem}

Note that the error bound $e$ in Theorem~\ref{thm:qTBdec} above is just under half the distance bound $d$ from Theorem~\ref{thm:qTBdis}; replacing the $\ell$ above with $\ell-1$ gives the expression for $d/2$. Thus Theorem~\ref{thm:qTBdec} shows that we can efficiently decode from adversarial errors up to nearly half the distance, which is optimal.


To show that Algorithm~\ref{alg:qTBdec} runs in polynomial time, we need the following lemma, which shows that the distance computations $\dis(\evl(g)-a,B^\perp)$ in the final line of the algorithm can be performed efficiently.

\begin{lemma}
  \label{lem:distoperp}
  Letting $B^\perp=\evl(\bF_q[X]^{[q-1]\cap(1+r\bZ)})$, there exists a $O(q\poly\log q)$-time algorithm that takes as input $b:\bF_q^*\rightarrow\bF_q$, and outputs $\dis(b,B^\perp)$.
\end{lemma}
\begin{proof}
  To compute $\dis(b,B^\perp)=\min_{b'\in B^\perp}|b-b'|$ efficiently, recall that $B^\perp$ is the space of piecewise linear functions by Lemma~\ref{lem:piecelin}. Therefore we may simply find the closest linear function to $b$ within each coset $x\Omega_r\in\bF_q^*/\Omega_r$ separately. Formally, for each coset $x\Omega_r\in\bF_q^*/\Omega_r$, we may compute $\beta_{x\Omega_r}:=\argmax_{\beta\in\bF_q}|\{x'\in x\Omega_r:b(x')=\beta\cdot x'\}|$, and then we set $b'(x')=\beta_{x\Omega_r}x'$ for all $x'\in x\Omega_r$. Then by construction $\dis(b,B^\perp)=\dis(b,b')$.

  Note that above each $\beta_{x\Omega_r}$ can be computed in time $O(r\poly\log q)$, as we may simply set $\beta_{x\Omega_r}$ equal to the mode value of $\beta=b(x')/x'$ across all $x'\in x\Omega_r$. Thus $\dis(b,B^\perp)$ may be computed in time $O(q\poly\log q)$.
\end{proof}

\begin{proof}[Proof of Theorem~\ref{thm:qTBdec}]
  Let $\cC=\CSS(C,C)$ be a qTB code with parameters $q,r,\ell$. By Proposition~\ref{prop:CSSdec}, it suffices to construct an algorithm $\Dec_C$ that takes as input a corrupted codeword $a=c+b$ for some $c\in C$ and some corruption $b\in\bF_q^*\rightarrow\bF_q$ of Hamming weight $|b|\leq e$, and outputs some $c'\in C$ such that $c'-c\in C^\perp$.

  The desired algorithm is given in Algorithm~\ref{alg:qTBdec}. We first show it correctly decodes as described above, and then analyze the running time.

  Consider a corrupted codeword $a=c+b$ for some codeword $c=\evl(f)\in C$ and some corruption $b:\bF_q^*\rightarrow\bF_q$ of weight $|b|\leq e$. For a given $1\leq i\leq r-1$, by definition $a^{(i)}(x)=\omega_r^{-i}a(\omega_r^ix)-a(x)$ equals $c^{(i)}(x)=f^{(i)}(x)=\omega_r^{-i}f(\omega_r^ix)-f(x)$ at every point $x$ for which $x,\omega_r^ix\notin\supp(b)$. Meanwhile, within a given coset $x\Omega_r\in\bF_q^*/\Omega_r$, the number of ordered pairs of distinct points $y,y'\in x\Omega_r$ such that $y,y'\notin\supp(b)$ is precisely $(r-|b|_{x\Omega_r}|)(r-|b|_{x\Omega_r}-1|)$. Thus the sum over all $i\in\{1,\dots,r-1\}$ of the number of points $x\in\bF_q^*$ where $a^{(i)}(x)=f^{(i)}(x)$ satisfies
  \begin{align*}
    \sum_{i=1}^{r-1}|\{x\in\bF_q^*:a^{(i)}(x)=f^{(i)}(x)\}|
    &\geq \sum_{i=1}^{r-1}|\{x\in\bF_q^*:x,\omega_r^ix\notin\supp(b)\}| \\
    &= \sum_{x\Omega_r\in\bF_q^*/\Omega_r}|\{(y,y')\in(x\Omega_r)^2:y\neq y'\text{ and }y,y'\notin\supp(b)\}| \\
    &= \sum_{x\Omega_r\in\bF_q^*/\Omega_r}(r-|b|_{x\Omega_r}|)(r-|b|_{x\Omega_r}-1|) \\
    &\geq \frac{q-1}{r}\bE_{x\Omega_r\sim\Unif(\bF_q^*/\Omega_r)}[(r-|b|_{x\Omega_r}|)^2]-(q-1-|b|) \\
    &\geq \frac{q-1}{r}\bE_{x\Omega_r\sim\Unif(\bF_q^*/\Omega_r)}[(r-|b|_{x\Omega_r}|)]^2-(q-1-|b|) \\
    &= \frac{r}{q-1}(q-1-|b|)^2-(q-1-|b|) \\
    &\geq \frac{r}{q-1}(q-1-e)^2-(q-1-e),
  \end{align*}
  where the final inequality above holds becuase the function $\beta\mapsto r(q-1-\beta)^2/(q-1)-(q-1-\beta)$ is decreasing for all $\beta\leq e\leq(q-1)/2$. Then averaging over all $i\in\{1,\dots,r-1\}$, there must be some such $i$ for which
  \begin{equation}
    \label{eq:agreef}
    |\{x\in\bF_q^*:a^{(i)}(x)=f^{(i)}(x)\}| \geq \frac{r(q-1-e)^2}{(r-1)(q-1)}-\frac{q-1-e}{r-1}.
  \end{equation}

  Recall by Lemma~\ref{lem:qTBdef}, we may decompose $f(X)=g(X)+h(X)$ for $g(X)=\sum_jg_jX^j\in\bF_q[X]^{[\ell]\setminus(\pm 1+r\bZ)}$ and $h(X)\in\bF_q[X]^{[q-1]\cap(1+r\bZ)}$. Define $g^{(i)},h^{(i)}$ analogouosly to $f^{(i)}$, so that $f^{(i)}=g^{(i)}+h^{(i)}$. Then because
  \begin{equation}
    \label{eq:coeffchange}
    \omega_r^{-i}(\omega_r^iX)^j-X^j = (\omega_r^{i(j-1)\pmod{r}}-1)X^j
  \end{equation}
  equals $0$ for all $j\equiv 1\pmod{r}$, it follows that $h^{(i)}=0$ and thus $f^{(i)}=g^{(i)}$. Therefore~(\ref{eq:agreef}) is equivalent to
  \begin{equation}
    \label{eq:agreeg}
    |\{x\in\bF_q^*:a^{(i)}(x)=g^{(i)}(x)\}| \geq \frac{r(q-1-e)^2}{(r-1)(q-1)}-\frac{q-1-e}{r-1}.
  \end{equation}
  But by~(\ref{eq:coeffchange}), the coefficients of $g^{(i)}(X)=\sum_{j\in[\ell]\setminus(\pm 1+r\bZ)}g_j^{(i)}X^j$ are given by $g_j^{(i)}=(\omega_r^{i(j-1)\pmod{r}}-1)g_j$, and thus $g$ and $g^{(i)}$ have coefficients of the same support, that is, $g_j=0$ if and only if $g_j^{(i)}=0$. In particular, it follows that $\deg(g^{(i)})=\deg(g)<\ell$. In other words, $\evl(g^{(i)})\in\mathrm{RS}(q,\ell)$ is a Reed-Solomon codeword, so~(\ref{eq:agreeg}) says that $a^{(i)}=\evl(g^{(i)})+b^{(i)}$ is a corrupted Reed-Solomon codeword with the corruption $b^{(i)}$ of weight
  \begin{align*}
    |b^{(i)}|
    &= q-1 - |\{x\in\bF_q^*:a^{(i)}(x)=g^{(i)}(x)\}| \\
    &\leq q-1 - \left(\frac{r(q-1-e)^2}{(r-1)(q-1)}-\frac{q-1-e}{r-1}\right) \\
    &= (q-1)\left(1-\left(\frac{r}{r-1}\left(1-\frac{e}{q-1}\right)^2-\frac{1}{r-1}\left(1-\frac{e}{q-1}\right)\right)\right).
  \end{align*}
  By Theorem~\ref{thm:RSdec}, the output of running the list decoder $\ListDec_{\mathrm{RS}(q,\ell)}(a^{(i)})$ is a list containing $\evl(g^{(i)})$ as long as the right hand side above is at most $(q-1)(1-\sqrt{\ell/(q-1)})$, which simplifies to needing that
  \begin{equation*}
    e \leq (q-1)\left(1-\frac{1}{2r}-\sqrt{\frac{1}{4r^2}+\frac{r-1}{r}\cdot\sqrt{\frac{\ell}{q-1}}}\right).
  \end{equation*}
  But the above holds by the definition of $e$ along with Lemma~\ref{lem:twojohnsons}, so the list $\cL_i$ in Algorithm~\ref{alg:qTBdec} will contain $g^{(i)}$, and thus after the $i$th iteration of the for loop, the list $\cL$ will contain $\evl(g(X))=\evl(\sum_j(\omega_r^{(j-1)i}-1)^{-1}g_j^{(i)}X^j)$. Thus if Algorithm~\ref{alg:qTBdec} outputs $\evl(g')\in\cL$, then $g'\in\bF_q[X]^{[\ell]\setminus(\pm 1+r\bZ)}$ and
  \begin{equation*}
    \dis(\evl(g')-a,B^\perp) \leq \dis(\evl(g)-a,B^\perp) \leq |\evl(g+h)-a| = |\evl(f)-a| = |b| \leq e.
  \end{equation*}
  But as $B^\perp\subseteq C^\perp$, it follows that $\dis(\evl(g)-a,C^\perp) \leq e$ and $\dis(\evl(g')-a,C^\perp) \leq e$, so $\dis(\evl(g)-\evl(g'),C^\perp)\leq 2e$. But by the definition of $e$ along with Theorem~\ref{thm:qTBdis}, $\cC$ has distance $\min_{c\in C\setminus C^\perp}|c|>2e$, 
  so it follows that $\evl(g)-\evl(g')\in C^\perp$. Thus $\Dec_C(a)$ outputs some $\evl(g')\in\evl(g)+C^\perp$, as desired.

  It remains to show that Algorithm~\ref{alg:qTBdec} runs in time $q^{O(1)}$. Theorem~\ref{thm:RSdec} implies that the calls to $\ListDec_{\mathrm{RS}(q,\ell)}$ run in $q^{O(1)}$ time, while Lemma~\ref{lem:distoperp} implies that the $\argmin$ computation in the final line of the algorithm runs in $O(|\cL|\cdot q\poly\log q)=q^{O(1)}$ time. The rest of the algorithm by definition also runs in $q^{O(1)}$ time, so the result follows.
\end{proof}

\subsection{Folded Quantum Tamo-Barg Codes}
We now present an efficient decoding algorithm for the fqTB codes in Definition~\ref{def:fqTBcode}. The algorithm is similar to the algorithm presented in Section~\ref{sec:qTBdec} for the unfolded qTB codes, and decodes up to an error fraction approaching half our distance bound in Theorem~\ref{thm:fqTBdis}. Yet because the distance bound in Theorem~\ref{thm:fqTBdis} approaches the quantum Singleton bound as the locality parameter $r$ grows large, our decoding algorithm for fqTB codes for large $r$ tolerates error fractions approaching half the quantum Singleton, which is optimal.

We obtain this improved error tolerance for fqTB decoding compared to qTB decoding by leveraging the folding in two separate ways, namely, through the improved distance of fqTB codes (Theorem~\ref{thm:fqTBdis}) compared to qTB codes (Theorem~\ref{thm:qTBdis}), and through the improved list-decodability of classical fRS codes (Theorem~\ref{thm:fRSdec}) compared to RS codes (Theorem~\ref{thm:RSdec}). Combining these improvements for folded codes, but otherwise following the proof of Theorem~\ref{thm:qTBdec}, we obtain the following result.

\begin{algorithm}[h]
  \caption{Classical decoding algorithm $\Dec_{\tilde{C}}$ for a fqTB code $\tilde{\cC}=\CSS(\tilde{C},\tilde{C})$ with parameters $q,r,\ell,s$. By Proposition~\ref{prop:CSSdec}, $\Dec_{\tilde{C}}$ can be used to obtain an efficient quantum decoding algorithm for $\tilde{\cC}$ from $\leq e$ errors in unknown locations. As a point of notation, $\ListDec_{\mathrm{fRS}(q,\ell)}$ refers to the folded Reed-Solomon list decoder in Theorem~\ref{thm:fRSdec}, whose output is a list of message polynomials. Also, we let $\tilde{B}^\perp=\fevl(\bF_q[X]^{[q-1]\cap(1+r\bZ)})$ as in~(\ref{eq:polyduals}) denote the folded space of piecewise linear functions (see Lemma~\ref{lem:piecelin}).}
  \label{alg:fqTBdec}
  \SetKwInOut{Input}{Input}
  \SetKwInOut{Output}{Output}
  
  \SetKwFunction{FnDec}{$\Dec_{\tilde{C}}$}
  \SetKwFunction{FnListDecRS}{$\ListDec_{\mathrm{fRS}(q,\ell,s)}$}
  \SetKwProg{Fn}{Function}{:}{}

  \Input{Corrupted codeword $\tilde{a}\in\widetilde{\bF_q^{\bF_q^*}}$ with $\dis(\tilde{a},\tilde{C})\leq e$}
  \Output{$\tilde{c}'\in\tilde{C}$ such that $\dis(\tilde{c}'-\tilde{a},\tilde{C}^\perp)\leq e$}
  
  \Fn{\FnDec{$\tilde{a}$}}{
    $\cL\gets\emptyset$ \\
    \For{$i\in\{1,\dots,r-1\}$}{
      Define $a^{(i)}:\bF_q^*\rightarrow\bF_q$ by $a^{(i)}(x)=\omega_r^{-i}a(\omega_r^ix)-a(x)$ \\
      $\cL_i\gets\FnListDecRS{$\tilde{a}^{(i)}$}$ \\
      \For{$g^{(i)}(X)=\sum_{j\in[\ell]}g_j^{(i)}X^j\in\cL_i$}{
        \If{$g_j^{(i)}=0$ for every $j\equiv \pm 1\pmod{r}$}{
          Add $\fevl(g)$ to $\cL$ for $g(X):=\sum_{j\in[\ell]}(\omega_r^{(j-1)i}-1)^{-1}g_j^{(i)}X^j$
        }
      }
    }
    \KwRet{$\argmin_{\fevl(g)\in\cL}\dis(\fevl(g)-\tilde{a},\tilde{B}^\perp)$}
  }
\end{algorithm}

\begin{theorem}
  \label{thm:fqTBdec}
  Let $\tilde{\cC}=\CSS(\tilde{C},\tilde{C})$ be the fqTB code in Definition~\ref{def:fqTBcode} with parameters $q,r,\ell,s$ such that $r$ is prime and such that the uncertainty principle in Proposition~\ref{prop:uncertainty} holds for $r$ over $\bF_q$. Then $\cC$ can be decoded from errors in up to
  \begin{equation}
    \label{eq:fqTBdecradius}
    e = \min\left\{\frac{d}{2}-1,e'\right\}
  \end{equation}
  unknown locations in time $q^{O(\sqrt{s})}$, where $d$ is defined as in~(\ref{eq:fqTBdis}) and $e'$ is defined as in~(\ref{eq:listdecradius}).
\end{theorem}

For sufficiently large folding parameter $s$, we can simplify the decoding radius in~(\ref{eq:fqTBdecradius}) to obtain the following bound, which approaches half the distance bound in Corollary~\ref{cor:fqTBdisnice}.
\begin{corollary}
  \label{cor:fqTBdecradiusnice}
  For any $\gamma>0$, let $\tilde{\cC}=\CSS(\tilde{C},\tilde{C})$ be the fqTB code in Definition~\ref{def:fqTBcode} with parameters $q,r,\ell,s$ such that $r$ is a prime for which the uncertainty principle in Proposition~\ref{prop:uncertainty} holds for $r$ over $\bF_q$. Furthermore assume that $s\geq s_0(\gamma,r)$ and $q\geq q_0(s)$, where $s_0(\gamma,r)$ and $q_0(s)$ are some sufficiently large numbers with respect to $\gamma,r$ and $s$ respectively. Then $\cC$ can be decoded from errors in up to
  \begin{equation*}
    e \geq \frac{q-1}{s}\cdot\frac12\left(1-\frac{\ell-1}{q-1}-\sqrt{\frac{1}{r}\left(1-\frac{\ell-1}{q-1}\right)}-\gamma\right),
  \end{equation*}
  unknown locations in time $q^{O(\sqrt{s})}$.
\end{corollary}

Specifically, Corollary~\ref{cor:fqTBdecradiusnice} follows from showing that our current bound on $d/2-1$ in Corollary~\ref{cor:fqTBdisnice} is smaller than $e'$, so our distance bound is the limiting contraint on decoding radius. Thus if we could show a tighter distance bound for fqTB codes than Theorem~\ref{thm:fqTBdis}, we would immediately obtain efficient decoding up to a larger radius. For instance, if we could show that the fqTB codes have relative distance $d/((q-1)/s)\geq 1-\ell/(q-1)-O(1/r)$, it would immediately follow that for all sufficiently large $s,q$, Algorithm~\ref{alg:fqTBdec} decodes fqTB codes for errors on up to $(1-\ell/(q-1)-O(1/r))/2$ fraction of the qudits

\begin{proof}[Proof of Corollary~\ref{cor:fqTBdecradiusnice}]
  Letting $s\rightarrow\infty$ in the bound on the fqTB distance $d$ in Corollary~\ref{cor:fqTBdisnice} and in the expression for $e'$ in~(\ref{eq:listdecradius}), we see that for all sufficiently large $s$ relative to $\gamma,r$ and sufficiently large $q$ relative to $s$, then
  \begin{align*}
    \frac{d}{2}-1 &\geq \frac{q-1}{s}\cdot\frac12\left(1-\frac{\ell-1}{q-1}-\sqrt{\frac{1}{r}\left(1-\frac{\ell-1}{q-1}\right)}-\gamma\right)
  \end{align*}
  and
  \begin{align*}
    e' &\geq \frac{q-1}{s}\cdot\left(1-\frac{1}{2r}-\sqrt{\frac{1}{4r^2}+\frac{r-1}{r}\cdot\frac{\ell-1}{q-1}}-\frac{\gamma}{2}\right).
  \end{align*}
  Thus to show the desired inequality in the corollary statement, because $e=\min\{d/2-1,e'\}$ by Theorem~\ref{thm:fqTBdec}, it suffices to show that
  \begin{align*}
    \frac{q-1}{s}\cdot\frac12\left(1-\frac{\ell-1}{q-1}-\sqrt{\frac{1}{r}\left(1-\frac{\ell-1}{q-1}\right)}-\gamma\right)
    &\leq \frac{q-1}{s}\cdot\left(1-\frac{1}{2r}-\sqrt{\frac{1}{4r^2}+\frac{r-1}{r}\cdot\frac{\ell-1}{q-1}}-\frac{\gamma}{2}\right),
  \end{align*}
  or equivalently, that
  \begin{align*}
    \frac12\left(1-\frac{\ell-1}{q-1}-\sqrt{\frac{1}{r}\left(1-\frac{\ell-1}{q-1}\right)}\right)
    &\leq 1-\frac{1}{2r}-\sqrt{\frac{1}{4r^2}+\frac{r-1}{r}\cdot\frac{\ell-1}{q-1}}.
  \end{align*}
  But the above inequality follows directly from Lemma~\ref{lem:johnsonvsfrs}.
\end{proof}

To show that Algorithm~\ref{alg:fqTBdec} runs in polynomial time, we need the following analogue of Lemma~\ref{lem:distoperp} for the folded case, which shows that the distance computations $\dis(\fevl(g)-\tilde{a},\tilde{B}^\perp)$ in the final line of the algorithm can be performed efficiently.

\begin{lemma}
  \label{lem:distoperpfolded}
  Letting $\tilde{B}^\perp=\fevl(\bF_q[X]^{[q-1]\cap(1+r\bZ)})$, there exists a $O(r\cdot q\poly\log q)$-time algorithm that takes as input $\tilde{b}\in\widetilde{\bF_q^{\bF_q^*}}$, and outputs $\dis(\tilde{b},\tilde{B}^\perp)$.
\end{lemma}
\begin{proof}
  To compute $\dis(\tilde{b},\tilde{B}^\perp)=\min_{\tilde{b}'\in \tilde{B}^\perp}|\tilde{b}-\tilde{b}'|$ efficiently, recall that $B^\perp$ is the space of piecewise linear functions by Lemma~\ref{lem:piecelin}. Therefore we may consider each $\foldi{x}\Omega_r\in\tilde{\bF}_q^*/\Omega_r$ separately, and for each $\foldi{x}\Omega_r$ we must compute the sequence of values $\beta=(\beta_{y\Omega_r})_{y\in\foldi{x}}\in\bF_q^{\foldi{x}}$ that maximizes the number $w(\beta)$ of coset elements $x'\in x\Omega_r$ for which all $y'\in\foldi{x'}$ have $b(y')=\beta_{y'\Omega_r}\cdot y'$; note here that our notation suppresses the implicit dependence of $\beta$ and $w(\beta)$ on the choice of coset $\foldi{x}\Omega_r$. Indeed, by definition we will then have $\dis(\tilde{b},\tilde{B}^\perp)=(q-1)/s-\sum_{\foldi{x}\Omega_r\in\tilde{\bF}_q^*/\Omega_r}w(\beta^{(\foldi{x}\Omega_r)})$, where $\beta^{(\foldi{x}\Omega_r)}$ denotes the optimal choice of $\beta$ for a given $\foldi{x}\Omega_r$.

  To compute the optimal $\beta=\beta^{(\foldi{x}\Omega_r)}$ for a given $\foldi{x}\Omega_r\in\tilde{\bF}_q^*/\Omega_r$, observe that if $w(\beta)>0$, then there must be some $x'\in x\Omega_r$ for which all $y'\in\foldi{x'}$ have $b(y')=\beta_{y'\Omega_r}\cdot y'$. Therefore $\beta_{y'\Omega_r}=b(y')/y'$ for all $y'\in\foldi{x'}$, which completely determines the value of $\beta$. Thus we may simply loop through the $r$ coset elements $x'\in x\Omega_r$, and for each $x'$ compute $w(\beta)$ for $\beta$ given by $\beta_{y'\Omega_r}=b(y')/y'$ for all $y'\in\foldi{x'}$; then whichever of these $r$ values of $\beta$ maximizes $w(\beta)$ is the optimal value $\beta=\beta^{(\foldi{x}\Omega_r)}$. This algorithm by definition computes a given $\beta^{(\foldi{x}\Omega_r)}$ in time $O(r^2s\poly\log q)$ time, so it computes all $\beta^{(\foldi{x}\Omega_r)}$ for $\foldi{x}\Omega_r\in\tilde{\bF}_q^*/\Omega_r$ in $O(r\cdot q\poly\log q)$ time. Thus we can compute $\dis(\tilde{b},\tilde{B}^\perp)$ in $O(r\cdot q\poly\log q)$ time.
\end{proof}

\begin{proof}[Proof of Theorem~\ref{thm:fqTBdec}]
  The proof will follow closely the proof of Theorem~\ref{thm:fqTBdec}, with the main difference our use of the folded RS list decoding algorithm in Theorem~\ref{thm:fRSdec} and the folded qTB distance bound in Theorem~\ref{thm:fqTBdis} in place of their unfolded counterparts.

  By Proposition~\ref{prop:CSSdec}, it suffices to construct an algorithm $\Dec_{\tilde{C}}$ that takes as input a corrupted codeword $\tilde{a}=\tilde{c}+\tilde{b}$ for some $\tilde{c}\in \tilde{C}$ and some corruption $\tilde{b}\in\widetilde{\bF_q^{\bF_q^*}}$ of Hamming weight $|\tilde{b}|\leq e$, and outputs some $\tilde{c}'\in\tilde{C}$ such that $\tilde{c}'-\tilde{c}\in \tilde{C}^\perp$.

  The desired algorithm is given in Algorithm~\ref{alg:fqTBdec}. We first show it correctly decodes as described above, and then analyze the running time.

  Consider a corrupted codeword $\tilde{a}=\tilde{c}+\tilde{b}$ for some codeword $\tilde{c}=\fevl(f)\in\tilde{C}$ and some corruption $\tilde{b}\in\widetilde{\bF_q^{\bF_q^*}}$ of weight $|\tilde{b}|\leq e$. For a given $1\leq i\leq r-1$, by definition $a^{(i)}(x)=\omega_r^{-i}a(\omega_r^ix)-a(x)$ equals $c^{(i)}(x)=f^{(i)}(x)=\omega_r^{-i}f(\omega_r^ix)-f(x)$ at every point $x\in\bF_q^*$ for which $x,\omega_r^ix\notin\supp(b)$. Meanwhile, within a given coset $x\Omega_r\in\bF_q^*/\Omega_r$, the number of ordered pairs of distinct points $y,y'\in x\Omega_r$ such that $y,y'\notin\supp(b)$ is precisely $(r-|b|_{x\Omega_r}|)(r-|b|_{x\Omega_r}-1|)$. Thus the sum over all $i\in\{1,\dots,r-1\}$ of the number of points $\foldi{x}\in\tilde{\bF}_q^*$ where $\tilde{a}^{(i)}(\foldi{x})=\tilde{f}^{(i)}(\foldi{x})$ satisfies
  \begin{align*}
    \hspace{1em}&\hspace{-1em} \sum_{i=1}^{r-1}|\{\foldi{x}\in\tilde{\bF}_q^*:\tilde{a}^{(i)}(\foldi{x})=\tilde{f}^{(i)}(\foldi{x})\}| \\
                &\geq \sum_{i=1}^{r-1}|\{\foldi{x}\in\tilde{\bF}_q^*:\foldi{x},\omega_r^i\foldi{x}\notin\supp(\tilde{b})\}| \\
                &= \sum_{\foldi{x}\Omega_r\in\tilde{\bF}_q^*/\Omega_r}|\{(\foldi{y},\foldi{y'})\in(\foldi{x}\Omega_r)^2:\foldi{y}\neq\foldi{y'}\text{ and }\foldi{y},\foldi{y'}\notin\supp(\tilde{b})\}| \\
                &= \sum_{\foldi{x}\Omega_r\in\tilde{\bF}_q^*/\Omega_r}(r-|\tilde{b}|_{\foldi{x}\Omega_r}|)(r-|\tilde{b}|_{\foldi{x}\Omega_r}-1|) \\
                &\geq \frac{q-1}{rs}\bE_{\foldi{x}\Omega_r\sim\Unif(\tilde{\bF}_q^*/\Omega_r)}[(r-|\tilde{b}|_{\foldi{x}\Omega_r}|)^2]-\left(\frac{q-1}{s}-|\tilde{b}|\right) \\
                &\geq \frac{q-1}{rs}\bE_{\foldi{x}\Omega_r\sim\Unif(\tilde{\bF}_q^*/\Omega_r)}[(r-|\tilde{b}|_{\foldi{x}\Omega_r}|)]^2-\left(\frac{q-1}{s}-|\tilde{b}|\right) \\
                &= \frac{rs}{q-1}\left(\frac{q-1}{s}-|\tilde{b}|\right)^2-\left(\frac{q-1}{s}-|\tilde{b}|\right) \\
                &\geq \frac{rs}{q-1}\left(\frac{q-1}{s}-e\right)^2-\left(\frac{q-1}{s}-e\right),
  \end{align*}
  where the final inequality above holds becuase the function $\beta\mapsto rs((q-1)/s-\beta)^2/(q-1)-((q-1)/s-\beta)$ is decreasing for all $\beta\leq e\leq(q-1)/2s$. Then averaging over all $i\in\{1,\dots,r-1\}$, there must be some such $i$ for which
  \begin{equation}
    \label{eq:agreeffolded}
    |\{\foldi{x}\in\tilde{\bF}_q^*:\tilde{a}^{(i)}(\foldi{x})=\tilde{f}^{(i)}(\foldi{x})\}| \geq \frac{rs}{(r-1)q-1}\left(\frac{q-1}{s}-e\right)^2-\frac{1}{r-1}\left(\frac{q-1}{s}-e\right).
  \end{equation}

  Recall by Lemma~\ref{lem:qTBdef}, we may decompose $f(X)=g(X)+h(X)$ for $g(X)=\sum_jg_jX^j\in\bF_q[X]^{[\ell]\setminus(\pm 1+r\bZ)}$ and $h(X)\in\bF_q[X]^{[q-1]\cap(1+r\bZ)}$. Define $g^{(i)},h^{(i)}$ analogouosly to $f^{(i)}$, so that $f^{(i)}=g^{(i)}+h^{(i)}$. Then because the expression in~(\ref{eq:coeffchange}) equals $0$ for all $j\equiv 1\pmod{r}$, it follows that $h^{(i)}=0$ and thus $f^{(i)}=g^{(i)}$. Therefore~(\ref{eq:agreeffolded}) is equivalent to
  \begin{equation}
    \label{eq:agreegfolded}
    |\{\foldi{x}\in\tilde{\bF}_q^*:\tilde{a}^{(i)}(\foldi{x})=\tilde{g}^{(i)}(\foldi{x})\}| \geq \frac{rs}{(r-1)q-1}\left(\frac{q-1}{s}-e\right)^2-\frac{1}{r-1}\left(\frac{q-1}{s}-e\right).
  \end{equation}
  But by~(\ref{eq:coeffchange}), the coefficients of $g^{(i)}(X)=\sum_{j\in[\ell]\setminus(\pm 1+r\bZ)}g_j^{(i)}X^j$ are given by $g_j^{(i)}=(\omega_r^{i(j-1)\pmod{r}}-1)g_j$, and thus $g$ and $g^{(i)}$ have coefficients of the same support, that is, $g_j=0$ if and only if $g_j^{(i)}=0$. In particular, it follows that $\deg(g^{(i)})=\deg(g)<\ell$. In other words, $\evl(g^{(i)})\in\mathrm{fRS}(q,\ell)$ is a folded Reed-Solomon codeword, so~(\ref{eq:agreegfolded}) says that $\tilde{a}^{(i)}=\fevl(g^{(i)})+\tilde{b}^{(i)}$ is a corrupted Reed-Solomon codeword with the corruption $\tilde{b}^{(i)}$ of weight
  \begin{align*}
    |\tilde{b}^{(i)}|
    &= \frac{q-1}{s} - |\{\foldi{x}\in\tilde{\bF}_q^*:\tilde{a}^{(i)}(\foldi{x})=\tilde{g}^{(i)}(\foldi{x})\}| \\
    &\leq \frac{q-1}{s} - \left(\frac{rs}{(r-1)q-1}\left(\frac{q-1}{s}-e\right)^2-\frac{1}{r-1}\left(\frac{q-1}{s}-e\right)\right) \\
    &= \frac{q-1}{s}\left(1-\left(\frac{r}{r-1}\left(1-\frac{es}{q-1}\right)^2-\frac{1}{r-1}\left(1-\frac{es}{q-1}\right)\right)\right).
  \end{align*}
  By Theorem~\ref{thm:fRSdec}, the output of running the list decoder $\ListDec_{\mathrm{fRS}(q,\ell,s)}(\tilde{a}^{(i)})$ is a list containing $\fevl(\tilde{g}^{(i)})$ as long as the right hand side above is at most
  \begin{equation*}
    \frac{q-1}{s}\left(1-\left(1+\frac{2}{\sqrt{s}}\right)\left(\frac{\ell}{q-1}\right)^{1-1/\sqrt{s}}\right)-2,
  \end{equation*}
  which simplifies to needing that $e\leq e'$ for
  \begin{equation}
    \label{eq:listdecradius}
    e' = \frac{q-1}{s}\left(1-\frac{1}{2r}-\sqrt{\frac{1}{4r^2}+\frac{r-1}{r}\left(\left(1+\frac{2}{\sqrt{s}}\right)\left(\frac{\ell}{q-1}\right)^{1-1/\sqrt{s}}+\frac{2s}{q-1}\right)}\right).
  \end{equation}
  But the above holds by the definition of $e$, so the list $\cL_i$ in Algorithm~\ref{alg:fqTBdec} will contain $g^{(i)}$, and thus after the $i$th iteration of the for loop, the list $\cL$ will contain $\fevl(g(X))=\fevl(\sum_j(\omega_r^{(j-1)i}-1)^{-1}g_j^{(i)}X^j)$. Thus if Algorithm~\ref{alg:fqTBdec} outputs $\fevl(g')\in\cL$, then $g'\in\bF_q[X]^{[\ell]\setminus(\pm 1+r\bZ)}$ and
  \begin{equation*}
    \dis(\fevl(g')-\tilde{a},\tilde{B}^\perp) \leq \dis(\fevl(g)-\tilde{a},\tilde{B}^\perp) \leq |\fevl(g+h)-\tilde{a}| = |\fevl(f)-\tilde{a}| = |\tilde{b}| \leq e.
  \end{equation*}
  But as $\tilde{B}^\perp\subseteq\tilde{C}^\perp$, it follows that $\dis(\fevl(g)-\tilde{a},\tilde{C}^\perp) \leq e$ and $\dis(\fevl(g')-\tilde{a},\tilde{C}^\perp) \leq e$, so $\dis(\fevl(g)-\fevl(g'),\tilde{C}^\perp)\leq 2e$. But by the definition of $e$ along with Theorem~\ref{thm:fqTBdis}, $\cC$ has distance $\min_{c\in C\setminus C^\perp}|c|>2e$, so it follows that $\fevl(g)-\fevl(g')\in\tilde{C}^\perp$. Thus $\Dec_{\tilde{C}}(\tilde{a})$ outputs some $\fevl(g')\in\fevl(g)+\tilde{C}^\perp$, as desired.

  It remains to show that Algorithm~\ref{alg:fqTBdec} runs in time $q^{O(\sqrt{s})}$. Theorem~\ref{thm:fRSdec} implies that the calls to $\ListDec_{\mathrm{fRS}(q,\ell,s)}$ run in $q^{O(\sqrt{s})}$ time, while Lemma~\ref{lem:distoperpfolded} implies that the $\argmin$ computation in the final line of the algorithm runs in $q^{O(1)}$ time. The rest of the algorithm by definition also runs in $q^{O(1)}$ time, so the result follows.
\end{proof}

\section{Impossibility of Quantum Locally Correctable Codes}
In this section, we show that quantum codes are inherently unable to perform local recovery from a large number of erasures. Thus the local decoding capabilities of qLRCs is in some sense close to optimal for quantum codes. This result is in constrast to the classical setting, where there exist locally correctable codes (LCCs), which can recover every given code symbol from a constant number $r$ of other code symbols, even after a linear number of code symobls have been erased.

Our impossibility result is stated below. Informally, it states that any qudit in a quantum code that can be recovered from two disjoint sets of other qudits must be useless for error correction; that is, such a qudit is entirely unentangled from the remainder of the code state, and contains no information about the encoded message.

\begin{theorem}
  \label{thm:notworec}
  Let $\cC$ be a quantum code of block length $n$ and dimension $k>0$. Assume that for some $i\in[n]$, there exist two subsets $I_i^1,I_i^2\subseteq[n]$ satisfying $I_i^1\cap I_i^2=\{i\}$ such that for each $b=1,2$, there is an associated recovery channel $\Rec_i^b:\cM(I_i^b\setminus\{i\})\rightarrow\cM(I_i^b)$ with the guarantee that for every code state $\psi\in\cC$,
  \begin{equation*}
    \Rec_i^b\otimes I_{[n]\setminus I_i^b}(\psi_{[n]\setminus\{i\}}) = \psi.
  \end{equation*}
  Then there exists a 1-qudit density matrix $\alpha\in\cM(1)$ such that every $\psi\in\cC$ can be decomposed as $\psi=\alpha_i\otimes\psi_{[n]\setminus\{i\}}$.
\end{theorem}

In comparison, qLRCs require each qudit $i\in[n]$ to be recoverable from just a single set $I_i\setminus\{i\}$ of other qudits. Thus Theorem~\ref{thm:notworec} shows that it is impossible to extend a qLRC to even just have a second disjoint recovery set for each qudit.

This result is in contrast to the classical case, where code components can be recovered from many different disjoint subsets of components. Indeed, classically there exist locally correctable code (LCCs), which have the property that even after erasing any constant fraction of the code components, each erased component can be recovered from some consant number of unerased components.

Theorem~\ref{thm:notworec} shows that no such code LCC can exist quantumly. Specifically, assume a quantum code $\cC$ is such that each code qudit $i$ can be recovered from qudits $I_i\setminus\{i\}$. Then Theorem~\ref{thm:notworec} shows that for every $i\in[n]$, if the qudits in $I_i$ are erased, then it is impossible to locally recover qudit $i$ without performing a more global decoding operation that also recovers other qudits in $I_i$.

\begin{proof}[Proof of Theorem~\ref{thm:notworec}]
  Define two additional qudits $i_1=n$ and $i_2=n+1$. For $b=1,2$, define $J^b=I_i^b\cup\{i_1\}\setminus\{i\}$. Now for any given pure code state $\psi=\ket{\psi}\bra{\psi}\in\cC$, let $\rho\in\cM([n+2])$ be the state obtained by applying $\Rec_i^1$ but placing the recovered copy of qudit $i$ in register $i_1$, and also applying $\Rec_i^2$ but placing the recovered copy of qudit $i$ in register $i_2$. That is,
  \begin{align*}
    \rho &= (\Rec_i^1)_{J^1}(\Rec_i^2)_{J^2}(\psi).
  \end{align*}
  Note that by definition $J^1$, $J^2$, and $\{i\}$ are all disjoint, so in particular $(\Rec_i^1)_{J^1}$ and $(\Rec_i^2)_{J^2}$ act on disjoint sets of qudits and therefore commute. Letting $\psi_{[n]\setminus\{i\}\cup\{i_b\}}$ denote the state $\psi$ with qudit $i$ moved to position $i_b$, it follows that
  \begin{align}
    \label{eq:traceonerec}
    \begin{split}
      \Tr_{i_1}(\Tr_i\rho)
      &= (\Rec_i^2)_{J^2}\Tr_{i_1}(\Rec_i^1)_{J^1}(\psi_{[n]\setminus\{i\}}) \\
      &= (\Rec_i^2)_{J^2}\Tr_{i_1}(\psi_{[n]\setminus\{i\}\cup\{i_b\}}) \\
      &= (\Rec_i^2)_{J^2}(\psi_{[n]\setminus\{i\}}) \\
      &= \psi_{[n]\setminus\{i\}\cup\{i_2\}}.
    \end{split}
  \end{align}
  Therefore we have shown that tracing over $i_1$ in $\Tr_i\rho$ yields the pure state $\psi_{[n]\setminus\{i\}\cup\{i_2\}}\cong\psi$, and thus there is a tensor decomposition
  \begin{equation*}
    \Tr_i\rho=\alpha_{i_1}\otimes\psi_{[n]\setminus\{i\}\cup\{i_2\}}
  \end{equation*}
  for some 1-qudit density matrix $\alpha\in\cM(1)$.

  But the same reasoning used to show~(\ref{eq:traceonerec}) also implies that $\Tr_{i_2}(\Tr_i\rho) = \psi_{[n]\setminus\{i\}\cup\{i_1\}}$, so
  \begin{align*}
    \psi_{[n]\setminus\{i\}\cup\{i_1\}}
    &= \Tr_{i_2}(\Tr_i\rho) \\
    &= \Tr_{i_2}(\alpha_{i_1}\otimes\psi_{[n]\setminus\{i\}\cup\{i_2\}}) \\
    &= \alpha_{i_1}\otimes\psi_{[n]\setminus\{i\}},
  \end{align*}
  or equivalently,
  \begin{equation*}
    \psi = \alpha_i\otimes\psi_{[n]\setminus\{i\}}.
  \end{equation*}

  It remains to be shown that $\alpha_i$ is the same for all $\psi\in\cC$. This conclusion follows Lemma~\ref{lem:localind}, which states that if a quantum code can recover from erasures on a given set of qudits (here $\Rec_i^b$ recovers from erasures on qudit $i$), then the reduced density matrix of a code state on those qudits contains no information about the encoded message state.
\end{proof}



\section{Acknowledgments}
We thank Thiago Bergamaschi for helpful discussions.

\bibliographystyle{alpha}
\bibliography{library}

\appendix

\section{Technical Lemmas}
\label{app:teclem}
Below we prove the technical lemmas used in the Section~\ref{sec:fqTBdis} for bounding the distance of the fqTB codes.

The following proof of Lemma~\ref{lem:detrank} is well known, and is included for completeness.

\begin{proof}[Proof of Lemma~\ref{lem:detrank}]
  To show that $\det$ has a root of multiplicity $m-t$ at $x=(x_{ij})_{i,j\in[m]}$ of rank $t$, it suffices to show that all $d$th derivatives of $\det$ vanish at $x$ for all $d\in[m-t]$. But by definition every $d$th derivative of $\det(X_{ij})_{i,j\in[m]}$ either vanishes, or is equal (up to a sign) of the derivative of a $(m-d)\times(m-d)$ submatrix of $(X_{ij})_{i,j\in[m]}$. For $X_{ij}=x_{ij}$ and $d\in[m-t]$, every such submatrix has rank at most $\rank(x_{ij})_{i,j\in[m]}=t<m-d$, so every such submatrix is not full rank and therfore has determinant $0$. Thus every $d$th derivative of $\det(X_{ij})_{i,j\in[m]}$ for $d\in[m-t]$ vanishes at $X_{ij}=x_{ij}$, and thus $\det$ has a root of multiplicity $m-t$ at $x$.
\end{proof}

The following lemma bounds the term $\epsilon$ defined by~(\ref{eq:fqTBloss}) in Theorem~\ref{thm:fqTBdis}.

\begin{lemma}
  \label{lem:boundloss}
  Define $\epsilon$ as in~(\ref{eq:fqTBloss}). If we set $s=cr^2$ for $c\geq 2$ and let $\lambda=1-(\ell-1)/(q-1)$, then
  \begin{equation}
    \label{eq:boundloss}
    \epsilon
    \leq \left(1+\frac{1}{c}\right)\sqrt{\frac{\lambda}{r}}.
  \end{equation}
\end{lemma}
\begin{proof}
  By definition
  \begin{equation*}
    \epsilon \leq \max_{1\leq m\leq r}\lambda\cdot\frac{m-1}{r} < \lambda.
  \end{equation*}
  Thus if $\lambda\leq 1/r$, then
  \begin{equation*}
    \epsilon \leq \lambda = \sqrt{\lambda^2} \leq \sqrt{\frac{\lambda}{r}},
  \end{equation*}
  from which~(\ref{eq:boundloss}) follows. 

  Therefore assume that $\lambda\geq 1/r$. By definition
  \begin{equation*}
    \epsilon \leq \max_{1\leq m\leq r}\min\left\{\frac{\lambda}{r}\cdot m,\;\frac{1}{m}+\frac{1}{s}\cdot m\right\}.
  \end{equation*}
  Because $s=cr^2\geq r^2$ by assumption, the expression $1/m+m/s$ is decreasing in $m$ for $1\leq m\leq r$, while the expression $\lambda m/r$ is increasing in $m$. Thus because $m^*=1/\sqrt{\lambda/r-1/s}$ is the unique positive real number for which $\lambda m^*/r=1/m^*+m^*/s$, it follows that
  \begin{align*}
    \epsilon
    &\leq \frac{\lambda}{r}\cdot m^* = \frac{1}{m^*}+\frac{1}{s}\cdot m^* \\
    &= \sqrt{\frac{\lambda}{r}} \cdot \frac{1}{\sqrt{1-\frac{r}{\lambda s}}} \\
    &\leq \sqrt{\frac{\lambda}{r}} \cdot \frac{1}{\sqrt{1-\frac{1}{c}}} \\
    &\leq \left(1+\frac{1}{c}\right)\sqrt{\frac{\lambda}{r}},
  \end{align*}
  where the second inequality above holds because $s=cr^2$ and $\lambda\geq 1/r$, and the final inequality holds because $1/\sqrt{1-x}\leq 1+x$ for all $x\leq 1/2$. Thus~(\ref{eq:boundloss}) holds.
\end{proof}

The following bound is used in the proof of Theorem~\ref{thm:qTBdec}.

\begin{lemma}
  \label{lem:twojohnsons}
  For all real numbers $0\leq x\leq 1/6$ and $1/2\leq y\leq 1$, it holds that
  \begin{equation*}
    \frac12\left(1-x-\sqrt{x^2+(1-2x)y}\right) \leq 1-x-\sqrt{x^2+(1-2x)\sqrt{y}}.
  \end{equation*}
\end{lemma}
\begin{proof}
  Let $f(x,y)$ denote the LHS minus the RHS of the inequality in the lemma statement. Letting $D=[0,1/6]\times[1/2,1]$, then our goal is to show that $f(x,y)\leq 0$ for all $(x,y)\in D$. By definition $f(x,1)=0$ for all $0\leq x\leq 1/6$, so it is sufficient to show that $\pdv{f}{y}\geq 0$ for all $(x,y)\in D$. Now by definition
  \begin{equation*}
    \pdv{f}{y} = -\frac{(1-2x)/2}{2\sqrt{x^2+(1-2x)y}}+\frac{(1-2x)/(2\sqrt{y})}{2\sqrt{x^2+(1-2x)\sqrt{y}}},
  \end{equation*}
  so setting the above $\geq 0$ and rearranging gives that we must show
  \begin{equation*}
    (1-2x-x^2)y-(1-2x)\sqrt{y}+x^2 \leq 0
  \end{equation*}
  But solving the quadratic equation on the LHS above gives roots at $\sqrt{y}=1$ and $\sqrt{y}=x^2/(1-2x-x^2)\leq 1/23$ for $0\leq x\leq 1/6$, so indeed the above inequality holds for all $(x,y)\in D$, and thus all $(x,y)\in D$ have $\pdv{f}{y}\geq 0$ and therefore $f(x,y)\leq 0$, as desired.
\end{proof}

The following bound is used in the proof of Theorem~\ref{thm:fqTBdec}.

\begin{lemma}
  \label{lem:johnsonvsfrs}
  For all real numbers $0\leq x\leq 1/6$ and $0\leq y\leq 1/2$, it holds that
  \begin{equation*}
    \frac12(y-\sqrt{2xy}) \leq 1-x-\sqrt{x^2+(1-2x)(1-y)}.
  \end{equation*}
\end{lemma}
\begin{proof}
  Let $f(x,y)$ denote the LHS minus the RHS of the inequality in the lemma statement. Letting $D[0,1/6]\times[1/2,1]$, then our goal is to show that $f(x,y)\leq 0$ for all $(x,y)\in D$. Computing the partial derivatives
  \begin{align*}
    \pdv{f}{x}(x,y) &= -\frac{1}{2\sqrt{2}}\cdot\sqrt{\frac{y}{x}}+1-\frac{1-x-y}{\sqrt{x^2+(1-2x)(1-y)}} \\
    \pdv{f}{y}(x,y) &= \frac12\left(1-\sqrt{\frac{x}{2y}}\right)-\frac{1-2x}{2\sqrt{x^2+(1-2x)(1-y)}} \\
    \pdv{f}{x}{y}(x,y) &= -\frac{1}{4\sqrt{2}}\cdot\frac{1}{\sqrt{xy}} - \frac{(1-x-y)(1-2x)}{2(x^2+(1-2x)(1-y))^{3/2}} + \frac{1}{\sqrt{x^2+(1-2x)(1-y)}},
  \end{align*}
  it follows that the statements below hold for all $(x,y)\in D$:
  \begin{enumerate}
  \item\label{it:fx0} $f(x,0)=0$
  \item\label{it:fy0} $f(0,y)=y/2+\sqrt{1-y}-1\leq 0$
  \item\label{it:pfx0} ${\pdv{f}{x}}(x,0)=0$
  \item\label{it:pfxy} ${\pdv{}{y}}{\pdv{f}{x}}(x,y)={\pdv{f}{x}{y}}(x,y)\leq 0$ if $x\leq 1/50$ or $y\leq 1/20$
  \item\label{it:pfxl} ${\pdv{f}{x}}(x,y)\leq 1$
  \item\label{it:pfyl} ${\pdv{f}{y}}(x,y)\leq 1/2$
  \end{enumerate}
  Combining item~\ref{it:pfx0} and item~\ref{it:pfxy} above implies that ${\pdv{f}{x}}(x,y)\leq 0$ for all $(x,y)\in D$ such that either $x\leq 1/50$ or $y\leq 1/20$. Combining this fact with item~\ref{it:fy0} above implies that $f(x,y)\leq 0$ for all $(x,y)\in D$ such that either $x\leq 1/50$ or $y\leq 1/20$. 

  Therefore if we find some sufficiently large integer $n\geq 100$ such that $f(x,y)\leq -2/n$ for all $(x,y)\in([1/100,1/6]\cap\bZ/n)\times([1/40,1/2]\cap\bZ/n)$, then item~\ref{it:pfxl} and item~\ref{it:pfyl} above imply that $f(x,y)\leq 0$ for all $(x,y)\in[1/50,1/6]\times[1/20,1/2]$, which combined with our conclusion above that $f(x,y)\leq 0$ when $x\leq 1/50$ or $y\leq 1/20$, implies the desired inequality $f(x,y)\leq 0$ for all $(x,y)\in D$. But we may numerically verify that $n=1000$ satisfies the desired property, which completes the proof. 
\end{proof}

\section{Omitted Proofs}
\label{app:omitproofs}
This section provides proofs that were omitted in the main text.

\begin{proof}[Proof of Proposition~\ref{prop:ael}]
  The desired decoding algorithm for $\cC$ simply applies $\pi_G^{-1}$ to unpermute the symbols, then applies the decoder for $\cC_{\text{in}}$ to each of the $n_{\text{out}}$ inner code blocks, and finally applies the decoder for $\cC_{\text{out}}$ to the resulting state. By definition this decoder runs in $\poly(n\log q)$ if the outer and inner codes' decoding algorithms both run in $\poly(n\log q)$ time. Thus we just need to verify the correctness of this decoding algorithm.

  Assume that the corruptions occur on some set $T\subseteq[n]$ of $|T|\leq\alpha n$ components of $\cC$. Also assume for a contradiction that our decoder fo $\cC$ fails to recover the original message, which can only occur if after applying $\pi_G^{-1}$, at least $\alpha_{\text{out}}n_{\text{out}}$ of the inner code blocks have $\geq\alpha_{\text{in}}n_{\text{in}}$ corruptions, so that all of these inner code blocks are ``overloaded'' and their inner decodings fail to recover the correct value of the outer code component. Let $S\subseteq[n]\cong[n_{\text{out}}]\times[n_{\text{in}}/\Delta]$ denote the set of all the folded components in in these overloaded blocks. Then because $\geq\alpha_{\text{out}}n_{\text{out}}$ inner code blocks are overloaded, we have
  \begin{equation}
    \label{eq:Slowbound}
    |S| \geq \alpha_{\text{out}}n_{\text{out}}\cdot\frac{n_{\text{in}}}{\Delta} = \alpha_{\text{out}}n
  \end{equation}
  while because each of the $|S|\Delta/n_{\text{in}}$ overloaded inner code block has $\geq\alpha_{\text{in}}n_{\text{in}}$ of its $n_{\text{in}}$ symbols corrupted, we have
  \begin{equation*}
    |E(S,T)| \geq \frac{|S|\Delta}{n_{\text{in}}}\cdot\alpha_{\text{in}}n_{\text{in}} = \alpha_{\text{in}}\Delta|S|,
  \end{equation*}
  as each corrupted inner code symbol in an overloaded inner code block corresponds to an edge from $S$ to $T$. But the expander mixing lemma implies that
  \begin{align*}
    |E(S,T)|
    &\leq \frac{\Delta|S||T|}{n}+\lambda\Delta\sqrt{|S||T|} \\
    &\leq \alpha\Delta|S|+\lambda\Delta\sqrt{|S|\cdot\alpha n} \\
    &< \left(\alpha_{\text{in}}-\lambda\sqrt{\frac{\alpha_{\text{in}}}{\alpha_{\text{out}}}}\right)\Delta|S|+\lambda\Delta\sqrt{|S|\cdot\alpha_{\text{in}}n},
  \end{align*}
  where we have applied the definition of $\alpha=\alpha_{\text{in}}-\lambda\sqrt{\alpha_{\text{in}}/\alpha_{\text{out}}}$. Combining the two above inequalities gives that
  \begin{equation*}
    \alpha_{\text{in}}\Delta|S| < \left(\alpha_{\text{in}}-\lambda\sqrt{\frac{\alpha_{\text{in}}}{\alpha_{\text{out}}}}\right)\Delta|S|+\lambda\Delta\sqrt{|S|\cdot\alpha_{\text{in}}n},
  \end{equation*}
  which simplifies to
  \begin{equation*}
    |S| < \sqrt{\alpha_{\text{out}}n},
  \end{equation*}
  contradicting~(\ref{eq:Slowbound}). Thus the assumption that the decoder fails was false, so the decoder must succeed in correcting errors on any set $|T|$ of $|T|\leq\alpha n$ components of $\cC$.
\end{proof}

\end{document}